\documentclass[runningheads,envcountsame]{llncs}
\usepackage[T1]{fontenc}
\usepackage{graphicx}
\usepackage[shortlabels]{enumitem}

\usepackage{amsmath,amssymb}
\usepackage[hyperref]{ntheorem}
\usepackage{todonotes}
\usepackage{float}
\usepackage{tkz-graph}
\usepackage[colorlinks,hypertexnames=false]{hyperref}
\usepackage{cleveref}
\usepackage{url}
\usepackage{booktabs}

\newtheorem{fact}[theorem]{Fact}
\newtheorem{observation}[theorem]{Observation}

\spnewtheorem{mclaim}{Claim}[theorem]{\bfseries}{\itshape}
\Crefname{mclaim}{Claim}{Claims}

\newcommand{\sharedVertex}[1]{C^{\mathrm{V}}_{#1}}
\newcommand{\sharedEdge}[1]{C^{\mathrm{E}}_{#1}}
\newcommand{\forbiddenGraph}{F}
\newcommand{\forbiddenClass}{\mathcal{F}}
\newcommand{\PPP}{\mathcal{P}}
\newcommand{\Hgraph}[1]{H_{#1}}
\newcommand{\subdivStar}[1]{S_{#1}}
\newcommand{\join}{\bowtie}

\newcommand{\td}{{\textup{td}}}
\newcommand{\tw}{{\textup{tw}}}
\newcommand{\cw}{{\textup{cw}}}
\newcommand{\pw}{{\textup{pw}}}
\newcommand{\SM}{{\;{|}\;}} 
\newcommand{\SE}{\,\}}
\newcommand{\SB}{\{\,}

\newcommand{\CCC}{\mathcal{C}}

\newenvironment{claimproof}
{\noindent {\em Proof of Claim:} }
{\hfill $\diamond$ \smallskip}

\title{Bounding Width on Graph Classes of Constant~Diameter}
\titlerunning{Bounding Width on Graph Classes of Constant~Diameter}

\author{Konrad K. Dabrowski
\and
  Tala Eagling-Vose
  \and
  Noleen K\"ohler
  \and
  Sebastian Ordyniak
  \and
  Dani\"el Paulusma
 }

\authorrunning{K.K.\ Dabrowski, T.\ Eagling-Vose, N. K\"{o}hler, S.\
  Ordyniak, D. Paulusma}

\institute{School of Computing, Newcastle University, Newcastle, UK\\\email{konrad.dabrowski@newcastle.ac.uk}\\
  Department of Computer Science, Durham University, Durham, UK\\\email{\{tala.j.eagling-vose,daniel.paulusma\}@durham.ac.uk}\\
  School of Computer Science, University of Leeds, Leeds, UK \\\email{n.koehler@leeds.ac.uk,sordyniak@gmail.com}}

\begin{document}

\maketitle

\begin{abstract}
We determine if the width of a graph class ${\cal G}$ changes from unbounded to bounded if we consider only those graphs from ${\cal G}$ whose diameter is bounded. As parameters we consider treedepth, pathwidth, treewidth and clique-width, and as graph classes we consider classes defined by forbidding some specific graph $F$ as a minor, induced subgraph or subgraph, respectively. Our main focus is on treedepth for $F$-subgraph-free graphs of diameter at most~$d$ for some fixed integer $d$. We give classifications of boundedness of treedepth for $d\in \{4,5,\ldots\}$ and partial classifications for $d=2$ and $d=3$.
\end{abstract}

\section{Introduction}
Graph width parameters play a prominent role in modern graph theory. One of the reasons is that large sets of NP-complete graph problems may become polynomial-time solvable on graph classes on which some width parameter is bounded (by a constant). For example, the celebrated meta-theorem of Courcelle~\cite{Co90} states that every problem definable in MSO$_2$ is polynomial-time solvable on graph classes of bounded treewidth. Another well-known meta-theorem, due to Courcelle, Makowsky and Rotics~\cite{CMR00}, states that every problem definable in MSO$_1$ is polynomial-time solvable on graph classes of bounded clique-width. The logic MSO$_1$ is more restrictive than MSO$_2$. However, any graph class of bounded treewidth has bounded clique-width, whereas the reverse statement does not hold. That is, clique-width is more {\it powerful} than treewidth. 

Due to the above algorithmic implications and also out of a graph-structural interest, there exist many papers in the literature that research whether certain graph classes have bounded width. The framework of graph containment opens the way for a more systematic approach. For example, the recent treewidth dichotomy of Lozin and Razgon~\cite{LR22} determines exactly for which finite sets~${\cal F}$, the class of ${\cal F}$-free graphs has bounded treewidth; here, {\it ${\cal F}$-free} means not containing any graph from ${\cal F}$ as an {\it induced} subgraph. Hickingbotham~\cite{Hi23} showed that in the treewidth dichotomy we may replace ``treewidth'' by ``pathwidth''. In contrast, there still exist some pairs $(H_1,H_2)$ for which boundedness of clique-width is open for $(H_1,H_2)$-free graphs; see the survey~\cite{DJP19} for details.

\medskip
\noindent
{\bf Our Focus.} 
A width parameter~$p$ that is unbounded on a graph class ${\cal G}$ may be bounded on a subclass ${\cal G'}$ of ${\cal G}$. Ideally, we would like ${\cal G}'$ to be as large as possible to optimally benefit from the algorithmic benefits if $p$ is bounded. We consider the {\it diameter} of the input graph and ask:

\medskip
\noindent
{\it For a graph class ${\cal G}$ of unbounded width, what is the largest $d$ such that the graphs in ${\cal G}$ of diameter at most~$d$ have bounded width?}

\medskip
\noindent
This is a natural question, as there exist numerous NP-complete problems that stay NP-complete even on graphs of diameter~$d=2$; see, e.g. the diameter study for $k$-{\sc Colouring}~\cite{DPR22,MS16}, in particular for ${\cal F}$-free graphs in~\cite{DPR22,KS23,MPS22,MS16}. There are even problems that are NP-complete {\it only if} $d=2$ (e.g. {\sc Disconnected Cut}~\cite{MP15}). Answering the above question will have a wide range of algorithmic consequences, in particular due to meta-theorems, as we discussed 
above~\cite{CMR00,CO00}.

\medskip
\noindent
{\bf Our Approach.} We work within the framework of graph containment and thus focus on graph classes defined by some set ${\cal F}$ of forbidden graphs. To get a handle on these, we restrict ourselves to the case when ${\cal F}$ consists of a single graph~$F$.  
To answer our research question, we selected some classical graph containment relations and width parameters. We forbid $F$ as an induced subgraph, subgraph or minor. We say that a graph $G$ is {\it $F$-free}, {\it $F$-subgraph-free} or {\it $F$-minor-free} if it does not contain $F$ as an induced subgraph, subgraph or minor, respectively. Note that $F$-minor-free graphs are $F$-subgraph-free and $F$-subgraph-free graphs are $F$-free. As width parameters, we will consider pathwidth $\pw$, treewidth $\tw$, clique-width $\cw$ and also treedepth $\td$
(which has algorithmic applications for many problems where treewidth cannot be used; see 
e.g.~\cite{GajarskyHlineny12,GanianOrdyniak18,GJW16,IwataOgasawaraOhsaka18,KLO18}). 
We write $p \rhd q$ if $p$ is more powerful than $q$. It is well known~\cite{BGHK95,CO00} that $\cw \rhd \tw \rhd \pw \rhd \td$. 
 
\medskip
\noindent
{\bf Known Results}.
 We first describe, in Table~\ref{t-table},
 the situation without a diameter bound. For $r\geq 1$, the
 graph $P_r$ is the $r$-vertex path. The set ${\mathcal S}$ consists
 of all graphs, every component of which is a path or 
 {\it subdivided claw} (cubic tree with exactly one vertex of
 degree~$3$). The set $\overline{{\mathcal S}}$ consists of all graphs that are  subgraphs of any {\it subdivided star}  (any tree with exactly one vertex of degree at least~$3$).
 
 Table~\ref{t-table} also includes known results on diameter-width. Eppstein~\cite{Ep00} defined a graph class~${\cal G}$ to have the {\it diameter-treewidth} property if the treewidth of every graph in ${\cal G}$ is bounded by a function of the diameter of $G$. For graph classes closed under taking subgraphs, this notion coincides with bounded local treewidth, a crucial notion in bidimensionality theory. We define the properties of diameter-clique-width, diameter-pathwidth and diameter-treedepth analogously. A graph~$G$ is {\it apex ${\cal G}$} for a graph class ${\cal G}$ if $G-v\in {\cal G}$ for some vertex $v\in V(G)$. So, for instance, the class of {\it apex linear forests} consists of all graphs that become {\it linear forests} (disjoint unions of paths) after removing at most one vertex.
 
  The $(k \times k)$-wall, is the $k$-by-$k$ hexagonal grid.
  We first note that by adding a dominating vertex to a wall we obtain a graph of diameter $d=2$ whose clique-width can be arbitrarily large. This graph only contains apex planar graphs as minors. Hence, 
  the result of Eppstein~\cite{Ep00} implies that
  a class of $F$-minor-free graphs of diameter~$2$ has bounded clique-width if and only if $F$ is apex planar. As this is not true for $d=1$ (just take $F=K_6$), we say that $d=2$ is {\it tight} for diameter-clique-width for minors.  
  By adding a dominating vertex to a 
  full binary tree, we find that $d=2$ is also tight for
 diameter-pathwidth for minors.

\begin{table}[t]
\begin{center}
 \aboverulesep=0ex
 \belowrulesep=0ex
\begin{tabular}{l|l|l|l}
\toprule
& \hspace*{17.5mm}minor & induced subgraph$\;\;\;$ & \hspace*{6.5mm} subgraph \\\midrule
$\cw$ &   planar~\cite{DP16} \hspace*{7mm} (\textcolor{blue}{\bf apex planar})            &  $\subseteq_i P_4$~\cite{CO00} (\textcolor{blue}{\bf $\mathbf{\subseteq_i P_4}$}) & $\mathcal{S}$~\cite{BL02} \hspace*{15.5mm}
(?)\\
$\tw$ &       planar~\cite{RS86} \hspace*{8mm}(apex planar~\cite{Ep00})           &  $\subseteq_i P_2$ \hspace*{3.5mm} (\textcolor{blue}{\bf $\mathbf{\subseteq_i P_2}$}) & $\mathcal{S}$~\cite{RS84} 
\hspace*{14mm}
(?)\\
$\pw$ &       forest~\cite{BRST91} 
 \hspace*{11mm}(apex forest~\cite{DEJMW20})                   &  $\subseteq_i P_2$  
 \hspace*{3.5mm}
 (\textcolor{blue}{\bf $\mathbf{\subseteq_i P_2}$}) & $\mathcal{S}$~\cite{RS84}
 \hspace*{14mm}
 (?)\\
$\td$ &  linear forest~\cite{NM12} (\textcolor{blue}{\bf apex linear forest})    
& $\subseteq_i P_2$ \hspace*{3.5mm} (\textcolor{blue}{\bf $\mathbf{\subseteq_i P_2}$}) &  linear forest~\cite{NM12} (\textcolor{blue}{$\mathbf{\overline{{\mathcal S}}}$})\\\bottomrule
\end{tabular}
\end{center}
\caption{Overview of known and new results. Entries without brackets classify the graphs $F$ such that the width of $F$-free graphs is bounded. For example, a class of $F$-minor free graphs has bounded clique-width if and only if $F$ is a planar graph. Entries within brackets classify the graphs $F$ such that the class of $F$-free graphs has the diameter-width property.  We use the notation $F \subseteq_i G$ to indicate that $F$ is an induced subgraph of $G$. Unreferenced results indicate a trivial/folklore result. A ``?'' indicates an open case. Results in bold/blue are new results proven in this paper.}\label{t-table}
\vspace*{-0.5cm}
\end{table}

\medskip
\noindent
{\bf Our Results.}
Table~\ref{t-table} also contains several new results, and all the new and known results together already show a clear impact of bounding the diameter. 
For each of the new results in Table~\ref{t-table} we show that diameter $d=2$ is tight 
(see \Cref{sec:indsubgraph} and \Cref{sec:minor})
except for one result. Namely, to classify the diameter-treedepth property, we prove the following two results in Section~\ref{sec:subgraphfor}, the second of which shows that $d=5$ is tight.  In our proofs, we will exploit a result of Galvin, Rival and Sands~\cite{GALVIN19827},
who proved that a graph of large treedepth must either contain a large complete bipartite graph as a subgraph or a large induced path; see Section~\ref{sec:pre} for a more detailed discussion of this result and its consequences. 

\begin{theorem}[Classification for diameter $d \geq 5$]\label{thm:classificationDiam5}
Let $d\geq 5$. For a graph $\forbiddenGraph$, the class of $\forbiddenGraph$-subgraph-free graphs of diameter at most $d$ has bounded treedepth if and only if $\forbiddenGraph$ is a subgraph of a subdivided star. 
\end{theorem}

\begin{theorem}[Classification for diameter $4$]\label{thm:classificationDiam4}
For a graph $\forbiddenGraph$, the class of $\forbiddenGraph$-subgraph-free graphs of diameter at most $4$ has bounded treedepth if and only if $\forbiddenGraph$ is a subgraph of a subdivided star or
$\Hgraph{2}^{\ell}$ for some $\ell \in \mathbb{N}$ (see also Fig.~\ref{fig:forbiddenSubgraphs}).
\end{theorem}

\noindent
Theorems~\ref{thm:classificationDiam5} and~\ref{thm:classificationDiam4} show that there is a considerable scope for improvement by considering graphs of diameter $2$ and $3$, in line with our research question. This is unlike the other cases in Table~\ref{t-table}, for which $d=2$ is always tight.

We were not able to give complete classifications for treedepth under the subgraph relation for $d=2$ and $d=3$. However, by a deeper exploration of the result of Galvin, Rival and Sands~\cite{GALVIN19827} and using a result on polarity graphs, due to Erd{\H{o}}s, R{\'e}nyii and S{\'o}s~\cite{ERTS66},
we were able to prove a variety of results summarized in the state-of-the-art theorem below.
We again refer to Fig.~\ref{fig:forbiddenSubgraphs} for an explanation of the various graphs that we define in these statements. 

  \begin{figure}[H]
        \centering
        \includegraphics[width=\textwidth]{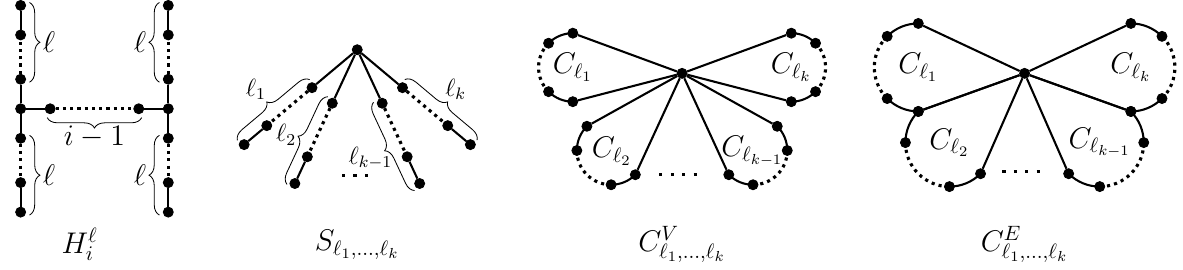}
        \caption{The subdivided ``H''-graph $H_i^\ell$, the subdivided star $S_{\ell_1,\ldots, \ell_k}$, the $V$-type graph $C^V_{\ell_1,\ldots,\ell_k}$ (set of cycles sharing one common vertex) and the $E$-type graph $C^V_{\ell_1,\ldots,\ell_k}$ (set of consecutive cycles sharing an edge). We write $C^V_{\ell_1,\ldots,\ell_k}= C_{k\times [\ell_1]}^V$ and $C^E_{\ell_1,\ldots,\ell_k}=C_{k\times [\ell_1]}^E$ if $\ell_1=\dots=\ell_k$, and $C^V_{\ell_1,\ldots,\ell_k}=C_{i\times [\ell_1],k-i \times [\ell_k]}^V$ if $\ell_1=\dots=\ell_i$, $\ell_{i+1}=\dots=\ell_k$. }
        \label{fig:forbiddenSubgraphs}
\vspace*{-2mm}
    \end{figure}

\begin{theorem}[Partial classification for diameters $2$ and $3$]\label{t-d23}
For a graph~$\forbiddenGraph$ and $d \in \{2,3\}$,
the class of $\forbiddenGraph$-subgraph-free graphs of diameter at most~$d$ has:\\[-5mm]
\begin{enumerate}[(i)]
\item bounded treedepth if:
\begin{enumerate}[1.]
\item $d=3$ and $\forbiddenGraph$ is an acyclic apex linear forest (\Cref{thm:diam3aalf});
\item $d=3$ and $\forbiddenGraph$ is~$C_8$ (\Cref{thm:C8d3});
\item $d=2$ and $\forbiddenGraph$ is a bipartite, $C_4$-subgraph-free apex linear forest that contains exactly one cycle (\Cref{thm:diam2-unicyclic});
\item $d=2$ and $\forbiddenGraph$ is $\sharedVertex{2\ell_1,2\ell_2}$ for some $\ell_1,\ell_2 \geq 3$ (\Cref{thm:diam2-CV24});
\item $d=2$ and $\forbiddenGraph$ is a subgraph of $\sharedVertex{k*[2\ell]}$ for some $\ell\geq 3$ and $k\geq 1$ (\Cref{thm:diam2-CVktimesl}) or
\item $d=2$ and $\forbiddenGraph$ is $\sharedEdge{2\ell_1,2\ell_2}$ for some $\ell_1,\ell_2\geq 3$ (\Cref{thm:CEtwoLength});\\[-9pt]
\end{enumerate}

\item unbounded treedepth if:
\begin{enumerate}[1.]
\item $d=2$ and $\forbiddenGraph$ is not bipartite (\Cref{obs:canonicalUnboundedExamples});
\item $d=2$ and $\forbiddenGraph$ 
is bipartite, but is not an apex linear forest (\Cref{obs:canonicalUnboundedExamples});
\item $d=2$ and $\forbiddenGraph$ is a supergraph of~$C_4$ (\Cref{thm:c4}~\cite{ERTS66});
\item$d=2$ and  $\forbiddenGraph$ is a supergraph of 
$\sharedVertex{12\times [6],12\times [8]}$
(\Cref{thm:CVunbounded});
\item $d=2$ and $\forbiddenGraph$ is a supergraph of $\sharedEdge{k*[2\ell]}$ for some $\ell \geq 3$ and $k = 2 (2\ell -3)$ (\Cref{thm:edgekCl});
\item $d=3$ and $\forbiddenGraph$ is a supergraph of~$C_6$ (\Cref{thm:c6}) or
\item $d=3$ and $\forbiddenGraph$ is a supergraph of either $\sharedVertex{4\ell_1,4\ell_2}$, $\sharedVertex{2\ell,2\ell}$,  $\sharedEdge{4\ell_1,4\ell_2}$ or $\sharedEdge{2\ell,2\ell}$ for $\ell_1,\ell_2 \geq 2$ and $\ell\geq 4$ (\Cref{lem:samecyc}).
\end{enumerate}
\end{enumerate}
\end{theorem}

\noindent
From Theorem~\ref{t-d23} we note a jump from boundedness for diameter~$d=2$ to unboundedness for diameter $d=3$ if $F=C_6$. Moreover, there is a difference for $d=2$ when we forbid a graph $F$ whose cycles share a unique vertex or an edge with a common end-vertex, i.e. a {\it $V$-type} graph $\sharedVertex{k*[2\ell]}$ or {\it $E$-type} graph $\sharedEdge{k*[2\ell]}$, respectively. The treedepth also becomes unbounded even for $d=2$ if the cycles of a $V$-type graph or $E$-type graph have only two different lengths.

We further emphasize that our results only hold for bounding treedepth and do not generalize to other width parameters such as pathwidth. We give specific examples exhibiting this in \Cref{s-comp}.

We show Observation~\ref{obs:canonicalUnboundedExamples} in the beginning of \Cref{sec:subgraph}.
We prove \Cref{thm:diam3aalf} in \Cref{sec:subgraphtree}.
\Cref{sec:unicyclic} contains the proof of \Cref{thm:c6}, \Cref{thm:diam2-unicyclic} and \Cref{thm:C8d3}. Finally, we show \Cref{thm:diam2-CV24}, \Cref{thm:CEtwoLength}, \Cref{thm:diam2-CVktimesl},  \Cref{thm:CVunbounded}, \Cref{thm:edgekCl} and \Cref{lem:samecyc} in \Cref{sec:subgraphmult}.

\section{Preliminaries}\label{sec:pre}

We only consider finite, simple, undirected graphs $G=(V(G),E(G))$.
Let $G$ be a graph. We denote the {\it neighbourhood} of a vertex $v\in V(G)$ by $N_G(v)=\{ u\in V(G)\; |\; uv\in E(G)\}$ (we may also just write $N(v)$). For a subset $S\subseteq V(G)$, we write $N_G(S)=\bigcup_{v \in S}N_G(v)$.
A graph~$H$ is a {\it subgraph} of $G$ if $H$ can be obtained from $G$ by a sequence of vertex deletions and edge deletions, whereas $H$ is an {\it induced subgraph} of $G$ if $H$ can be obtained from $G$ by a sequence of vertex deletions.
For a vertex set $S\subseteq V(G)$, we write $G[S]$ to denote the subgraph of $G$ {\it induced by} $S$, that is, the graph obtained from $G$ after deleting the vertices not in $S$. 
The {\it contraction} of an edge $e=uv$ in $G$ replaces $u$ and $v$ by a new vertex~$w$ that is adjacent to every vertex in $(N_G(u)\cup N_G(v)) \setminus \{u,v\}$ (without creating parallel edges). We let $G/e$ denote the graph obtained from $G$ after contracting $e$.
A graph $H$ is a {\it minor} of $G$ if $H$ can be obtained from $G$ by a sequence of edge deletions, edge contractions and vertex deletions.
For a set of graphs $\forbiddenClass$, we say that $G$ is
\emph{$\forbiddenClass$-subgraph free}, \emph{$\forbiddenClass$-free}, or \emph{$\forbiddenClass$-minor
  free} if $G$ does not contain any graph in $\forbiddenClass$ as a subgraph, induced
subgraph or minor, respectively.

We may refer to a path $P$ with vertices $u_0,\ldots,u_l$ and edges $u_{i-1}u_i$ for $1\leq i\leq l$ by the
sequence $(u_0,u_1,u_2,\dots,u_l)$. The \emph{length} of $P$ is its number of edges~$l$.
The \emph{distance} $d_G(u,v)$ between two vertices $u$ and $v$ of a graph $G$ is the length of a
shortest path from $u$ to $v$.
For two graphs $G$ and $H$, we let $G\join H$ denote
the graph obtained from the disjoint union of $G$ and $H$ after adding
all edges between the vertices in $V(G)$ and the vertices in $V(H)$.

\paragraph{Special forbidden graphs.} We let $K_n$, $C_n$ and $P_n$ denote the complete graph, cycle and path, respectively on $n$ vertices. We further let $\Hgraph{1}$ be the graph consisting of an edge $e$ with a pair of pendant vertices at each endpoint of $e$. We denote by $\Hgraph{i}$ the graph obtained from $\Hgraph{1}$ by subdividing $e$ exactly $i-1$ times (i.e. $e$ is replaced by a path of length $i$). Let $\Hgraph{i}^{\ell}$ be the graph obtained from $\Hgraph{i}$ by subdividing the edges to the pendant vertices $\ell-1$ times respectively.
We call the graphs $\Hgraph{i}$ and $\Hgraph{i}^{\ell}$,
$i,\ell\in \mathbb{N}$  \emph{$\Hgraph{}$-graphs}. For
$k,\ell_1,\dots, \ell_k\in \mathbb{N}$ we let
$\subdivStar{\ell_1,\dots,\ell_k}$ denote the star with $k$ pendant
vertices in which the edges to the pendant vertices have been
subdivided $\ell_1-1,\dots, \ell_k-1$ times, respectively. The graphs
$\subdivStar{\ell_1,\dots,\ell_3}$ are also called \emph{tripods}.
We let
$\sharedVertex{\ell_1,\dots,\ell_k}$ be the graph obtained by taking
cycles $C_{\ell_1},\dots,C_{\ell_k}$ and one vertex $v_i$ for each
cycle $C_{\ell_i}$ and then identifying $v_1,\dots, v_k$.
We further let $\sharedEdge{\ell_1,\dots,\ell_k}$ be the graph
obtained from $\sharedVertex{\ell_1,\dots,\ell_k}$ after identifying
$n_i$ and $n_{i+1}$ for every $i \in [k-1]$, where for every $i \in [k]$, $n_i$ is one of the
two neighbours of the
vertex $v$ common to all cycles in $C_{\ell_i}$.
If $\ell_1 = \dots = \ell_k$ we may also denote these graphs by
$\sharedVertex{k * [\ell_1]}$ and $\sharedEdge{k * [\ell_1]}$, respectively.

\paragraph{Treedepth}
An \emph{elimination forest} of a graph $G$ is a rooted
forest $T$ such that $V(G)=V(T)$ and for every $uv\in E(G)$
both $u$ and $v$ are on the same root-to-leaf path of $T$. The \emph{treedepth} $\td(G)$ is the minimum
height of an elimination forest of $G$. 
\begin{fact}[\cite{NM12}]\label{fact:pathLogtd}
  For a graph~$G$ with a longest path of length $\ell$,
  $\log(\ell) \leq  \td(G)\leq \ell$.
\end{fact}
The following theorem, combined with Fact~\ref{fact:pathLogtd},
has a useful consequence.
\begin{theorem}[\cite{GALVIN19827}]\label{thm:inducPathComBi}
  For all $r,s,\ell\in \mathbb{N}$, there is a number $c(r,s,\ell)$
  such that every $K_{r,s}$-subgraph-free graph with a path of length $c(r,s,\ell)$
  has an induced $P_\ell$.
\end{theorem}
\begin{corollary}\label{cor:longInducedPath}
    For all $r,s,\ell\in \mathbb{N}$, there is a number $c(r,s,\ell)$
    such that every $K_{r,s}$-subgraph-free graph of treedepth $c(r,s,\ell)$ has an induced $P_\ell$.
\end{corollary}

\paragraph{Treewidth and Pathwidth.} Treewidth is an important decompositional parameter for (sparse) graphs that informally measures how close a graph is to being a tree and that has many algorithmic applications, see,
e.g.~\cite{Marx10b} for a brief introduction into treewidth. Moreover,
pathwidth is strongly related to treewidth and informally measures the
similarity of a graph to a path. Here, we will
not formally define treewidth or pathwidth, since we will only use the following
well-known facts about treewidth and pathwidth.
\begin{fact}\label{fact:TW}
  Let $\CCC$ be any class of graphs that contains an arbitrary large
  clique or bi-clique as a minor. Then, $\CCC$ has unbounded treewidth
  and unbounded pathwidth. Moreover, if $\CCC$ contains all trees,
  then it has unbounded pathwidth.
\end{fact}   

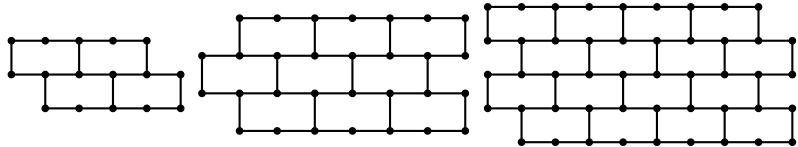
\begin{figure}
\begin{center}
\begin{minipage}{0.2\textwidth}
\centering
\begin{tikzpicture}[scale=0.45, every node/.style={scale=0.4}]
\GraphInit[vstyle=Simple]
\SetVertexSimple[MinSize=6pt]
\Vertex[x=1,y=0]{v10}
\Vertex[x=2,y=0]{v20}
\Vertex[x=3,y=0]{v30}
\Vertex[x=4,y=0]{v40}
\Vertex[x=5,y=0]{v50}

\Vertex[x=0,y=1]{v01}
\Vertex[x=1,y=1]{v11}
\Vertex[x=2,y=1]{v21}
\Vertex[x=3,y=1]{v31}
\Vertex[x=4,y=1]{v41}
\Vertex[x=5,y=1]{v51}

\Vertex[x=0,y=2]{v02}
\Vertex[x=1,y=2]{v12}
\Vertex[x=2,y=2]{v22}
\Vertex[x=3,y=2]{v32}
\Vertex[x=4,y=2]{v42}

\Edges(    v10,v20,v30,v40,v50)
\Edges(v01,v11,v21,v31,v41,v51)
\Edges(v02,v12,v22,v32,v42)

\Edge(v01)(v02)

\Edge(v10)(v11)

\Edge(v21)(v22)

\Edge(v30)(v31)

\Edge(v41)(v42)

\Edge(v50)(v51)

\end{tikzpicture}
\end{minipage}
\begin{minipage}{0.3\textwidth}
\centering
\begin{tikzpicture}[scale=0.5, every node/.style={scale=0.4}]
\GraphInit[vstyle=Simple]
\SetVertexSimple[MinSize=6pt]
\Vertex[x=1,y=0]{v10}
\Vertex[x=2,y=0]{v20}
\Vertex[x=3,y=0]{v30}
\Vertex[x=4,y=0]{v40}
\Vertex[x=5,y=0]{v50}
\Vertex[x=6,y=0]{v60}
\Vertex[x=7,y=0]{v70}

\Vertex[x=0,y=1]{v01}
\Vertex[x=1,y=1]{v11}
\Vertex[x=2,y=1]{v21}
\Vertex[x=3,y=1]{v31}
\Vertex[x=4,y=1]{v41}
\Vertex[x=5,y=1]{v51}
\Vertex[x=6,y=1]{v61}
\Vertex[x=7,y=1]{v71}

\Vertex[x=0,y=2]{v02}
\Vertex[x=1,y=2]{v12}
\Vertex[x=2,y=2]{v22}
\Vertex[x=3,y=2]{v32}
\Vertex[x=4,y=2]{v42}
\Vertex[x=5,y=2]{v52}
\Vertex[x=6,y=2]{v62}
\Vertex[x=7,y=2]{v72}

\Vertex[x=1,y=3]{v13}
\Vertex[x=2,y=3]{v23}
\Vertex[x=3,y=3]{v33}
\Vertex[x=4,y=3]{v43}
\Vertex[x=5,y=3]{v53}
\Vertex[x=6,y=3]{v63}
\Vertex[x=7,y=3]{v73}

\Edges(    v10,v20,v30,v40,v50,v60,v70)
\Edges(v01,v11,v21,v31,v41,v51,v61,v71)
\Edges(v02,v12,v22,v32,v42,v52,v62,v72)
\Edges(    v13,v23,v33,v43,v53,v63,v73)

\Edge(v01)(v02)

\Edge(v10)(v11)
\Edge(v12)(v13)

\Edge(v21)(v22)

\Edge(v30)(v31)
\Edge(v32)(v33)

\Edge(v41)(v42)

\Edge(v50)(v51)
\Edge(v52)(v53)

\Edge(v61)(v62)

\Edge(v70)(v71)
\Edge(v72)(v73)
\end{tikzpicture}
\end{minipage}
\begin{minipage}{0.35\textwidth}
\centering
\begin{tikzpicture}[scale=0.45, every node/.style={scale=0.4}]
\GraphInit[vstyle=Simple]
\SetVertexSimple[MinSize=6pt]
\Vertex[x=1,y=0]{v10}
\Vertex[x=2,y=0]{v20}
\Vertex[x=3,y=0]{v30}
\Vertex[x=4,y=0]{v40}
\Vertex[x=5,y=0]{v50}
\Vertex[x=6,y=0]{v60}
\Vertex[x=7,y=0]{v70}
\Vertex[x=8,y=0]{v80}
\Vertex[x=9,y=0]{v90}

\Vertex[x=0,y=1]{v01}
\Vertex[x=1,y=1]{v11}
\Vertex[x=2,y=1]{v21}
\Vertex[x=3,y=1]{v31}
\Vertex[x=4,y=1]{v41}
\Vertex[x=5,y=1]{v51}
\Vertex[x=6,y=1]{v61}
\Vertex[x=7,y=1]{v71}
\Vertex[x=8,y=1]{v81}
\Vertex[x=9,y=1]{v91}

\Vertex[x=0,y=2]{v02}
\Vertex[x=1,y=2]{v12}
\Vertex[x=2,y=2]{v22}
\Vertex[x=3,y=2]{v32}
\Vertex[x=4,y=2]{v42}
\Vertex[x=5,y=2]{v52}
\Vertex[x=6,y=2]{v62}
\Vertex[x=7,y=2]{v72}
\Vertex[x=8,y=2]{v82}
\Vertex[x=9,y=2]{v92}

\Vertex[x=0,y=3]{v03}
\Vertex[x=1,y=3]{v13}
\Vertex[x=2,y=3]{v23}
\Vertex[x=3,y=3]{v33}
\Vertex[x=4,y=3]{v43}
\Vertex[x=5,y=3]{v53}
\Vertex[x=6,y=3]{v63}
\Vertex[x=7,y=3]{v73}
\Vertex[x=8,y=3]{v83}
\Vertex[x=9,y=3]{v93}

\Vertex[x=0,y=4]{v04}
\Vertex[x=1,y=4]{v14}
\Vertex[x=2,y=4]{v24}
\Vertex[x=3,y=4]{v34}
\Vertex[x=4,y=4]{v44}
\Vertex[x=5,y=4]{v54}
\Vertex[x=6,y=4]{v64}
\Vertex[x=7,y=4]{v74}
\Vertex[x=8,y=4]{v84}

\Edges(    v10,v20,v30,v40,v50,v60,v70,v80,v90)
\Edges(v01,v11,v21,v31,v41,v51,v61,v71,v81,v91)
\Edges(v02,v12,v22,v32,v42,v52,v62,v72,v82,v92)
\Edges(v03,v13,v23,v33,v43,v53,v63,v73,v83,v93)
\Edges(v04,v14,v24,v34,v44,v54,v64,v74,v84)

\Edge(v01)(v02)
\Edge(v03)(v04)

\Edge(v10)(v11)
\Edge(v12)(v13)

\Edge(v21)(v22)
\Edge(v23)(v24)

\Edge(v30)(v31)
\Edge(v32)(v33)

\Edge(v41)(v42)
\Edge(v43)(v44)

\Edge(v50)(v51)
\Edge(v52)(v53)

\Edge(v61)(v62)
\Edge(v63)(v64)

\Edge(v70)(v71)
\Edge(v72)(v73)

\Edge(v81)(v82)
\Edge(v83)(v84)

\Edge(v90)(v91)
\Edge(v92)(v93)
\end{tikzpicture}
\end{minipage}
\caption{A wall of height $2$, $3$ and~$4$, respectively.}\label{fig:walls}
\end{center}
\end{figure}

\paragraph{Clique-Width.}  
We will not formally define clique-width, but refer to
e.g. the survey~\cite{DJP19} for more details. We will
need the following facts.
The first one
is well known and follows immediately from the definition of clique-width. Here, we denote the complement of a graph $G$ by 
$\overline{G}$.

\begin{fact}\label{fact:CW-COMP}
 For every graph $G$, it holds that
  $\cw(G)\leq 2\cw(\overline{G})$.
\end{fact}
Finally, we need the well known notion of a wall, which is illustrated
in Figure~\ref{fig:walls} which we will not
formally define (see, e.g.~\cite{Ch15} for a
formal definition). It is well known that walls have unbounded
clique-width (see e.g.~\cite{DBLP:journals/dam/KaminskiLM09}).
A $k$-subdivided wall is the graph obtained from a wall by subdividing it each edge~$k$ times.
\begin{fact}[{\cite{LR06}}]\label{fact:walls}
For every $k \geq 0$, the class of $k$-subdivided walls has unbounded clique-width.
\end{fact}
Finally, we need the following facts, providing the relationships
between the width parameters defined thus far.
\begin{fact}[{\cite{CR05}}]\label{fact:decrel}
  For every graph~$G$, it holds that $\cw(G) \leq 3 \times 2^{\tw(G)-1}$ and $\tw(G) \leq \pw(G)\leq \td(G)$.
\end{fact}

\section{Induced Subgraph Relation}\label{sec:indsubgraph}
Let $\forbiddenGraph$ be a graph and let $\mathcal{C}$ be the class of
$\forbiddenGraph$-free graphs with diameter at most~$d$. In this
section, we provide dichotomies characterizing exactly when $\mathcal{C}$ has
bounded treedepth, pathwidth, treewidth, or clique-width.
Note that the theorem shows that $d=2$ is tight.

\begin{theorem}\label{thm:subgraphdec}
  Let $\forbiddenGraph$ be a graph.
  The class of $\forbiddenGraph$-free graphs of diameter at most~$d$, for $d\geq
  1$, has bounded treedepth/pathwidth/treewidth if and only if either:
  \begin{itemize}
  \item $d=1$ and $\forbiddenGraph$ is a clique or
  \item $d\geq 2$ and $\forbiddenGraph\in \{K_1,K_2\}$. 
  \end{itemize}
\end{theorem}
\begin{proof}
  Let $\forbiddenGraph$ be a graph such that the class $\forbiddenClass$ of
  $\forbiddenGraph$-free graphs of diameter at most~$d$ has bounded treedepth/pathwidth/treewidth. If $\forbiddenGraph$ is not a clique, then $K_n$ is
  $\forbiddenGraph$-free, and therefore $\forbiddenClass$ contains
  graphs with diameter $1$ and unbounded treedepth/pathwidth/treewidth. Note that this completes the proof for $d=1$, since there are finitely many connected 
  $F$-free graphs with diameter at most~$1$, if $F$ is a clique.
  If $\forbiddenGraph=K_i$ for $i\geq3$, then $K_{n,n}$ is
  $\forbiddenGraph$-free, and therefore $\forbiddenClass$ contains
  graphs of diameter $2$ with arbitrarily large treedepth/pathwidth/treewidth. Hence,
  $\forbiddenGraph \in \{K_1, K_2\}$. Since the class of connected $K_i$-free
  graphs is finite if $i \in \{1,2\}$, the
  theorem follows. \qed
\end{proof}
The following theorem now provides our dichotomy result for
clique-width. Note that $d=2$ is tight.

\begin{theorem}\label{thm:IS-CW}
Let $d\geq 2$. For a graph~$\forbiddenGraph$, the class of
  $\forbiddenGraph$-free graphs of diameter at most $d$
  has bounded clique-width if and only if $\forbiddenGraph$ is an
  induced subgraph of $P_4$.
\end{theorem}
\begin{proof}
  Suppose that~$\forbiddenGraph$ contains an induced cycle~$C_k$.
  Let~$G$ be a $k$-subdivided wall and let $G'=G \join K_1$ and note that~$G'$ has diameter $2$.
  Observe that~$G$ is $C_k$-free and if $k>3$, then $G'$ is also $C_k$-free and therefore $F$-free.
  By Fact~\ref{fact:walls}, the class of $k$-subdivided walls has unbounded clique-width, so the clique-width of~$G'$ can be arbitrarily large.
  Therefore, if~$F$ contains a cycle, then we may assume that this cycle is a~$C_3$.  

    \begin{table}[]
    \centering
    \begin{tabular}{c|c|c}
      & $w_i$ & $x_i$\\
      \hline
      $w_n$ & adjacent & $w_i$\\
      $x_n$ & $w_n$ & adjacent\\
      $w_b$ & $x_i$ & adjacent\\
      $x_b$ & adjacent & $w_i$\\
      $w_r$ & $r$ & $\in N_{W_h}(w_r) \setminus N_{W_h}(w_i)$\\
      $x_r$ & $\in N_{W_h}(w_i) \setminus N_{W_h}(w_r)$ & $\in B \setminus (N_{W_h}(w_i) \cup N_{W_h}(w_j))$ \\
      $r$ & adjacent & $x_i$\\
      $b$ & $w_{n}$ & $w_b$\\
    \end{tabular}
    \caption{Table used in the proof of~\Cref{thm:IS-CW} to show that $G_h$
      has diameter $2$. Here, $w_i \in R$ and $w_n$, $w_b$, and $w_r$
      are arbitrary vertices in $N_{W_h}(w_i)$, $B
      \setminus N_{W_h}(w_i)$, and $R \setminus \{w_i\}$,
      respectively. Moreover, $x_i$, $x_n$, $x_b$, $x_r$ are the
      vertices corresponding to $w_i$, $w_n$, $w_b$, and $w_r$ in $X$,
      respectively. Each cell of the table either states that the two
      vertices corresponding to the row and column are adjacent or
      provides a common neighbour of both vertices.}
    \label{tab:commonNei}
  \end{table}

  We will now construct a class of $C_3$-free graphs of diameter~$2$ with
  unbounded clique-width. Namely, for every~$h$, we construct the
  graph~$G_h$ that is obtained from the wall~$W_h$ of height~$h$ as
  follows. Let
  $\{R,B\}$ be a proper $2$-colouring of $W_h$. We add two new vertices
  $r$ and $b$ together with the edges $(r, v)$ for every $v \in R$, $(b, v)$ for every $v \in B$ and the edge $(r, b)$. Moreover,
  for every $w_i \in V(W_h)$ we add a vertex $x_i$.
  If $w_i \in R$, then $x_i$ is
  adjacent to every $v \in B \setminus N_{W_h}(w_i)$, to $x_j$ for every $w_j \in
  N_{W_h}(w_i)$ and to $w_i$; see Figure~\ref{fig:clique-width-xi} for an illustration. We describe the
  neighbourhood of $x_i$ for $w_i \in B$ similarly interchanging the sets
  $R$ and $B$. This completes the construction of $G_h$.
  
  \begin{figure}
    \centering
    \includegraphics[page=1, width=0.7\linewidth]{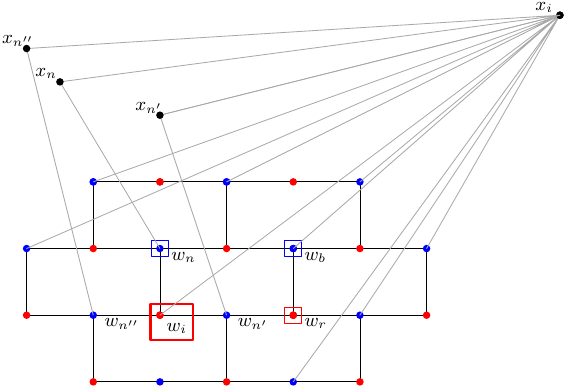}
    \caption{The construction used in~\Cref{thm:IS-CW} to provide a
      $C_3$-free class of graphs with diameter~$2$ and unbounded
      clique-width. The illustration shows a wall $W$ with a proper
      $2$-colouring using colours red and blue together with one additional
      vertex $x_i$ for every $w_i \in V(W)$.}
    \label{fig:clique-width-xi}
  \end{figure}

  We first show that $G_h$ does not contain an induced~$C_3$. This is because the graph
  $G_h[V(W_h)\cup \{r,b\}]$ is bipartite with bipartition $\{R\cup
  \{b\},B\cup \{r\}\}$ and therefore $C_3$-free. Therefore, any
  $C_3$ in $G_h$ would have to have at least one vertex from $X=\SB
  x_i\SM w_i \in V(W_h)\SE$. Since no vertex in $X$ is adjacent to both end-points
  edge of $W_h$, there can be no triangle containing exactly one vertex from~$X$.
  Similarly, since no two adjacent
  vertices in $X$ have a common neighbour in~$W_h$, there is no~$C_3$
  containing exactly two vertices from $X$. Finally, any~$C_3$
  containing three vertices from $X$ would give rise to a~$C_3$
  in $W_h$, which cannot exist since $W_h$ is bipartite.
  
  We show next that $G_h$ has diameter~$2$. To see this first consider a
  vertex $w_i$ and assume without loss of generality that $w_i \in
  R$.
  In relation to $w_i$, every vertex of $W_h$ is in one of the following three sets: $N_{W_h}(w_i)$, $B
  \setminus N_{W_h}(w_i)$ or $R \setminus \{w_i\}$. Let $w_n$, $w_b$ and
  $w_r$ be arbitrary vertices from each of these
  sets respectively. Table~\ref{tab:commonNei} now gives, for each pair of vertices,
  either a common neighbour or shows that they are adjacent, which
  shows that every pair of vertices in $G_h$ (apart from $r$ and $b$,
  which are adjacent) have distance at most~$2$.
  Since $G_h$ is $C_3$-free and therefore $\forbiddenGraph$-free, this completes the case when~$\forbiddenGraph$ contains an induced~$C_3$.
    
  It remains to consider the case when~$\forbiddenGraph$ does not contain an induced cycle, i.e. when it is a forest. If
  $\forbiddenGraph$ is not an induced subgraph of $P_4$, then it contains an induced $2P_2$ or $3P_1$.
  A wall is $\{C_3,C_4\}$-free, so the complement of a wall is $\{2P_2,3P_1\}$-free, and therefore $\forbiddenGraph$-free.
  Complements of walls have diameter~$2$ and they have arbitrarily large clique-width by Facts~\ref{fact:CW-COMP} and~\ref{fact:walls}.

  We may therefore assume that $\forbiddenGraph$ is an induced subgraph of~$P_4$.
  It is readily seen and well-known that $P_4$-free graphs have clique-width at most~$2$ even without the restriction on diameter. \qed
\end{proof}

\section{Minor Relation}\label{sec:minor}
Here, we
present our dichotomy for minor-closed classes of graphs of
bounded diameter. That is, we provide dichotomies characterizing
exactly for which graph classes~$\forbiddenClass$ , the
$\forbiddenClass$-minor-free class of graphs of diameter at most~$d$ has
bounded treedepth or bounded clique-width, respectively. W

Recall that a graph~$G$ is an apex planar graph if there is a vertex $v\in V(G)$
such that $G-v$ is planar. Analogously, $G$ is an \emph{apex forest}
or \emph{apex linear forest} if for some $v\in V(G)$, $G-v$ is a
forest or linear forest, respectively. Note that a forest is an cyclic
undirected graph, and a linear forest is a disjoint union of paths.

Recall that a class of graphs $\mathcal{C}$ has the diameter-treewidth property if there is a function
$f:\mathbb{N}\rightarrow \mathbb{N}$ such that every subgraph of a
graph in $\mathbb{C}$ with diameter at most~$d$ has treewidth at most~$f(d)$.  

\begin{theorem}[Theorem 1 in \cite{Ep00}]\label{thm:Eppstein}
  Let $\mathcal{C}$ be a minor-closed family of graphs. Then, $\mathcal{C}$ has the diameter-treewidth property if and only if $\mathcal{C}$ does not contain all apex planar graphs.
\end{theorem}
We obtain the following as a corollary. Note that $d=2$ is tight.
\begin{corollary}
Let $d\geq 2$.
  For a class of graphs $\forbiddenClass$, 
  the class $\mathcal{C}$ of $\forbiddenClass$-minor-free graphs of
  diameter at most~$d$ has bounded treewidth/clique-width if and only if
  $\forbiddenClass$ contains some apex planar graph.
\end{corollary}
\begin{proof}
  First assume that $\forbiddenClass$ contains some apex planar graph. Then
  $\mathcal{C}$ has bounded treewidth (and therefore also bounded
  clique-width~\Cref{fact:decrel}) by \Cref{thm:Eppstein}. On the
  other hand, if $\forbiddenClass$ does not contain an apex planar graph,
  then $\mathcal{C}$ contains all apex planar graphs of diameter at most~$d$
  and therefore $\mathcal{C}$ contains any wall with one added universal
  vertex, which due to~\Cref{fact:walls} and~\Cref{fact:decrel}
  implies that $\mathcal{C}$ does not have bounded treewidth or
  bounded clique-width. \qed
\end{proof}

\begin{theorem}[{Theorem 1.5 in~\cite{DEJMW20}}]\label{thm:localpw}
  Let $\mathcal{C}$ be a minor-closed family of graphs. Then
  $\mathcal{C}$ has bounded local pathwidth if and only if
  $\mathcal{C}$ does not contain all apex forests.
\end{theorem}

We obtain the following corollary, where we note that $d=2$ is tight.
\begin{corollary}
Let $d\geq 2$.
  For  class of graphs $\forbiddenClass$, 
  the class $\mathcal{C}$ of $\forbiddenClass$-minor-free graphs of
  diameter at most~$d$ has bounded pathwidth if and only if
  $\forbiddenClass$ contains some apex forest.
\end{corollary}
\begin{proof}
  First assume that $\forbiddenClass$ contains some apex forest. Then
  $\mathcal{C}$ has bounded pathwidth by \Cref{thm:localpw}. On the
  other hand, if $\forbiddenClass$ does not contain an apex forest,
  then $\mathcal{C}$ contains all apex forests of diameter at most~$d$
  and therefore $\mathcal{C}$ contains $T\join K_1$, where $T$ is any tree, which has diameter at most~$2$ and which due to~\Cref{fact:TW} implies that $\mathcal{C}$ does not have bounded pathwidth. \qed
\end{proof}

An analogous result holds for treedepth as well, where again $d=2$ is tight.

\begin{theorem}
\label{thm:td-minor}
Let $d\geq 2$.
For a class of graphs $\forbiddenClass$, the
   class $\mathcal{C}$ of $\forbiddenClass$-minor-free graphs of diameter at most~$d$ has bounded treedepth if and only if $\forbiddenClass$ contains an apex linear forest.
\end{theorem}
\begin{proof}
  We start by showing the reverse direction of the statement.
  Let $F$ be a linear apex forest contained in $\forbiddenClass$.
  We claim that for every $d$, there is a constant $c(d,|V(F)|)$ such that every graph
  with diameter at most~$d$ and treedepth at least $c(d,|V(F)|)$
  contains $P_{|V(F)|} \join K_1$ as a minor. Since $P_{|V(F)|}
  \join K_1$ contains every apex linear forest of size at most
  $|V(F)|+1$ as a minor (and therefore also $F$), this shows the reverse direction of the
  statement given in the theorem.
  
  By~\Cref{cor:longInducedPath}, there is a constant $c(|V(F)|,|V(F)|,d|V(F)|)$
  such that every graph with treedepth at least $c(|V(F)|,|V(F)|,d|V(F)|)$ either
  contains $K_{|V(F)|,|V(F)|}$ or $P_{d|V(F)|}$ as an induced subgraph. Clearly,
  $K_{|V(F)|,|V(F)|}$ contains $P_{|V(F)|}\join K_1$ as a minor, as required.
  So suppose that $G$ is a graph with diameter at most~$d$ containing
  $P=P_{d|V(F)|}$ as an induced subgraph.
  
  Now consider any vertex $v \in V(G)\setminus V(P)$, which must exists
  because $G$ has diameter at most~$d$, and let $R_v$ be the
  set of all vertices reachable from $v$ in $G\setminus V(P)$. Then,
  $R_v$ is connected and, moreover, if $P$ contains a subpath~$P'$ of
  $d$ vertices that have no neighbour in $R_v$, then $v$ has
  distance at least $d+1$ to some vertex on $P'$, contradicting our
  assumption that $G$ has diameter at most~$d$. But then, $P$ can have
  length at most $d|V(P)|$, since otherwise 
  we can obtain $K_1 \join P_{|V(P)|}$ as a
  minor of $G$ by first contracting all edges in $R_v$ and then
  contracting edges on $P$ adjacent to vertices not adjacent to $R_v$.
  
  If instead $\forbiddenClass$ does not contains any apex linear forests, then it
  contains $P_n\join K_1$ for every $n$, which has diameter at most~$2$ and unbounded treedepth because of~\Cref{fact:pathLogtd}. \qed
\end{proof}

\section{Subgraph Relation}\label{sec:subgraph}

Let $\forbiddenGraph$ be a graph and let $\mathcal{C}$ be the class of
$\forbiddenGraph$-subgraph free graphs with diameter at most~$d$. In this
section, we provide a complete characterization characterizing exactly when
$\mathcal{C}$ has bounded treedepth for $d\geq 4$ and we provide a
partial characterization for the case that $d \in \{2,3\}$.

\noindent
By Fact~\ref{fact:pathLogtd}, complete
bipartite graphs and graphs $P_n\join K_1$ have unbounded treedepth. As both classes consist of graphs of diameter at most~$2$, we obtain:
\begin{observation}\label{obs:canonicalUnboundedExamples}
Let $d\geq 2$.
  For a graph $\forbiddenGraph$, the class of
  $\forbiddenGraph$-subgraph-free graphs with diameter at most~$d$ has unbounded treedepth if 
  there is no integer $n$ such that
  $\forbiddenGraph$ is a subgraph of both $K_{n,n}$ and of $P_n\join
  K_1$.
\end{observation}
Every subgraph of $P_n\join K_1$ can have at most one
component which is not a path. Moreover, in a graph of large enough
treedepth, we can find any number of disjoint paths as subgraphs. Hence, \Cref{obs:canonicalUnboundedExamples} implies the following:

\begin{observation}
  For any graph $\forbiddenGraph$, there is a component $C$
  of $\forbiddenGraph$ such that the class of
  $\forbiddenGraph$-subgraph-free graphs with diameter $d$ has bounded
  treedepth if and only if the class of $C$-subgraph-free graphs with
  diameter $d$ has
  bounded treedepth.
\end{observation}
Therefore, from now on we consider all forbidden subgraphs
$\forbiddenGraph$ to be connected. \\

\subsection{Comparing the Four Width Parameters}\label{s-comp}
The following theorem gives a contrast between treedepth and pathwidth (or treewidth) in our setting. 

\begin{theorem}
    The class of $\Hgraph{3}$-subgraph-free graphs of diameter at most~$2$ has bounded pathwidth but unbounded treedepth. 
\end{theorem}
\begin{proof}
It follows from \Cref{obs:canonicalUnboundedExamples} that the class of
$\Hgraph{3}$-subgraph-free graphs of diameter at most~$2$ has unbounded pathwidth.

    Let $G$ be some $\Hgraph{3}$-subgraph-free graph of diameter at most~$2$ with pw$(G) \geq c(4,4,6)$. As $\Hgraph{3}$ is a subgraph of $K_{4,4}$, we find that $G$ contains an induced $6$-vertex path $P= (p_0,p_1,p_2,p_3,p_4,p_5)$ due to Corollary~\ref{cor:longInducedPath}.

    As $G$ has diameter at most~$2$,
    we find that $p_1$ and $p_4$ have a common neighbour~$x$. For the same reason,  $p_0$ and $p_4$ have a common neighbour $y$. If $y\neq x$, then $G$ contains $\Hgraph{3}$ as a subgraph with $p_1$ and $p_4$ as its vertices of degree~$3$. Hence, $x=y$. For the same reason, $p_5$ must be adjacent to $x$.
    If $p_2$ has some neighbour not in $V(P) \cup \{x\}$ then $G$ contains $\Hgraph{3}$ as a subgraph with $p_2$ and $x$ as its vertices of degree~$3$. Note that $p_2$ and $p_5$ must have a common neighbour, as $G$ has diameter at most~$2$. Hence, $p_2$ is also adjacent to $x$, and by symmetry the same holds for $p_3$. The shortest path from $p_2$ to every vertex in $G-V(P)$ must contain $x$ implying $x$ dominates $G$. We also note that $p_1$, $p_3$, and $p_4$ are of degree~$3$, due to the $\Hgraph{3}$-subgraph-freeness of $G$.
        
    We now claim $G-x$ has degree at most~$2$, which implies pw$(G) \leq 3$. Namely, if $G-x$ contains some vertex~$v$ of degree at least~$3$ with some neighbour $v'$, then there must be at least three consecutive path vertices of $P$ not contained in $N_G[v]$. Say $p_i$ is the middle of these three path vertices. As $p_i$ has degree~$3$ and needs to have a common neighbour with $v$, we find that $v'$ must be adjacent to $x$. However, we now find that $G$ contains $\Hgraph{3}$ in which vertices $v$ and $p_i$ have degree~$3$, a contradiction. \qed
\end{proof}

The following theorem gives a contrast between treewidth and clique-width in our setting. We note that $C_3$-subgraph-free graphs, or equivalently, $C_3$-free graphs have unbounded clique-width by Theorem~\ref{thm:IS-CW}.

\begin{theorem}\label{t-cwtw}
For every $r\geq 2$, the class of $C_{2r+1}$-subgraph-free graphs of diameter at most~$2$ has bounded clique-width but unbounded treewidth.
\end{theorem}

\begin{proof}
Let $r\geq 2$. As the class of complete bipartite graphs has unbounded treewidth and is a subclass of the class of $C_{2r+1}$-subgraph-free graphs of diameter at most~$2$, we find that $C_{2r+1}$-subgraph-free graphs have unbounded treewidth. It remains to prove that $C_{2r+1}$-subgraph-free graphs of diameter at most~$2$ have bounded clique-width.

Let $G$ be a $C_{2r+1}$-subgraph-free graphs of diameter at most~$2$ of arbitrarily large treewidth. Corollary~\ref{cor:longInducedPath} implies that $G$ either has the complete bipartite graph $K_{p,s}$ as a subgraph for arbitrarily large values of $r$ and $s$, or the path $P_\ell$ as an induced subgraph for arbitrarily large value of $\ell$. 

First suppose the latter case holds. Let $P$ be an induced $P_\ell$ of $G$. Take two vertices $u$ and $v$  that are of distance $2r-1$ from each other on $P$. As $P$ is an induced path in $G$, and $G$ has diameter at most~$2$, there exists a vertex $w$ not on $P$ that is adjacent to both $u$ and $v$. This yields a subgraph of $G$ that is isomorphic to $C_{2r+1}$, a contradiction.

Now suppose the former case holds. Let $K$ be a subgraph of $G$ isomorphic to $K_{p,s}$. Assume that $K$ is a maximal complete bipartite graph of $G$. Let $A$ and $B$ be the partition classes of $K$. Note that $K$ is an induced subgraph of $G$, as otherwise $K$, and thus $G$, contains a subgraph isomorphic to $C_{2r+1}$ (assuming we have chosen $p$ and $s$ large enough). We now claim that $G=K$. If not, then $G$, which is connected as it has diameter at most~$2$, contains a vertex $u$ not in $K$ that is adjacent to at least one vertex of $K$. If $u$ is adjacent to both a vertex of $A$ and a vertex of $B$, we find again that $G$ has a subgraph isomorphic to $C_{2r+1}$. Hence, we may assume without loss of generality that $u$ is only adjacent to one or more vertices of $A$. By maximality of $K$, we find that $u$ is not adjacent to every vertex of $A$, say $u$ is not adjacent to $x\in A$. We recall  that $K$ is an induced subgraph of $G$ and also that $u$ is not adjacent to any vertex of $B$. Hence, as $G$ has diameter at most~$2$, there must exist a vertex $y\notin K$ that is adjacent to $x$ and to $u$. However, now again $G$ has $C_{2r+1}$ as a subgraph, a contradiction.
\qed
\end{proof}

\subsection{Forbidding a $k$-Subdivided Graph}\label{sec:subgraphfor}
 
We first note that the $1$-subdivisions of complete bipartite graphs have
diameter $4$ and unbounded treedepth alongside the graphs
$P_n^{1001}$, i.e. for $b \in \{0,1\}^*$, where $P_n^{b}$ is the graph
obtained from $P_{n}$ after adding a new
vertex $u$ and making $u$ adjacent to the $i$-th
vertex of $P_{n}$ if the $(i \bmod |b|)$-th bit of $b$ is $1$.
Likewise, the $2$-subdivisions of complete graphs
have diameter $5$ and unbounded treedepth. 
The following therefore
holds.
\begin{observation}\label{obs:canonicalUnboundedExamplesDiam4}
  Let $\forbiddenGraph$ be a graph and $\mathcal{C}$ the class of
  $\forbiddenGraph$-subgraph-free graphs and bounded diameter $d
  \geq 4$. If there is no $n\in \mathbb{N}$ such that
  $\forbiddenGraph$ is a subgraph of both the $1$-subdivision of
  $K_{n,n}$ and of $P_n^{1001}$, then
  $\mathcal{C}$ has unbounded treedepth.
\end{observation}
\begin{observation}\label{obs:canonicalUnboundedExamplesDiam5}
    Let $\forbiddenGraph$ be a graph and $\mathcal{C}$ the class of $\forbiddenGraph$-subgraph-free graphs and bounded diameter $d \geq 5$. If there is no $n\in \mathbb{N}$ such that $\forbiddenGraph$ is a subgraph of both the $1$-subdivision of $K_{n,n}$, $2$-subdivision of $K_{n}$ and of $P_n^{1001}$, then $\mathcal{C}$ has unbounded treedepth. 
\end{observation}

If $\forbiddenGraph$ is a subgraph of the $1$-subdivision of a complete bipartite graph and some $P_n^{1001}$, then $\forbiddenGraph$ is a subgraph of $\Hgraph{2}^{\ell}$ or $\subdivStar{\ell_1, \ldots , \ell_k}$ for some $\ell \geq 1$ and $\ell_1,\ldots,\ell_k \geq 1$. If $\forbiddenGraph$ is also a subgraph of the $2$-subdivision of a complete graph, this leaves only $\subdivStar{\ell_1, \ldots , \ell_k}$ for some $\ell \geq 1$.
Forbidding $\subdivStar{\ell_1, \ldots , \ell_k}$ in the first lemma bounds the treedepth of a graph of any diameter $d$ by $c((k\ell)/2+1, (k\ell)/2, \ell(k^d+1))$ where $\ell = \max\{\ell_1, \ldots , \ell_k\}$ and $c$ is the function from \Cref{thm:inducPathComBi}. We did not try to optimize this function.
\begin{figure}[b]
    \centering
    \includegraphics[width=0.4\linewidth]{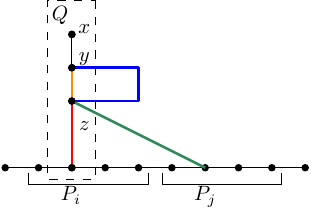}
    \caption{Path $Q$. 
    Paths $Q^1, Q^2,\overline{Q}^1$ and $\overline{Q}^2$ are orange, red, blue, and green, resp.}
    \label{fig:lem10}
\end{figure}
\begin{lemma}
\label{lem:subdivStar}
    Let $d\geq 1$. For all $\ell_1, \ldots, \ell_k\in \mathbb{N}$, the class of $\subdivStar{\ell_1, \ldots , \ell_k}$-subgraph-free graphs of diameter at most~$d$ has bounded treedepth.
\end{lemma}
\begin{proof}
    Let $G$ be some $\subdivStar{\ell_1, \ldots , \ell_k}$-subgraph-free graph of diameter at most~$d$ and $\ell = \max\{\ell_1, \ldots , \ell_k\}$. We claim $td(G) < c((k\ell)/2+1, (k\ell)/2, \ell(k^d+1))$. Suppose for contradiction $td(G) \geq c((k\ell)/2+1, (k\ell)/2, 2\ell(k^d+1))$. As $K_{(k\ell)/2+1, (k\ell)/2}$ contains $\subdivStar{\ell_1, \ldots , \ell_k}$ as a subgraph, $G$ cannot contain a large complete bipartite graph as a subgraph. From Corollary~\ref{cor:longInducedPath}, $G$ contains some induced path $P$ of length $2\ell(k^d+1)$. We may assume that  $k \geq 2$ by \Cref{fact:pathLogtd}.

    Pick disjoint subpaths $P_1,\ldots,P_{k^d+1}$ of $P$ of length $2\ell$. 
    In the following we prove inductively that for every $x\in V(G)$ and every $\delta\in [d]$ there is a path of length at most $\delta$ from $x$ to some vertex on $P_i$ for at most $k^\delta$ different $i\in [1,k^d+1]$. For $\delta=d$ this yields a contradiction to $G$ having diameter at most $d$.
    First observe that every vertex $x \in V(G)$ can have a neighbour in at most $k-1$ distinct subpaths from $\{P_1,\dots, P_{k^d+1}\}$ else $G$ contains $\subdivStar{\ell_1, \ldots , \ell_k}$. Considering that $x$ might be on $P$, we obtain that there is a path of length at most $1$ from $x$ to some vertex on $P_i$ for at most $k$ different $i\in [k^d+1]$.  
    
    Now, let $x\in V(G)$ and $\mathcal{I}\subseteq [k^d+1]$ be all indices such that there is a path of length at most $\delta+1$ from $x$ to some vertex on $P_i$. Let $\mathcal{Q}$ be a set consisting a single minimum length path $Q$ from $x$ to any vertex $q$ on $P_i$ for every $i\in \mathcal{I}$. Towards a contradiction, assume that $|\mathcal{Q}|> k^{\delta+1}$.
    As there are paths of length at most $\delta$ from any $v \in V(G)$ to some $q\in P_i$ for at most $k^\delta$ different $i\in [k^d+1]$ by assumption, there is a subset $\mathcal{Q}_0\subseteq \mathcal{Q}$ of paths that have length $\delta+1$ with $|\mathcal{Q}_0|\geq (k-1)k^\delta+1$.
    \begin{mclaim}\label{claim:overlappingPaths}
        For any path $Q\in \mathcal{Q}_0$ there are at most $k^\delta-1$ other paths $\overline{Q}\in \mathcal{Q}_0$ that intersect $Q$ in more than one vertex.  
    \end{mclaim}
    \begin{claimproof}
        Let $i\in [k^d+1]$ be the index such that the last vertex of $Q\in \mathcal{Q}$ is on $P_i$. Let $y$ be the second vertex of $Q$. 
        Let $j\in [k^d+1]$, $j\not=i$ be any index such that the path $\overline{Q}\in \mathcal{Q}$ from $x$ to some vertex in $P_j$ intersects $Q$ in some vertex $z\not=x$. We argue that there is a path from $y$ to some vertex on $P_j$ of length $\delta$.  
        Let $Q^1$ be the subpath of $Q$ from $y$ to $z$, $\overline{Q}^1$ be the subpath of $\overline{Q}$ from $y$ to $z$ and $\overline{Q}^2$ subpath of $\overline{Q}$ from $z$ to some vertex on $P_j$. As $Q$ is a shortest path from $x$ to any vertex on $P_i$ of length $\delta+1$ and $\overline{Q}$ is a shortest paths from $x$ to any vertex on $P_j$ of length $\delta+1$, $Q^1$ and $\overline{Q}^1$ have the same length. Hence, the concatenation of $Q^1$ and $\overline{Q}^2$ yields a path of length $\delta$ from $y$ to some vertex on $P_j$. 
        Since, additionally to $i$, there  can be paths of length $\delta$ to some vertex on $P_j$ for at most $k^\delta-1$ different $j\in [k^d+1]$, $j\not=i$, the claim follows.
    \end{claimproof}
    
    We now iteratively choose paths $Q_1,\dots, Q_m$ in such a way that $Q_i$ is any path from the set $\mathcal{Q}_{i}\subseteq \mathcal{Q}_0$ of paths that  intersect each $Q_j$, $j\in [i-1]$ in $x$ only. Here $m$ is the minimum integer such that $\mathcal{Q}_{m+1}=\emptyset$. 
    By \Cref{claim:overlappingPaths} we know that $|\mathcal{Q}_i|\geq |\mathcal{Q}_{i-1}|-k^\delta=(k-i)k^\delta + 1$. Hence, $|\mathcal{Q}_k|\geq 1$ which implies that there are paths $Q_1,\dots, Q_k$ pairwise only intersecting in vertex $x$. 
    Therefore, $G$ contains  $\subdivStar{\ell_1, \ldots , \ell_k}$ as a subgraph, a contradiction.
    \qed
\end{proof}

We obtain \Cref{thm:classificationDiam5} by combining \Cref{obs:canonicalUnboundedExamplesDiam5} and \Cref{lem:subdivStar}.

\begin{lemma}
\label{lem:H2}
    For any $\ell \in \mathbb{N}$, the class of $\Hgraph{2}^{\ell}$-subgraph-free graphs of diameter at most~$4$ has bounded treedepth.
\end{lemma}
\begin{proof}


    For some $\ell \geq 0$, let $G$ be some $\Hgraph{2}^{\ell}$-subgraph-free graph of diameter at most~$4$. We claim $td(G) \leq c((4\ell+1)/2, (4\ell+1)/2, 16(\ell+1)+1)$. Suppose for contradiction $td(G) > c((4\ell+1)/2, (4\ell+1)/2, 16(\ell+1)+1)$. As $K_{(4\ell+1)/2, (4\ell+1)/2}$ contains $\Hgraph{2}^{\ell}$ as a subgraph, $G$ cannot contain a large complete bipartite. From Corollary~\ref{cor:longInducedPath}, $G$ contains some induced path $P= (p_1, \ldots, p_{16(\ell+1)})$ of length $16(\ell+1)$.

    We first observe the following.
    \begin{mclaim}\label{claim:HGraph_distAtLeast3}
        For any two indices $i,j$ with $\ell \leq i < j \leq 16(\ell+1)-\ell$ and $|i-j| \geq 2\ell+1$ the shortest path between $p_{i}$ and $p_j$ must have length at least $3$.
    \end{mclaim}
    \begin{claimproof}
        As $P$ is induced there cannot be a path of length $1$ from $p_i$ to $p_j$. Furthermore, any path of length $2$ from $p_i$ to $p_j$  yields $\Hgraph{2}^{\ell}$ as a subgraph with degree 3 vertices $p_i$ and $p_j$. 
    \end{claimproof}

    We now argue that the shortest path between $p_i$ and $p_j$ for $i,j$ of sufficient distance has to be of length $4$. To prove this we use the following claim.
\begin{mclaim}\label{claim:HGraph_pathOfLength3}
    For any path $(p_i,x,x',p_j)$  with $2\ell \leq i < j \leq 16(\ell+1)-2\ell$ and $|i-j| \geq 3\ell+1$ it holds that $x$ has no neighbours besides $p_i$ on $P$ and $x'$ has no neighbours beside $p_j$ on $P$. 
\end{mclaim}
\begin{claimproof}
    As $P$ is induced $x, x' \notin P$. Furthermore,  we have that $N(x) \cap \{p_1,\dots, p_{i-(2\ell+1)}, p_{i+2\ell+1},\dots, p_{16(\ell+1)}\} = \emptyset$ by \Cref{claim:HGraph_distAtLeast3}.

    Suppose for contradiction $x$ is adjacent to $p_{k}$ with $k\not=i$ and $i-2\ell \leq  k\leq i+2\ell$. Observe that we can choose two disjoint path $P^i$ and $P^\ell$ of length $\ell-1$ within the vertices $\{p_{i-2\ell},\dots, p_{i+2\ell}\}$. Using $P^i$ and $P^\ell$ it is easy to observe that we found $\Hgraph{2}^{\ell}$ as a subgraph with degree $3$ vertices $x$ and $p_j$. Hence, $x$ cannot have any neighbours beside $x_i$ on $P$. The proof for $x'$ is symmetric.
\end{claimproof}
\begin{mclaim}\label{claim:HGraph_distance4}
    For any two indices $i,j$ with $2\ell \leq i < j \leq 16(\ell+1)-5\ell-1$ and $|i-j| \geq 3\ell+1$ the shortest path between $p_{i}$ and $p_j$ must have length $4$.
\end{mclaim}
\begin{claimproof}
    Suppose for a contradiction that $i,j$ are two indices with  $2\ell \leq i < j \leq 16(\ell+1)-5\ell-1$ and $|i-j| \geq 3\ell+1$ and there is a path $(p_i,x,x',p_j)$. 
    Let $k = j+3\ell+1$ and hence $k\leq 16(\ell+1)-2\ell$. The shortest path $Q$ from $x$ to $p_k$ must contain some vertex $p_{m} \in \{p_{i-\ell}, \dots, p_{i+\ell}\} \cup \{p_{j-\ell}, \dots, p_{j+\ell}\}$ else there is a $\Hgraph{2}^{\ell}$ with degree $3$ vertices $x$ and $p_{j}$. 
    As $m$ and $k$ satisfy the conditions of \Cref{claim:HGraph_distAtLeast3}, the subpath of $Q$ from $p_m$ to $p_k$ must have length at least $3$. As $Q$ is a shortest path this implies that $x$ is adjacent to $p_m$ and hence $k=i$ by \Cref{claim:HGraph_pathOfLength3}. Hence, there is a path of the form $(x,p_i,y,y',p_k)$. 
    We argue in a similar way for $x'$. Let $Q'$ be the shortest path from $x'$ to $p_{k+2}$. To avoid $\Hgraph{2}^{\ell}$ with degree 3 vertices $x'$ and $p_{i}$, $Q'$ must contain a vertex $p_{m'} \in \{p_{i-\ell}, \dots, p_{i+\ell}\} \cup \{p_{j-\ell}, \dots, p_{j+\ell}\}$. We get that $x'$ is adjacent to $p_{m'}$ as $Q'$ has length at most $4$ and the subpath of $Q'$ from $p_{m'}$ to $p_{k+2}$ must have length at least $3$ by \Cref{claim:HGraph_distAtLeast3}. By \Cref{claim:HGraph_pathOfLength3} we get that $j=m'$ and hence there is a path $(x',p_j,z,z',p_{k+2})$. As $y,y',z,z'$ must be pairwise disjoint by \Cref{claim:HGraph_pathOfLength3} this yields a  $\Hgraph{2}^{\ell}$ as a subgraph with degree $3$ vertices $p_k$ and $p_{k+2}$.
    Hence, the shortest path from $p_i$ to $p_j$ cannot be $3$ and therefore it must be $4$ as $G$ has diameter at most $4$.
\end{claimproof}
    
    By \Cref{claim:HGraph_distance4} the shortest path from $p_{2\ell}$ to $p_{5\ell+1}$ must have length $4$ and contain $3$ vertices not on $P$. The same holds for the path from $p_{5\ell+3}$ to $p_{8\ell+4}$ and the path from $p_{8\ell+6}$ to $p_{11\ell+7}$. Let us denote these paths by $(p_{2\ell},x,x',x'',p_{5\ell+1})$ and $(p_{5\ell+3},y,y',y'',p_{8\ell+4})$ and $(p_{8\ell+6},z,z',z'',p_{11\ell+7})$. If $x,x',x'',y,y',y''$ are all distinct then there is some $\Hgraph{2}^{\ell}$ with degree 3 vertices $p_{5\ell+1}$ and $p_{5\ell+3}$. This means $x'' = y$ as all other combinations lead to a shortest path of length at most~$3$ between some pair of vertices of $P$ with distance at least $3\ell+1$ contradicting \Cref{claim:HGraph_distance4}. This also holds for $y,y',y'', z,z',z''$ meaning $y''=z$. However, this leads to a $\Hgraph{2}^{\ell}$ with degree $3$ vertices $y$ and $z$. Thus a contradiction. \qed    
\end{proof}

Combining \Cref{obs:canonicalUnboundedExamplesDiam4}, \Cref{lem:subdivStar} and \Cref{lem:H2} gives us the  classification for diameter $4$ claimed in \Cref{thm:classificationDiam4}.

\subsection{Forbidding a Tree}\label{sec:subgraphtree}
In this subsection we prove the following result.

  \begin{figure}[t]
      \centering
      \includegraphics[width=0.55\linewidth]{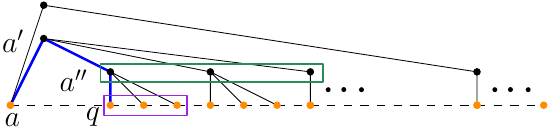}
      \caption{The dashed line is $P$ with vertices of $A$ in orange. Some vertices from $X(A,a)$ with their respective shortest paths are drawn, with the path $(a,a',a'',q)$ is in blue.}
      \label{fig:dis-paths}
  \end{figure}

\begin{theorem}
  \label{thm:diam3aalf}
  For every tree~$\forbiddenGraph$, the class of
  $\forbiddenGraph$-subgraph-free graphs of diameter at most~$3$
  has bounded treedepth if and only if $\forbiddenGraph$ is an
  acyclic apex linear forest.
\end{theorem} 
  We give a short overview of the proof of \Cref{thm:diam3aalf} before proceeding with the proof.
  The forward direction essentially follows from \Cref{obs:canonicalUnboundedExamples}. For the backwards direction,
  aiming for a proof by contradiction, we assume that for some acyclic apex linear forest $\forbiddenGraph$, the class of $\forbiddenGraph$-subgraph-free graphs of diameter at most $3$ has unbounded treedepth. Since $\forbiddenGraph$ is contained as a subgraph in $K_{n,n}$ for some sufficiently large $n$, we use \Cref{cor:longInducedPath} to obtain a graph $G$ in the class containing a long induced path $P$.
  We will proceed by showing that 
  it follows that $G$ contains, as an induced subgraph, the disjoint union of $\gamma$ copies of
  $S_{k*[\ell]}$ for some large number $\gamma$. Let $\mathcal{S}$ denote this set of $\gamma$ copies of $S_{k*[\ell]}$ and for each $S \in \mathcal{S}$ let $X_1(S)$ denote the set of vertices at distance one from the centre of $S$. We claim that, if $\gamma$ is sufficiently large, there exists some $S \in \mathcal{S}$ such that there is a bijective mapping from the vertices in $X_1(S)$ to $k$ distinct stars in $\mathcal{S} \setminus \{S\}$, such that there are disjoint paths between those vertices in $X_1(S)$ and their corresponding stars which contain no edges from $S$. 
  Finally, we show that the construction results in a copy of $\forbiddenGraph$ in $G$, a contradiction.
\begin{proof}
  A graph is an apex linear forest if, and only if it is a
  subgraph of $P_n\join K_1$ for some integer $n$. This together with
  \Cref{obs:canonicalUnboundedExamples} shows the forward direction.
  
  For the backward direction, let $\forbiddenGraph$
  be an acyclic apex linear forest. There exists some $k$ and $\ell$
  such that $\forbiddenGraph$ is a subgraph of the graph obtained from
  the disjoint union of a vertex $v$ and $k$ paths of
  length $\ell-1$ after making $v$ adjacent to exactly one vertex on
  every path. We claim the class $\mathcal{C}$ of $\forbiddenGraph$-
  subgraph-free graphs of diameter at most $3$ have bounded treedepth.

  Suppose for a contradiction that $\mathcal{C}$
  has unbounded treedepth, it must contain some graph~$G$ such that 
  $\td(G) \geq c((k \ell +1)/2, (k\ell+1)/2, \gamma(\ell+1)(5k)^7)$,
  where $\gamma$ is some function of $k$
  and the function $c$ is given by~\Cref{cor:longInducedPath}.
  As $\forbiddenGraph$ is contained in
  $K_{(k \ell +1)/2,(k\ell +1)/2}$, $G$ has no large
  complete bipartite subgraph. \Cref{cor:longInducedPath} implies that $G$
  contains some induced path~$P$ with length at least $\gamma(\ell+1)(5k)^7$. We further,  
  let $P^{\ell+1} = \{ p_i: i \mod (\ell+1)=0 \}
  \setminus p_{0}$, where $p_i$ is the $i$-th vertex of $P$ (starting
  from $i=0$).

  \begin{mclaim}\label{clm:diam3claim1}
    If $|P^{\ell+1}| > \gamma(5k)^7$, then
    $G$ contains the disjoint union of $\gamma$ copies of
    $S_{k*[\ell]}$ as an induced subgraph.
  \end{mclaim}
  \begin{claimproof}
    Consider some $A \subseteq P^{\ell+1}$ and
    $a \in A$. Let $X(A,a)$ contain all vertices that lie on some shortest
    path from $a$ to some $a' \in A$ in $G$.
    We claim that if $\vert A \vert \geq (k'-1)^3+1$,
    then there is some vertex $x \in X(A,a)$ with $k'$ disjoint
    shortest paths to vertices in $A$ for every $k' \geq k$. Note that every vertex in $G$ is adjacent to at
    most $k-1$ vertices in $A$; else
    $\forbiddenGraph$ is a subgraph of $G$; see the purple vertices in Figure~\ref{fig:dis-paths}.  Consider some $x
    \in X(A,a)$ together with all shortest paths of length at most $2$ from $x$
    to some vertex in $A$. 
    These paths can contain at most $k'-1$
    distinct intermediate vertices, else there are $k'$ disjoint
    paths from $x$ to vertices in $A$. Each intermediate vertex (green in Figure~\ref{fig:dis-paths}) has a path of length at
    most $1$ to at most $k'-1$ vertices in $A$, so there are at most
    $(k'-1)^2$ paths of length at most $2$ from $x$ to some vertex in
    $A$. As $G$ has diameter~$3$, $a$ has a path of length at most
    $3$ to every vertex in $A$. Let
    $(a,a',a'',q)$ be such a path. There are at most $(k'-1)^2$ shortest paths of
    length~$2$ from $a'$ to vertices in $A$. At most $(k'-1)^2$
    shortest paths from $a$ to some vertex in $A$ contain either $a'$
    or $a''$, so at least $(k'-1)^3+1 -(k'-1)^2$ paths (among all
    shortest paths from $a$ to some vertex in $A$) are disjoint
    from this $a$ to $q$ path. As $\vert A \vert \geq (k'-1)^3+1$,
    there are at least $k'$ disjoint shortest paths from $a$ to vertices in
    $A$.

    Let $x \in X(A,a)$ have $k'$ disjoint shortest paths to vertices in $A$, 
    these paths together with sections of $P$
    describe a $S_{k'*[\ell]}$ subgraph, which we call $S$. Let
    $X_1(S)$ and $X_2(S)$ be the sets of vertices of $S$ containing all
    vertices at distance $1$ and $2$
    from $x$ in $S$, respectively. We refer to edges in $G[S]$ but not in $S$ as
    cross edges. Given all vertices with distance at least $3$ from
    $x$ in $S$ are in $P$ and $P$ is an induced path, all cross edges have some endpoint in
    $\{x\} \cup X_1(S)\cup X_2(S)$.
    \begin{mclaim}\label{clm:x'branches}
      There is a set of at most $4k-3$ branches of $S$ whose removal
      from $S$ leaves no cross edges with an endpoint in $X_1(S) \cup
      \{x\}$.
    \end{mclaim}
    \begin{claimproof}
      If a cross edge is incident to $x$, then the other endpoint is
      in $P$ as $N(x) \cap X_2(S) = \emptyset$. If $x$ has some cross
      edge to $2k$ different branches, then $x$ has $k$ neighbours in
      $P$ each with pairwise distance at least $\ell-1$ along $P$,
      which would imply that $\forbiddenGraph$ is a subgraph of $G$.
      Therefore, after removing at
      most $2k-1$ branches, all cross edges have some endpoint in
      $X_1(S)$ or $X_2(S)$.
          
      The graph $G_{X_1}$ has one vertex for every
      branch of $S$ and an edge $(b,b')$ between two distinct
      branches $b$ and $b'$ if a vertex $x' \in
      X_1(S)$ laying on $b$ is adjacent to some vertex in $b'$. If
      $G_{X_1}$  has a matching of size  $k$, then $G$ has
      $\forbiddenGraph$ with centre $x$ as a subgraph. Matching edges indicate $k$
      vertices in $X_1(S)$ with a pair of disjoint paths of length $\ell$,
      the first via its respective branch in $S$ and the second via
      the branch given by the matching. As  $G_{X_1}$ has a maximum
      matching of size at most $k-1$, it has a vertex cover~$R$ of
      size at most $2(k-1)$. By deleting the branches
      indicated by $R$, no cross edge is adjacent to $X_1(S)$ either. 
    \end{claimproof}

\noindent
    If some $x'' \in X_2(S)$ has cross edges to $5k-3$ branches, let
    $S_{x''}$ be that $S_{(5k-3)*[\ell]}$ with centre $x''$ resulting from
    these cross edges. As $X_2(S_{x''}) \subseteq P$ all
    cross edges of $S_{x''}$ have an endpoint in $\{x''\} \cup X_1(S_{x''})$. From \Cref{clm:x'branches},
    removing at most $4k-3$ branches of $S_{x''}$ gives some induced $S_{k*[\ell]}$.
    If each vertex in $X_2(S)$ has a cross edge to at most $5k-4$ other branches
    and $S$ has $k(5k-4)+5k-4$ branches at least $k$ of these have no
    cross edges. That is, if $k' \geq (5k-4)(k+1)$, then there is some induced
    $S_{k*[\ell]}$.
    
    If there are $\gamma k(2k+1)$ disjoint subsets of $P^{\ell+1}$
    of size $(5k-4)^3(k+1)^3$ (recall that if $|A|\geq (k'-1)^3+1$,
    then $G$ has a $S_{k'*[\ell]}$ as a subgraph) there must be $\gamma k(2k+1)$
    induced stars, $S_1, \cdots, S_{\gamma(2k+1)}$, as described
    above. It remains to show that at least $\gamma$ among those stars
    are pairwise disjoint and have no edges between each other. Recall
    that every $S_i$ contains at most $2k+1$ vertices that are external to
    the path $P$. No vertex can be contained in or adjacent to $k$ of
    these stars
    as otherwise $G$ contains $\forbiddenGraph$ as a
    subgraph. Therefore, since we are given
    $\gamma k(2k+1)$ stars,
    there must be at least $\gamma$ that are pairwise disjoint and
    have no edges between each other.
  \end{claimproof}

\noindent
  Let $\mathcal{S} = \{S_1, \cdots, S_{\gamma}\}$ be the set of pairwise
  distinct, non-adjacent copies of
  $S_{k*[\ell]}$ in $G$ obtained
  from~\Cref{clm:diam3claim1}. For every $S \in
  \mathcal{S}$, let $x_S$ be the centre of $S$. For $b \in \{1,2\}$, let
  $X_b(S)$ be the set of all vertices in $S$
  of distance~$b$ from $x_S$ in $S$. 
  
  Let $T \subseteq \mathcal{S}$ with $|T| \geq 2k$. 
  For $S \in T$, if there is some bijective mapping between $k$ vertices
  in $X_1(S)$ and $k$ stars in $T$, such that there are disjoint paths between
  these $k$ vertices in $X_1(S)$ and their respective star in $T$ containing no
  edges from $S$, then $G$ contains $\forbiddenGraph$. For any pair $S,S' \in T$, the shortest path between any pair $u \in X_1(S)$ and $v \in X_1(S')$ has length at most $3$ and has at most two vertices not in $V(S) \cup V(S')$. Hence, there is a set $Z(S,T)$ of $2(k-1)$ vertices such that for every $x' \in X_1(S)$ with a neighbour $x'' \in X_2(S)$, the shortest path from $x'$ to each vertex in $X_1(S')$ contains some vertex in $Z(S,T) \cup \{x_S,x''\}$.

  \begin{figure}[b]
    \centering
    \includegraphics[width=0.7\linewidth]{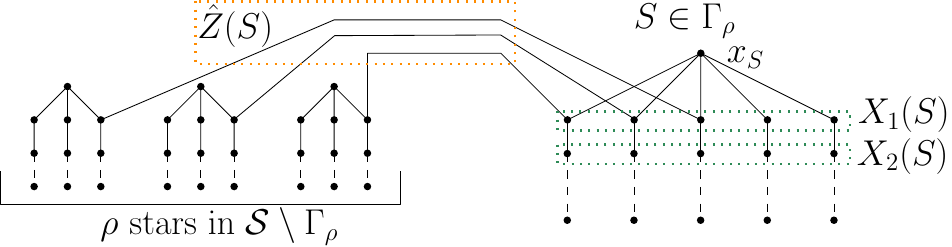}
    \caption{Some $S \in \Gamma_{\rho}$. Boxes in green indicate $X_1(S)$ and $X_2(S)$ and in orange $\hat{Z}(S)$.}\label{fig:gamma-rho}
\end{figure}

  For $\rho \leq k$, let $\Gamma_{\rho} \subseteq \mathcal{S}$ be
  such that for every $S \in \Gamma_{\rho}$ there are $\rho$ vertex
  disjoint shortest paths containing no edges from $S$ each from a
  vertex in $X_1(S)$ to its own star in $\mathcal{S} \setminus
  \Gamma_{\rho}$, i.e. there is a one-to-one mapping between the
  $\rho$ vertices and the $\rho$ stars. See Figure~\ref{fig:gamma-rho}. Note that
  $\Gamma_0=\mathcal{S}$, we claim that that if 
  $|\Gamma_{\rho}| \geq 2^{\rho+2}k^{\rho+2}$, then
  $|\Gamma_{\rho+1}|\geq \frac{|\Gamma_{\rho}|}{2^{\rho+1}k^{\rho+1}}-(k-1)^2$. This implies there is a
  constant $\gamma$ which is a function of $k$ such that if $|\mathcal{S}| = \gamma$, then $|\Gamma_k|\geq 1$. If $|\Gamma_k|\geq 1$ then $G$ contains $\forbiddenGraph$ as a subgraph, a contradiction. Therefore, for the proof of this theorem, it remains only to show that $|\Gamma_{\rho+1}|\geq \frac{|\Gamma_{\rho}|}{2^{\rho+1}k^{\rho+1}}-(k-1)^2$ for every $\rho \in \{0,\dotsc,k-1\}$.
  
  For $S \in \Gamma_{\rho}$, we let $\hat{Z}(S)$ be the set of
  all vertices that lie on the $\rho$ disjoint paths from $X_1(S)$
  to other stars and we denote by $x_{S}'$ an arbitrary vertex in
  $X_1(S)\setminus \hat{Z}(S)$; note that $x_S'$ always exists
  because $k> \rho$. Let $x_{S}''$ denote the vertex in
  $N(x_S') \cap X_2(S)$. For every $S' \in \Gamma_{\rho}$, the shortest
  path from $x_S'$ to $x_{S'}'$ must contain some vertex from
  $Z(S,\Gamma_{\rho}) \cup \{x_S, x_S''\}$. As
  $|Z(S,\Gamma_{\rho}) \cup \{x_S, x_S''\}|\leq 2k$, there exists some
  $z_{S}(\Gamma_{\rho}) \in Z(S,\Gamma_{\rho}) \cup \{x_S, x_S''\}$ that
  lies on at least $\frac{|\Gamma_{\rho}|-1}{2k}$ of these paths.
  Such a $z_{S}(T)$ can also be defined for every
  $T \subseteq \Gamma_{\rho}$ with size at least $2k$.

  \begin{mclaim}\label{claim-z0}
    For every $r \geq k^2$, if $|\Gamma_\rho| \geq r2^{\rho+1}k^{\rho+1}$, then
    there is some $\bar{S} \in \Gamma_\rho$ and $\Gamma_\rho'\subseteq
    \Gamma_\rho$ of size $r$ such that for some
    $z \notin \bigcup_{S' \in \Gamma_\rho'}\hat{Z}(S')$, $z$ lies on
    the shortest path from $x_{\bar{S}}'$ to $x_{S'}'$ for every
    $S' \in \Gamma_\rho'$.
  \end{mclaim}
  \begin{claimproof}
      Say, for contradiction, such set and vertex does not exist.
      We claim there exists a set $Q_{\rho+1} \subseteq \Gamma_\rho$
      of size at least $r$ and $\rho+1$ stars
      $A = \{S_0, \ldots, S_{\rho}\} \subseteq \Gamma_\rho \setminus
      Q_{\rho+1}$ such that the following holds.
      For every $S \in A$ there is a distinct vertex $z$, such that for every
      $S' \in Q_{\rho+1}$, $z \in \hat{Z}(S')$,
      $N(z) \cap V(S') = \emptyset$, and $z$ lies on the shortest path
      from $x_S'$ to $x_{S'}$: a contradiction as for any
      $S \in \Gamma_\rho$ there are at most $\rho$ vertices $\hat{z} \in \hat{Z}(S)$
      with $\hat{z}$ not adjacent to~$S$.

      Consider some arbitrary $S_0 \in \Gamma_{\rho}$ and its respective
      $z_{S_0}(\Gamma_{\rho})$. Let $T_0 \subseteq \Gamma_{\rho}$ be those
      $S \in \Gamma_{\rho}$ such that $z_{S_0}(\Gamma_{\rho})$ lies on the
      shortest path from $x_{S_0}'$ to $x_{S}'$. By assumption there are
      at most $r-1$ different $S \in T_0$ such that
      $z_{S_0}(\Gamma_{\rho}) \notin \hat{Z}(S)$. Given $z_{S_0}(\Gamma_{\rho})$
      has some neighbour in at most $k-1$ different stars in $\mathcal{S}$,
      there must exists some $Q_0 \subseteq T_0$ with size at least
      $|T_0|-(k-1)-r \geq \frac{|\Gamma_{\rho}|-1}{2k}-(k-1)-r$
      such that $z_{S_0}(S) \in \hat{Z}(S)$ and
      $N(z_{S_0}(S)) \cap V(S) = \emptyset$ for every $S \in Q_0$.

      Assume there exists some $Q_{\delta} \subseteq \Gamma_\rho$ and stars $A_{\delta} = \{S_0, \ldots, S_{\delta}\} \subseteq \Gamma_\rho \setminus Q_{\delta-1}$. For every $S \in A_{\delta}$ there is some distinct vertex
      $z$ such that for every $S' \in Q_{\delta}$, $z \in \hat{Z}(S')$;
      $N(z) \cap V(S') = \emptyset$; and $z$ lies on the shortest path
      from $x_S'$ to $x_{S'}$.
      
      Consider some arbitrary $S_{\delta+1} \in Q_{\delta}$ with its respective
      $z_{S_{\delta+1}}(Q_{\delta})$, let $T_{\delta+1} \subseteq Q_{\delta}$
      be those $S \in Q_{\delta}$ such that $z_{S_{\delta+1}}(S_{\delta})$
      lies on the shortest path from $x_{S_{\delta+1}}'$ to $x_{S}'$.
      If there are at most $r$ stars $S \in T_{\delta+1}$ such that
      $z_{S_{\delta+1}}(Q_{\delta}) \notin \hat{Z}(S)$, then there is some
      $Q_{\delta} \subseteq T_{\delta}$ with size at least
      $|T_{\delta}|-(k-1)-r \geq \frac{|Q_{\delta}|-k(r-1)}{2k}$
      such that $z_{S_{\delta+1}}(Q_{\delta}) \in \hat{Z}(S)$ and
      $N(z_{S_{\delta+1}}(Q_{\delta})) \cap V(S) = \emptyset$ for every
      $S \in Q_{\delta}$. If $|\Gamma_\rho| = 2^{\rho}k^{\rho}+(r-1)(\sum_{i=1}^{\rho} 2^{i}k^{i}) \leq r2^{\rho+1}k^{\rho+1}$ (for sufficiently large $r$ and $\rho$), then $|Q_{\rho+1}| \geq r$. Let $\Gamma_\rho' = Q_{\rho+1}$ which concludes the proof of this claim.
  \end{claimproof}

\noindent
    Let $S \in \Gamma_\rho$, $\Gamma_\rho'\subseteq\Gamma_\rho$ and $z$ be those
    obtained from Claim~\ref{claim-z0}. Let $\PPP$ be the set of
    shortest paths containing $z$ from $x_S'$ to $x_{S'}'$ for every
    $S' \in \Gamma_\rho'$. There are at most $k-1$
    intermediate vertices $q$ on the paths in $\PPP$ with $q \in \hat{Z}(S')$ for some
    $S' \in \Gamma_\rho'$. Each such vertex~$q$ is
    adjacent to some star, so $q$ can lie on at most $k-1$ of these paths
    from $z$ to some $S' \in \Gamma_\rho'$. Hence, there is some 
    $\Gamma_{\rho+1} \subseteq \Gamma_{\rho}'$ of size at least
    $|\Gamma_\rho'|-(k-1)^2$ such that for each $S' \in \Gamma_{\rho+1}$,
    the shortest path from $z$ to $x_{S'}'$ has no vertex from
    $\bigcup_{S'' \in \Gamma_\rho'}\hat{Z}(S'')$, thus concluding our proof. \qed
\end{proof}

\subsection{Forbidding a Unicyclic Graph}\label{sec:unicyclic} 
In this section 
we first focus on forbidding cycles and then on graphs containing only one cycle. We obtain a complete picture for diameter $2$ regarding such forbidden subgraph (\Cref{thm:diam2-unicyclic}) while for diameter $3$ we obtain partial results.
That is, for diameter $3$ we show that forbidding $C_4$ or $C_6$ does not bound the treedepth (\Cref{thm:c4} and \Cref{thm:c6}, respectively) while for $C_8$ we prove that the treedepth is bounded (\Cref{thm:C8d3}). We obtain the last result through an involved case distinction.

We recall that Erd{\H{o}}s, R{\'e}nyii and S{\'o}s~\cite{ERTS66} showed how a family of $C_4$-subgraph-free graphs with diameter $2$ can be constructed from a polarity of a projective plane. Making the observation that this family has unbounded minimum degree and so also unbounded treewidth we observed that the class $\mathcal{C}$ of $C_{4}$-subgraph-free graphs of diameter $2$ has unbounded treedepth.

\begin{theorem}[\cite{ERTS66}]
\label{thm:c4}
    The class of $C_{4}$-subgraph-free graphs of diameter $2$ has unbounded treedepth.
\end{theorem}

Considering geometries of higher dimensions an analogous result can be shown for diameter~3.

\begin{theorem}
\label{thm:c6}
    The class $\mathcal{C}$ of $C_{6}$-subgraph-free graphs of diameter $3$ has unbounded treedepth.
\end{theorem}

\begin{proof}
	Let $(\mathcal{P}, \mathcal{L},\mathcal{I})$ be a geometry defined by a set of points $\mathcal{P}$, lines $\mathcal{L}$ and incidence relation $\mathcal{I} \subseteq \mathcal{P} \times \mathcal{L}$. The corresponding instance graph~$G_I$ has vertices $\mathcal{P} \cup \mathcal{L}$ with $E(G_I) = \mathcal{I}$. In particular we consider consider regular generalised $m$-gons, these are finite geometries such that their incidence graph is $r$-regular with diameter $m$ and girth $2m$. In particular we consider where $m=4$, these are called generalised quadrangles.
	A polarity $\pi$ is a bijective function mapping points to lines and lines to points which both an involution and incidence is preserved i.e $\forall p \in \mathcal{P}, \forall l \in \mathcal{L}$ then $(\pi(l), \pi(p)) \in \mathcal{I}$ if and only if  $(p,l) \in \mathcal{I}$. From a finite geometry $(\mathcal{P}, \mathcal{L},\mathcal{I})$ and polarity $\pi$ the polarity graph~$G_{\pi}$ has the set of vertices $\mathcal{P}$ and edges $\{\{p, q\} : p, q \in \mathcal{P}, p \neq q, (p, \pi(q)) \in \mathcal{I}\}$. While such a polarity does not exist for all $(q+1)$-regular $m$-gons, a generalised quadrangles with a polarity exists where $q=p^{2\alpha+1}$, see \cite{Carter89,Tits60}. We denote this family of polarity graphs by $\mathcal{G}_{GQ}$ respectively. As the polarity of a  $(q+1)$-regular $m$-gon has minimum degree $q$,  $\mathcal{G}_{GQ}$ has unbounded treewidth.
	We claim  $\mathcal{G}_{GQ}$ is  $C_{6}$-subgraph-free with diameter $3$, more generally, if $\pi$ is a polarity of some regular generalised $m$-gon $(\mathcal{P}, \mathcal{L},\mathcal{I})$ then its corresponding polarity graph~$G_{\pi}$ is $C_{2(m-1)}$-subgraph-free with diameter $m-1$. If $G_{\pi}$ contained some $C_{2(m-1)}$ with vertices $(v_1,v_2, \ldots, v_{2(m-1)-1}, v_{2(m-1)})$ then $G_I$ contains the cycle with vertices
    $(v_1,\pi(v_2), \ldots, \pi(v_{2(m-1)-1}), v_{2(m-1)})$ as  $G_I$ has girth $m$, $G_{\pi}$ is $C_{2(m-1)}$-subgraph-free.
     As $G_I$ is a bipartite graph of diameter $m$, for any $\ell \in \mathcal{L}$ and $p \in \mathcal{P}$ if $(p,\ell) \notin \mathcal{I}$, there must be some path of length at most~$m-1$ between $p$ and $\ell$. Without loss this path can be given by $p, \ell_1,p_2, \ldots, p_{m-2}, \ell$. Let $u,v$ be a pair of non-adjacent vertices of $G_{\pi}$, we claim there must be a path of length at most~$m-1$ between them. As $(u, \pi(v)) \notin \mathcal{I}$ from above there must exist points and lines forming the path $u, \ell_1,p_2, \ldots, p_{m-2}, \pi(v)$ in $G_I$. This leads to the path $u, \pi(\ell_1),p_2, \ldots, p_{m-2}, v$ in $G_{\pi}$ with length at most~$m-1$. \qed	
\end{proof}

\Cref{thm:inducPathComBi} allows (apart from the exception of $C_4$) to classify for which cycles $\forbiddenGraph$ the class of $\forbiddenGraph$-subgraph-free graphs of bounded diameter $2$ has bounded treedepth. We prove a more general characterization concerning all graphs containing exactly one cycle in the following.

\begin{theorem}
\label{thm:diam2-unicyclic}
    Let $\forbiddenGraph$ be any graph containing exactly one cycle. The class of $\forbiddenGraph$-subgraph-free graphs of diameter at most~$2$ has bounded treedepth if and only if $\forbiddenGraph$ does not contain $C_4$, is bipartite and a subgraph of $P_n \join K_1$ for some large $n\in \mathbb{N}$.
\end{theorem}
\begin{proof}
    First note that the forward direction follows directly from \Cref{obs:canonicalUnboundedExamples} and \Cref{thm:c4}.
    
    Any graph $\forbiddenGraph$ that contains exactly one cycle of even length larger than $4$ that is a subgraph of $P_n \join K_1$ for some $n\in \mathbb{N}$ can be constructed by taking a single vertex $v$, and $k$ paths of lengths at most~$\ell$, making $v$ adjacent to one vertex on each path and choosing one paths for which $v$ has a second neighbour of distance $m-2$ from the first for some integers $k, \ell \geq 0$ and even $m>4$.  We claim that any class of $\forbiddenGraph$-subgraph-free graphs of diameter at most~$2$ has treedepth $< c(k(\ell+1)+1,k(\ell+1),2\ell(k+1)^2$.
    
    Towards a contradiction assume that there exists some graph $\forbiddenGraph$-subgraph-free graph, $G$, with diameter $2$ and treedepth at least $c(k(\ell+1)+1,k(\ell+1),2\ell(k+1)^2)$. As $K_{k(\ell+1)+1,k(\ell+1)}$ contains $\forbiddenGraph$, $G$ must contain some induced path $P = (p_0,\cdots, p_{2\ell(k+1)^2})$ by \Cref{cor:longInducedPath}. Let $P^\ell = \{ p_i: i\equiv 0 \mod (2\ell), i\not=0\}$ and note that $|P^\ell| = (k+1)^2$.
    
    \begin{mclaim}
        There exists some vertex of $G$ with at least $k+1$ neighbours in $P^\ell$.
    \end{mclaim}
    \begin{claimproof}
        We proof this by contradiction and hence assume that every vertex of $G$ has at most $k$ neighbours in $P^\ell$.
        Let $x$ be the common neighbour of the pair $(p_{\ell}, p_{\ell+m-2})$. For any vertex $p_i\in P^\ell$ we define a vertex $y_i$ which is a neighbour of $x$ and for any distinct vertices  $p_i,p_j \in P^{\ell}$ we define a vertex  $z_{ij}$ and a path $Q_{ij}$ of length at least $2\ell-1$ with middle vertex $y_i$  as follows.  Note that there  is a path of length at most $2$ from $p_i$ to $x$. We set $y_i$ to be $p_i$ in case $x$ is adjacent to $p_i$ or we choose $y_i$ to be the common neighbour of $x$ and $p_i$.  There is also a path of length at most $2$ from $y_i$ to $p_j$. We set $z_{ij}$ to be either $p_j$ if $y_i$ is adjacent to $p_j$ or to be the common vertex of $y_i$ and $p_j$. Notice that $y_i\not=z_{ij}$ and both $y_i$ and $z_{ij}$ are adjacent to some vertex in $P^{\ell}$. We can choose $Q_{i,j}$ fitting the criteria containing vertices from $\{p_{i},\dots, p_{i+\ell}), y_i,z_{i,j},p_j,\dots, p_{j+\ell}\}$. Given that no vertex is adjacent to $k+1$ vertices in $P^{\ell}$, there are at least $| P^{\ell} |-2k$ vertices $p_{i'}\in P^\ell$ such that $y_{i'} \neq y_i$ and $y_{i'} \neq z_{ij}$. Furthermore, there are at least $| P^{\ell} |-2k-1$ vertices $p_{j'}\not=p_{i'}$ such that $z_{i'j'} \neq y_i$, $z_{i'j'} \neq z_{ij}$ and $Q_{ij}$ and $Q_{i'j'}$ are disjoint. 
        Inductively, we obtain that there must be at least $k-1$ pairs $p_i,p_j\in P^\ell$ for which $Q_{ij}$ are pairwise disjoint and are disjoint from $(p_0,\dots, p_{2\ell})$ as $|P^\ell|\geq (k+1)^2$. Hence, we obtain $\forbiddenGraph$ as a subgraph with high degree vertex $x$.   
    \end{claimproof}

Let $x$ be some vertex in $G$ with at least $k+1$ neighbours in $P^\ell$ and $X$ be the set of neighbours of $x$ in $S^\ell$. If $p_i \in X$ then $p_{i+m-2} \notin X$ else $\forbiddenGraph$ is contained as a subgraph. We claim that $p_{i+2} \notin X$. Assume otherwise. As $p_i$ and $p_{i+m-2}$ must have a common neighbour $x' \neq x$ we obtain $(x',p_{i},x,p_{i+2},\dots, p_{i+m-2},x')$ as a subgraph. As this cycle of length $m$ contains $x$ and does not overlap with $(p_{j-\ell+1},\dots, p_{j+\ell})$ for any $p_j\in X$, $p_j\not=p_i$. Hence, choose any $p_i,p_{i'} \in X$. By our previous argument, the common neighbour of $p_{i+2}$ and $p_{i'+m-6}$ must have a common neighbour $y\not=x$. But then $(x,p_i,p_{i+1},p_{i+2},y,p_{i'+m-6},\dots, p_{i'},x)$ is a $C_m$ which is disjoint from any $(p_{j-\ell+1},\dots,p_{j+\ell})$ for $p_j\in X$ different from $p_i$ and $p_{i'}$. Hence, $G$ contains a copy of $\forbiddenGraph$ with high degree vertex $x$. \qed
\end{proof}

Finally, we show that forbidding $C_8$ bounds the treedepth even for graphs of diameter $3$.
\begin{theorem}\label{thm:C8d3}
    The class of $C_8$-subgraph-free graphs of diameter $d$ at most~$3$ has bounded treedepth.
\end{theorem}
\begin{proof}
    Assume $G$ is a $C_8$-subgraph-free graph of diameter at most~$3$ and  treedepth at least $c(4,4,42)$. Note that $G$ cannot contain a large complete bipartite subgraph as $K_{4,4}$ contains $C_8$ as a subgraph. Hence, by \Cref{cor:longInducedPath}, $G$ must contain $P_\ell=(p_0,\dots,p_\ell)$ where $\ell=42$ as an induced subgraph.

    First observe that for any $i\in [\ell-6]$  there cannot be a vertex $x$ not on $P$ which is adjacent to both $p_i$ and $p_{i+6}$ since $G$ is $C_8$-subgraph free. 
    Additionally, for any $i\in [\ell-5]$ there cannot be $x,y$ not on $P$ such that $(p_i,x,y,p_{i+5})$ is a path in $G$ as $(p_{i+5},x,y,p_i,\dots, p_{i+5})$ would yield a $C_8$. 
    We say that $i\in [\ell-5]$ is of 
    \begin{description}
        \item[distance $5$ type $1$] if there is $x$ such that $(p_i,x,p_{i+5})$ is a path in $G$;
        \item[distance $5$ type $2$] if there is $x$ such that $(p_{i+1},x,p_{i+5})$ is a path in $G$;
        \item[distance $5$ type $3$] if there is $x$ such that $(p_i,x,p_{i+4})$ is a path in $G$.
    \end{description}
    Since $G$ has diameter $3$ we have to ensure that the distance between $p_i$ and $p_{i+5}$ is at most~$3$. Therefore, any $i\in [\ell-5]$ has to be either of distance $5$ type $1$, $2$ or $3$.
    
    Similarly, if we consider any two vertices on $P$ of distance $6$ we get the following types. We say that $i\in [2,\ell-7]$ has  
    \begin{description}
        \item[distance $6$ type $1$] There are $x,y$ not on $P$ such that $(p_i,x, y, p_{i+6})$ is a path in $G$;
        \item[distance $6$ type $2$] There is $x$ such that $(p_i,x, p_{i+5})$ is a path in $G$;
        \item[distance $6$ type $3$] There is $x$ such that $(p_i,x, p_{i+7})$ is a path in $G$;
        \item[distance $6$ type $4$] There is $x$ such that $(p_{i-1},x, p_{i+6})$ is a path in $G$;
        \item[distance $6$ type $5$] There is $x$ such that $(p_{i+1},x, p_{i+6})$ is a path in $G$.
    \end{description}
    From our above observation it also follows that every $i\in [2,\ell-7]$ has to be of distance $6$ type $1$, $2$, 3, 4 or 5.

    We now use the types defined above to effectively consider all possible cases of how the neighbourhood (on $P$) of vertices not contained on $P$ can look like. In the following we show three claims forbidding certain configurations. Using the claims below, it is straight forward to prove that $G$ must have bounded treedepth. 
     \begin{mclaim}\label{claim:case057}
         There is no vertex $x\in V(G)$ such that $x$ is adjacent to $p_i$, $p_{i+5}$ and $p_{i+7}$ or $x$ is adjacent to $p_{i}$, $p_{i+2},p_{i+7}$ for some $i\in [10,\ell-10]$.
     \end{mclaim}
     \begin{claimproof}
         Assume the statement is not true and there is $x\in V(G)$ and $i\in [10,\ell-10]$ such that $p_i, p_{i+5},p_{i+7}\in N_G(x)$. The case when $p_i, p_{i+2},p_{i+7}\in N_G(x)$ is symmetric. We show that $G$ must contain $C_8$, a contradiction. 
         Note that  $p_{i-1},p_{i+1},p_{i+6}\notin N_G(x)$ as otherwise $(x,v_{i-1},\dots,v_{i+5},x)$ or $(x,v_{i+1},\dots, v_{i+7},x)$ or $(x, v_i,\dots, v_{i+6},x)$ is a $C_8$ in $G$.

         First, assume that $i+1$ has distance $5$ type $1$. As $p_{i+1}$ is not adjacent to $x$, there exists a vertex $y\not= x$ not on $P$ such that $(p_{i+1},y,p_{i+6})$ is a path in $G$. In this case we have $p_{i+4}\notin N_G(y)$ as otherwise $(x,p_i,p_{i+1},y,p_{i+4},\dots, p_{i+7},x)$ is a $C_8$ in $G$; $p_{i+4}\notin N_G(x)$ as otherwise $(y,p_{i+1},\dots, p_{i+4},x,p_{i+5},p_{i+6},y)$ is a $C_8$ in $G$; $p_{i-2}\notin N_G(x)$ as otherwise $(x,p_{i-2}, \dots, p_{i+1},y,p_{i+6},p_{i+7},x)$ is a $C_8$ in $G$.
         Now consider the distance $6$ type of $i-2$ (see \Cref{fig:claim057case1} for an illustration of the different cases). If $i-2$ has distance $6$ type $1$, there are $z_1,z_2$ not on $P$ and different from $x$ (as $p_{i-2}\notin N_G(x)$ and $p_{i+4}\notin N_G(x)$) such that $(p_{i-2},z_1,z_2,p_{i+4})$ is a path in $G$. Then $(p_{i-2},z_1,z_2,p_{i+4}, p_{i+5},x,p_i,p_{i-1}, p_{i-2})$ is a $C_8$ in $G$. In case $i-2$ has distance $6$ type $2$, there is $z\not= x$ (as $p_{i-2}\notin N_G(x)$) such that $(p_{i-2},z, p_{i+3})$ is a path in $G$. Then $(p_{i-2},z, p_{i+3}, p_{i+4}, p_{i+5},x,p_i,p_{i-1}, p_{i-2})$ is a $C_8$ in $G$. If $i-2$ has distance $6$ type $3$ then there is $z\not=x$ (as $p_{i-2}\notin N_G(x)$) such that $(p_{i-2},z, p_{i+5})$ is a path in $G$. Then $(p_{i-2},z, p_{i+5}, p_{i+6}, p_{i+7},x,p_i,p_{i-1}, p_{i-2})$ is a $C_8$ in $G$. If $i-2$ has distance $6$ type $4$ then there is $z\not=x$ (as $p_{i+4}\notin N_G(x)$) such that $(p_{i-3},z, p_{i+4})$ is a path in $G$. Then $(p_{i-3},z, p_{i+4}, p_{i+5},x,p_i,\dots, p_{i-3})$ is a $C_8$ in $G$. If $i-2$ has distance $6$ type $5$ then there is $z\not=y$ (as $p_{i+4}\notin N_G(y)$) such that $(p_{i-1},z, p_{i+4})$ is a path in $G$. Then $(p_{i-1},z, p_{i+4}, p_{i+5}, p_{i+6},y,p_{i+1},p_i, p_{i-1})$ is a $C_8$ in $G$. 
         \begin{figure}
            \centering
            \includegraphics[scale = 0.63]{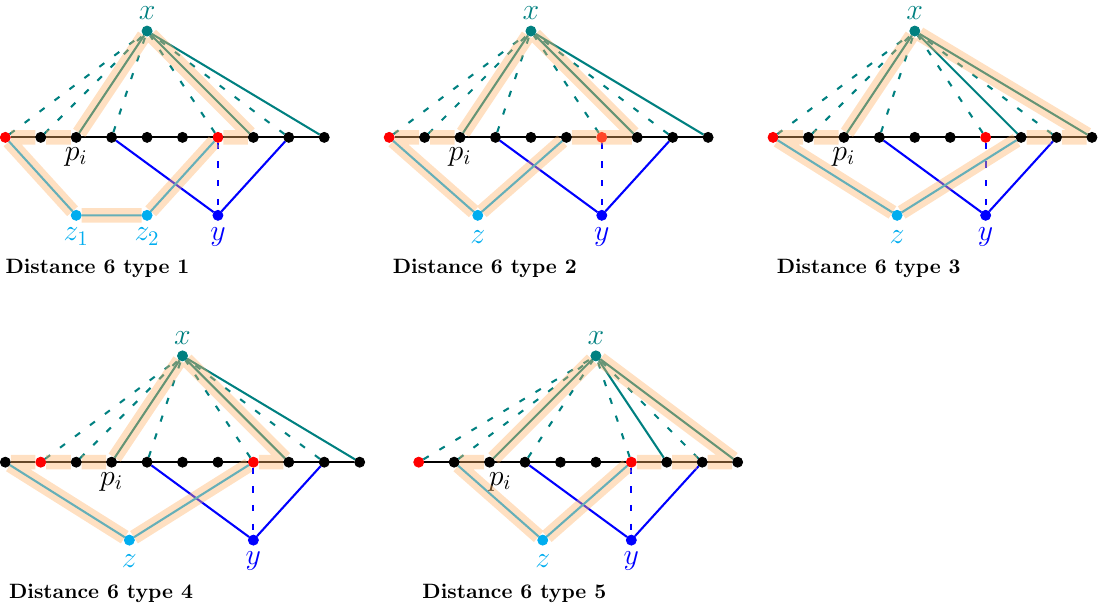}
            \caption{The five different distance $6$ types of $i-2$ in the case when $i+1$ has distance $5$ type $1$ in the proof of \Cref{claim:case057}.}
            \label{fig:claim057case1}
        \end{figure}

         Next, in case $i+1$ has distance $5$ type $2$, there is $y\not=x$ (as $p_{i+6}\notin N_G(x)$) such that $(p_{i+2},y,p_{i+6})$ is a path in $G$ and hence  $(y,p_{i+2}, \dots, p_{i+5},x,p_{i+7},p_{i+6},y)$ is a $C_8$ in $G$.

         Finally, assume that $i+1$ has distance $5$ type $3$. As $p_{i+1}\notin N_G(x)$, there is $y\not=x$ such that $(p_{i+1},y,p_{i+5})$ is a path in $G$. Note that in this case $p_{i+2}\notin N_G(x)$ as otherwise $(p_{i+2},y,p_{i+5},p_{i+6}, p_{i+7},x,p_i,p_{i+1},p_{i+2})$ is a $C_8$ in $G$ and $p_{i+8}\notin N_G(x)$ as in this case $(p_{i+8},x,p_i,p_{i+1},y,p_{i+5},\dots, p_{i+8})$ is a $C_8$ in $G$. 
         We now consider the distance $6$ type of $i+2$ (see \Cref{fig:claim057case1} for an illustration of the different cases). If $i+2$ has distance $6$ type $1$, then there are $z_1,z_2$ not on $P$ and different from $x$ (as $p_{i+2}\notin N_G(x)$ and $p_{i+8}\notin N_G(x)$) such that $(p_{i+2},z_1,z_2,p_{i+8})$ is a path in $G$. Then $(p_{i+2},z_1,z_2,p_{i+8}, p_{i+7},x,p_i,p_{i+1},p_{i+2})$ is a $C_8$ in $G$. In case $i+2$ has distance $6$ type $2$, there is $z\not=x$ (as $p_{i+2}\notin N_G(x)$) such that $(p_{i+2},z,p_{i+7})$ is a path in $G$. Then $(p_{i+2},z,p_{i+7}, p_{i+6}, p_{i+5},x,p_i,p_{i+1},p_{i+2})$ is a $C_8$ in $G$. In case $i+2$ has distance $6$ type $3$, there is $z\not=x$ (as $p_{i+2}\notin N_G(x)$) such that $(p_{i+2},z,p_{i+9})$ is a path in $G$. Then $(p_{i+2},z,p_{i+9},p_{i+8},p_{i+7},x,p_i,p_{i+1},p_{i+2})$ is a $C_8$ in $G$. If $i+2$ has distance $6$ type $4$, there is $z\not=x$ (as $p_{i+8}\notin N_G(x)$) such that $(p_{i+1},z,p_{i+8})$ is a path in $G$. Then $(p_{i+1},z,p_{i+8},\dots, p_{i+5},x,p_i,p_{i+1})$ is a $C_8$ in $G$. Finally, if $i+2$ has distance $6$ type $5$, there is $z\not=x$ (as $p_{i+8}\notin N_G(x)$) such that $(p_{i+3},z,p_{i+8})$ is a path in $G$. Then $(p_{i+3},z,p_{i+8},p_{i+7},x,p_i,\dots, p_{i+3})$ is a $C_8$ in $G$. We conclude that the claim holds.
         \begin{figure}
            \centering
            \includegraphics[scale = 0.65]{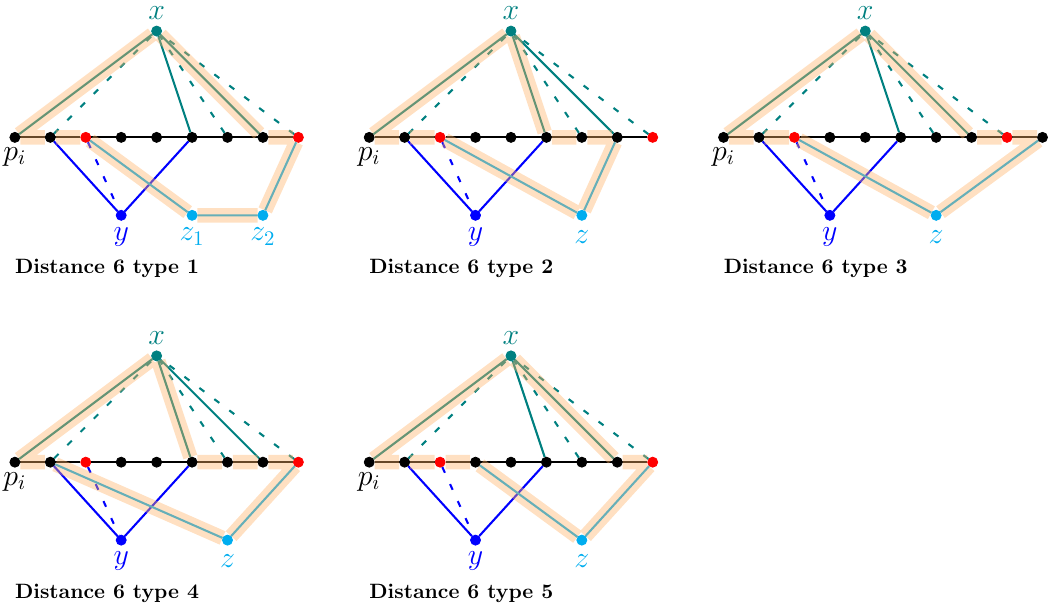}
            \caption{The five different distance $6$ types of $i+2$ in the case when $i+1$ has distance 5 type $3$ in the proof of \Cref{claim:case057}.}
            \label{fig:claim057case2}
        \end{figure}
     \end{claimproof}
     
    \begin{mclaim}\label{claim:case014}
        There is no vertex $x\in V(G)$ such that $x$ is adjacent to $p_i$, $p_{i+1}$ and $p_{i+4}$ or $x$ is adjacent to $p_i$, $p_{i+3}$ and $p_{i+4}$ for some $i\in [10,\ell-10]$.
    \end{mclaim}
    \begin{claimproof}
        Assume the claim is false and there is $x\in V(G)$ and $i\in [10,\ell-10]$ such that $p_i, p_{i+1},p_{i+4}\in N_G(x)$. The case where $p_i, p_{i+3},p_{i+4}\in N_G(x)$ is symmetric. To avoid $C_8$ we directly obtain that $p_{i-2},p_{i+6},p_{i+7}\notin N_G(x)$. 
        
        First, assume that $p_{i-1}\notin N_G(x)$. 
        As $G$ has diameter at most~$3$ there must be a path~$Q$ from $p_{i-1}$ and $p_{i+6}$ of length at most~$3$. As $p_{i-1},p_{i+6}\notin N_G(x)$ the path~$Q$ cannot contain $x$ (both inner vertices are adjacent to $p_{i-1}$ or $p_{i+6}$). As $P$ is induced, at least one vertex of $Q$ is not contained in $P$. Hence, the union of $P$ and $Q$ must contain a cycle $C$ of length 8,9 or 10. In each case, $C$ contains $(p_i,\dots,p_{i+4})$ as a subpath. If $C$ has length 9, then replacing $(p_{i+1},\dots, p_{i+4})$ by $(p_{i+1},x,p_{i+4})$  yields a $C_8$. On the other hand, if $C$ has length 10, then replacing $(p_i,\dots, p_{i+4})$ by $(p_i,x,p_{i+4})$  yields a $C_8$.
        
        On the other hand, in the case that $p_{i-1}\in N_G(x)$ there must be a path~$Q$ from $p_{i-2}$ and $p_{i+6}$ of length at most~$3$. As $p_{i-2},p_{i+6}\notin N_G(x)$ the path~$Q$ cannot contain $x$. Additionally, as $P$ is induced, at least one vertex of $Q$ is not contained in $P$. Therefore, the union of $P$ and $Q$ must contain a cycle $C$ of length 9,10 or 11. Note that $C$ must contain $(p_{i-1},\dots,p_{i+4})$ as a subpath. If $C$ has length 9, then replacing $(p_{i+1},\dots, p_{i+4})$ by $(p_{i+1},x,p_{i+4})$  yields a $C_8$. If $C$ has length 10, then replacing $(p_i,\dots, p_{i+4})$ by $(p_i,x,p_{i+4})$  yields a $C_8$. Finally, if $C$ has length 11, then replacing $(p_{i-1},\dots, p_{i+4})$ by $(p_{i-1},x,p_{i+4})$ yields a $C_8$. We conclude that the claimed must be true.
    \end{claimproof}
    
    \begin{mclaim}\label{claim:case04}
        There is no vertex $v\in V(G)$ such that $v$ is adjacent to $p_i$ and $p_{i+4}$ for some $i\in [20,\ell-20]$.
    \end{mclaim}
    \begin{claimproof}
        Assume the statement is not true and there is $v\in V(G)$ and $i\in [20,\ell-20]$ such that $p_i,p_{i+4}\in N_G(v)$. Note that to avoid $C_8$ we get that $p_{i-2},p_{i+6}\notin N_G(v)$ and additionally $p_{i+1},p_{i+3}\notin N_G(v)$ by \Cref{claim:case014}. \\

        \noindent\textbf{Case 1:} First assume that $i-3$ has distance $5$ type $1$. Hence, there is $w$ such that $(p_{i-3}, w, p_{i+2})$ is a path in $G$. Note that $w\not=v$ as a consequence of \Cref{claim:case057}.  Furthermore, to avoid $C_8$ we have that $p_{i-4},p_{i+3}\notin N_G(w)$ and $p_{i-5},p_{i+4}\notin N_G(w)$ by \Cref{claim:case057}. Additionally, $p_{i-2}\notin N_G(w)$ as otherwise $(p_{i-2},w,p_{i+2}, p_{i+3},p_{i+4},v,p_{i}, p_{i-1},p_{i-2})$ is a $C_8$ in $G$.
        We now consider the distance $6$ type of $i-2$.

        \noindent\textbf{Case 1a:} First consider $i-2$ has distance $6$ type $1$. In this case there  are $x_1,x_2$ not on $P$ and different from $w$ (as $p_{i-2},p_{i+4}\notin N_G(w)$) such that $(p_{i-2},x_1,x_2,p_{i+4})$ is a path in $G$. In this case $(p_{i-2},x_1,x_2,p_{i+4},p_{i+3},p_{i+2},w,p_{i-3},p_{i-2})$ is a $C_8$ in $G$.

        \begin{figure}
            \centering
            \includegraphics[scale = 0.63]{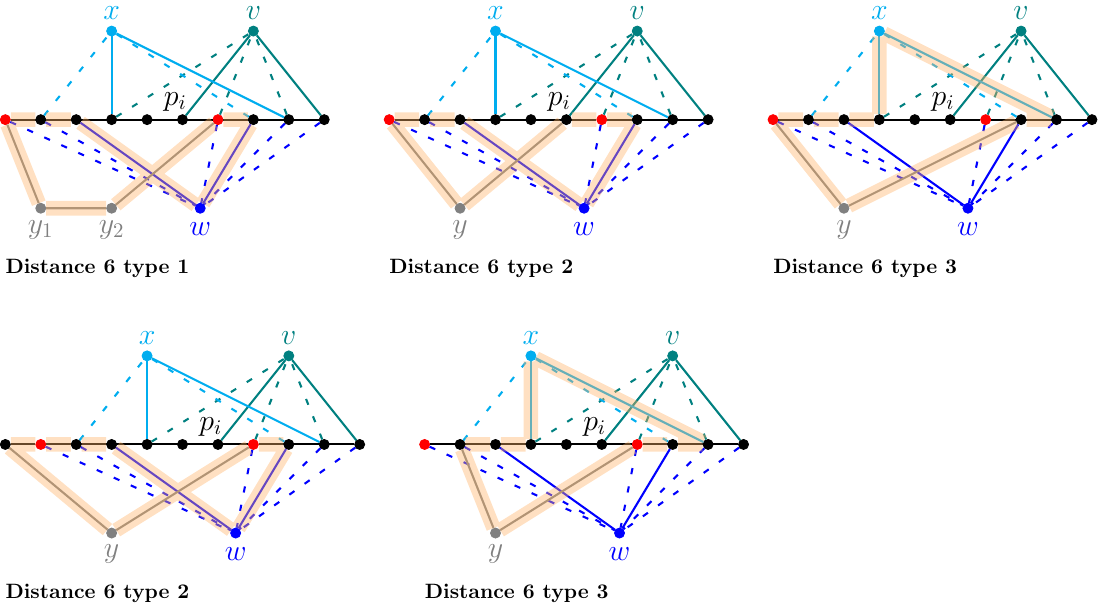}
            \caption{The five different distance $6$ types of $i-5$ in the case that $i-3$ has distance $5$ type $1$ and $i-2$ has distance $6$ type $2$ in the proof of \Cref{claim:case04}.}
            \label{fig:claim04case1}
        \end{figure}
        \noindent\textbf{Case 1b:} Next, assume $i-2$ has distance $6$ type $2$. In this case there is $x\not=v$, $x\not=w$ (as $p_{i-2}\notin N_G(v)$ and $p_{i-2}\notin N_G(w)$) such that $(p_{i-2},x,p_{i+3})$ is a path in $G$. Note that $p_{i-4}\notin N_G(x)$ by \Cref{claim:case057} and $p_{i+2}\notin N_G(x)$ as otherwise $(p_{i-2},x, p_{i+2},p_{i+3},p_{i+4},v, p_{i},p_{i-1},p_{i-2})$ is a $C_8$ in $G$. Furthermore, in this case $p_{i+1}\notin N_G(w)$ as otherwise $(x, p_{i-2},\dots, p_{i+1},w,p_{i+2},p_{i+3},x)$ is a (non-induced) $C_8$ in $G$.   
        We now consider the distance $6$ type of $i-5$ (see \Cref{fig:claim04case1} for an illustration of the different cases). 
        First assume $i-5$ has distance $6$ type $1$. In this case, there are $y_1,y_2$ not on $P$ and different from $w$ (as $p_{i-5},p_{i+1}\notin N_G(w)$) such that $(p_{i-5},y_1,y_2,p_{i+1})$ is a path in $G$. Then $(p_{i-5},y_1,y_2,p_{i+1},p_{i+2},w,p_{i-3},p_{i-4},p_{i-5})$ is a $C_8$ in $G$. 
        Hence, consider that $i-5$ has distance $6$ type $2$. Then there is $y\not= w$ (as $p_{i-5}\notin N_G(w)$) such that $(p_{i-5},y,p_{i})$ is a path in $G$. Then $(p_{i-5},y,p_{i},p_{i+1},p_{i+2},w,p_{i-3},p_{i-4},p_{i-5})$ is a $C_8$ in $G$. 
        Next, consider $i-5$ has distance $6$ type $3$. In this case there is $y\not= x$ (as $p_{i+2}\notin N_G(x)$) such that $(p_{i-5},y,p_{i+2})$ is a path in $G$. Then $(p_{i-5},y,p_{i+2},p_{i+3},x, p_{i-2},\dots, p_{i-5})$ is a $C_8$ in $G$. 
        If $i-5$ has distance $6$ type $4$, then there is $y\not= w$ (as $p_{i+1}\notin N_G(w)$) such that $(p_{i-6},y,p_{i+1})$ is a path in $G$. Then $(p_{i-6},y,p_{i+1},p_{i+2},w,p_{i-3},\dots,p_{i-6})$ is a $C_8$ in $G$. 
        Finally, if $i-5$ has distance $6$ type $5$, there is $y\not= x$ (as $p_{i-4}\notin N_G(x)$) such that $(p_{i-4},y,p_{i+1})$ is a path in $G$. Then $(p_{i-4},y,p_{i+1},p_{i+2},p_{i+3},x, p_{i-2},p_{i-3}, p_{i-4})$ is a $C_8$ in $G$, a contradiction.

        \noindent\textbf{Case 1c:} Next, consider the case that $i-2$ has distance $6$ type $3$. In this case there  is $x\not=w$ (as $p_{i-2}\notin N_G(w)$) such that $(p_{i-2},x,p_{i+5})$ is a path in $G$. In this case $(p_{i-2},x,p_{i+5},\dots,p_{i+2},w,p_{i-3},p_{i-2})$ is a $C_8$ in $G$.

        \begin{figure}
            \centering
            \includegraphics[scale = 0.63]{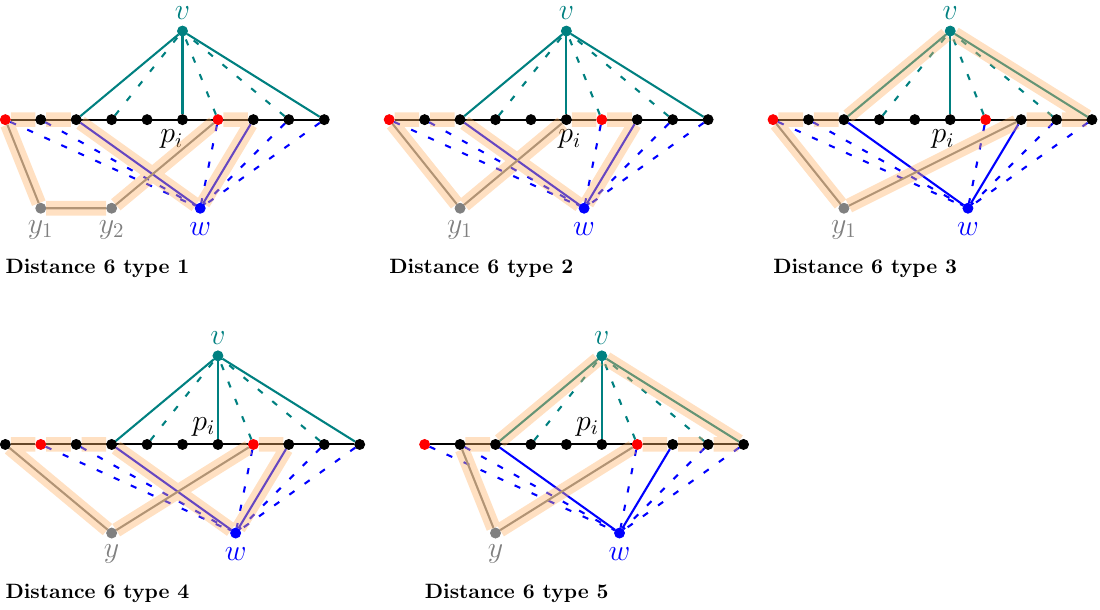}
            \caption{The five different distance $6$ types of $i-5$ in the case that $i-3$ has distance $5$ type $1$ and $i-2$ has distance $6$ type $4$ in the proof of \Cref{claim:case04}.}
            \label{fig:claim04case2}
        \end{figure}
        \noindent\textbf{Case 1d:} In case $i-2$ has distance $6$ type $4$, there is $x\not= w$ (as $p_{i+4}\notin N_G(w)$) such that $(p_{i-3},x, p_{i+4})$ is a path in $G$. First observe that in case $x\not= v$ we obtain $(p_{i-3},x,p_{i+4},v,p_i,p_{i+1},p_{i+2},x,p_{i-3})$ as a $C_8$ in $G$. Hence, $x=v$ and therefore additionally $p_{i-3}\in N_G(v)$. We now consider the distance $6$ type of $i-5$ (see \Cref{fig:claim04case2} for an illustration of the different cases). 
        First assume $i-5$ has distance $5$ type $1$. In this case there are $y_1,y_2$ not on $P$ and different from $v$ (as $p_{i-5},p_{i+1}\notin N_G(v)$, where the former would yields $(p_{i-5},v,p_i,p_{i+1},p_{i+2},w,p_{i-3},p_{i-4},p_{i-5})$ as a $C_8$ in $G$) such that $(p_{i-5},y_1,y_2,p_{i+1})$ is a path in $G$. But then we get $(p_{i-5},y_1,y_2,p_{i+1},p_i,v,p_{i-3},p_{i-4},p_{i-5})$ is a $C_8$ in $G$. Next, assume that $i-5$ has distance $6$ type $2$. In this case there is $y\not=w$ (as $p_{i-5}\notin N_G(w)$) such that $(p_{i-5},y,p_i)$ is a path in $G$. Hence, $(p_{i-5},y,p_i,p_{i+1},p_{i+2},w,p_{i-3},p_{i-4},p_{i-5})$ is a $C_8$ in $G$. 
        Next, consider the case that $i-5$ has distance $6$ type $3$. In this case there is $y\not=v$ (as $p_{i-5}\notin N_G(v)$) such that $(p_{i-5},y,p_{i+2})$ is a path in $G$. But then $(p_{i-5},y,p_{i+2},p_{i+1},p_i,v,p_{i-3},p_{i-4},p_{i-5})$ is a $C_8$ in $G$.
        Assume that $i-5$ has distance $6$ type $4$. Then there is $x\not=v$ (as $p_{i+1}\notin N_G(v)$) such that $(p_{i-6},y,p_{i+1})$ is a path in $G$. Hence, $(p_{i-6},y,p_{i+1}, p_i,v,p_{i-3},\dots, p_{i-6})$ is a $C_8$ in $G$.
        Finally, consider the case that $i-5$ has distance $6$ type $5$. In this case there is $y\not=v$ (as $p_{i+1}\notin N_G(v)$) such that $(p_{i-4},y,p_{i+1})$ is a path in $G$. Then $(p_{i-4},y,p_{i+1},\dots, p_{i+4},v,p_{i-3},p_{i-4})$ is a $C_8$ in $G$.

        \noindent\textbf{Case 1e:} Finally, assume that $i-2$ has distance $6$ type $5$. Hence, there  is $x\not=w$ (as $p_{i+4}\notin N_G(w)$) such that $(p_{i-1},x,p_{i+4})$ is a path in $G$. Then $(p_{i-1},x,p_{i+4},p_{i+3},p_{i+2},w,p_{i-3},p_{i-2},p_{i-1})$ is a $C_8$ in $G$.
        \\

        \noindent\textbf{Case 2:} In the case that $i-3$ has distance $5$ type $2$, there is $x\not=w$ (as $p_{i-2}\notin N_G(w)$) such that $(p_{i-2},x,p_{i+2})$ is a path in $G$. In this case we get that $(p_{i-2},x,p_{i+2},p_{i+3},p_{i+4},w,p_i, p_{i-1},p_{i-2})$ is a $C_8$ in $G$, a contradiction.\\

        \noindent\textbf{Case 3:} It remains to  consider the case that $i-3$ has distance $5$ type $3$. In this case there is $w\not=v$ (as $p_{i+1}\notin N_G(v)$) such that $(p_{i-3},w,p_{i+1})$ is a path in $G$. Note that $p_{i-5},p_{i+3}\notin N_G(w)$ to avoid $C_8$ and $p_{i-2},p_i\notin N_G(w)$ by \Cref{claim:case014}. We consider the distance $5$ type of $i-2$. First observe that in case $i-2$ has distance type $2$, we get $x\not=v$ (as $p_{i-2}\notin N_G(v)$) such that $(p_{i-2},x,p_{i+2})$ is a path in $G$ and hence $(p_{i-2},x,p_{i+2}, p_{i+3},p_{i+4},v,p_i,p_{i-1},p_{i-2})$ is a $C_8$ in $G$. Similarly, in case $i-2$ has distance type $3$ we get $x\not=w$ (as $p_{i+3}\notin N_G(w)$) such that $(p_{i-1},x,p_{i+3}, p_{i+2},p_{i+1},w,p_{i-3},p_{i-2},p_{i-1})$ is a $C_8$ in $G$. Hence, $i-2$ must have distance $5$ type $1$. Therefore, there exists $x\not=v$, $x\not=w$ (as $p_{i-2}\notin N_G(v)$ and $p_{i-2}\notin N_G(w)$) such that $(p_{i-2},x,p_{i+3})$ is a path in $G$. Note that $p_{i-3},p_{i+4}\notin N_G(x)$ to avoid $C_8$ and $p_{i-4},p_{i+5}\notin N_G(x)$ by \Cref{claim:case057}. Furthermore, $p_{i+2}\notin N_G(x)$ as otherwise $(p_{i-2},x,p_{i+2},p_{i+3},p_{i+4},x,p_i,p_{i-1},p_{i-2})$ is a $C_8$ in $G$.
        We now consider the 6 type of $i-4$.  

        \noindent\textbf{Case 3a:} First, assume that $i-4$ has distance $6$ type $1$. In this case there are $y_1,y_2$ not in $P$ and different from $x$ (as $p_{i-4},p_{i+2}\notin N_G(x)$) such that $(p_{i-4},y_1,y_2,p_{i+2})$ is a path in $G$. Then $(p_{i-4},y_1,y_2,p_{i+2},p_{i+3},x,p_{i-2},p_{i-3},p_{i-4})$ is a $C_8$ in $G$.

        \noindent\textbf{Case 3b:} Next, assume that $i-4$ has distance $6$ type $2$. Hence, there is $y\not=x$ (as $p_{i-4}\notin N_G(x)$) such that $(p_{i-4},y,p_{i+1})$ is a path in $G$. In this case  $(p_{i-4},y,p_{i+1},p_{i+2},p_{i+3},x,p_{i-2},p_{i-3},p_{i-4})$ is a $C_8$ in $G$.

        \begin{figure}
            \centering
            \includegraphics[scale = 0.63]{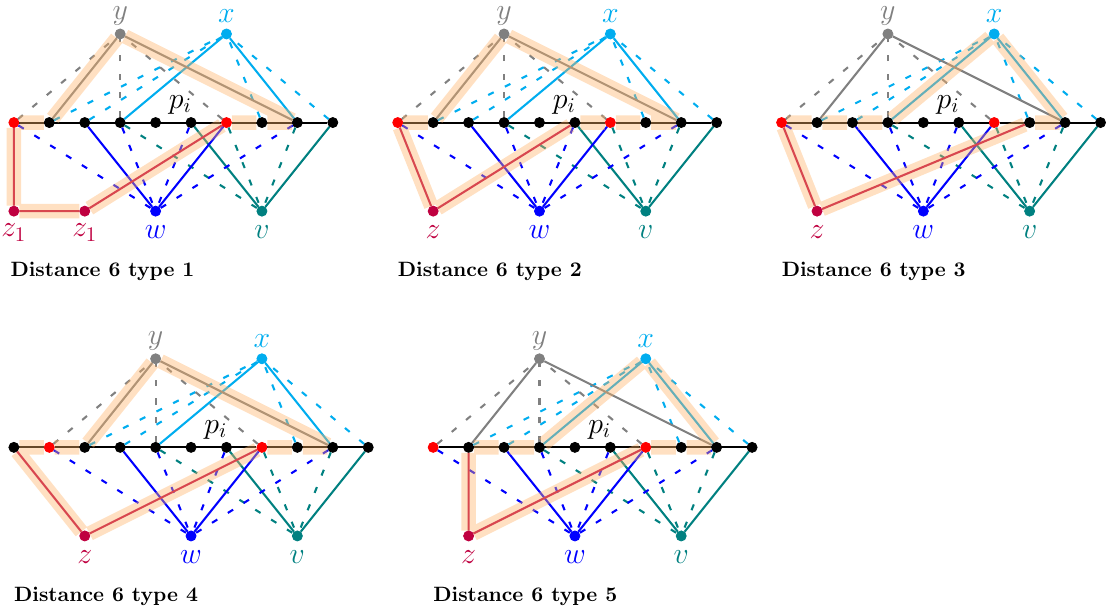}
            \caption{The five different distance $6$ types of $i-5$ in the case that $i-3$ has distance $5$ type $3$ (which implies that $i-2$ has distance $5$ type $1$) and $i-4$ has distance $6$ type $3$ in the proof of \Cref{claim:case04}.}
            \label{fig:claim04case3}
        \end{figure}
        \noindent\textbf{Case 3c:} Consider that $i-4$ has distance $6$ type $3$. Hence, there is $y\not=w$, $y\not=x$ (as $p_{i+3}\notin N_G(w)$ and $p_{i-4}\notin N_G(x)$) such that $(p_{i-4},y,p_{i+3})$ is a path in $G$. Note that $p_{i-2},p_{i+1}\notin N_G(y)$ by \Cref{claim:case057}. Additionally, $p_{i-5}\notin N_G(y)$ as otherwise $(p_{i-5},y,p_{i+3},p_{i+2},p_{i+1},w,p_{i-3},p_{i-4},p_{i-5})$ is a $C_8$ in $G$. We consider the distance $6$ type of $i-5$ (see \Cref{fig:claim04case3} for an illustration of the different cases). First assume that $i-5$ has distance $6$ type $1$. In this case there are $z_1,z_2$ not on $P$ and different from $y$ (as $p_{i-5},p_{i+1}\notin N_G(y)$) such that $(p_{i-5},z_1,z_2,p_{i+1})$ is a path in $G$. Hence,  $(p_{i-5},z_1,z_2,p_{i+1},p_{i+2},p_{i+3},y,p_{i-4},p_{i-5})$ is a $C_8$ in $G$. 
        Similarly, in case $i-5$ has distance $6$ type $2$ there is
        $z\not=y$ (as $p_{i-5}\notin N_G(y)$) such that $(p_{i-5},z,p_{i})$ is a path in $G$. Then $(p_{i-5},z,p_{i},\dots,p_{i+3},y,p_{i-4},p_{i-5})$ is a $C_8$ in $G$. 
        Now assume that $i-5$ has distance $6$ type $3$. Then there is $z\not=x$ (as $p_{i+2}\notin N_G(x)$) such that $(p_{i-5},z,p_{i+2})$ is a path in $G$. Hence, $(p_{i-5},z,p_{i+2},p_{i+3},x,p_{i-2},\dots,p_{i-5})$ is a $C_8$ in $G$.
        In case $i-5$ has distance $6$ type $4$, there is $z\not=y$ (as $p_{i+1}\notin N_G(y)$) such that such that $(p_{i-6},z,p_{i+1})$ is a path in $G$. Then $(p_{i-6},z,p_{i+1},p_{i+2},p_{i+3},y,p_{i-4},p_{i-5},p_{i-6})$ is a $C_8$ in $G$. 
        Finally, assume that $i-5$ has distance $6$ type $5$. Hence, there is $z\not=x$ (as $p_{i-4}\notin N_G(x)$) such that $(p_{i-4},z,p_{i+1})$ is a path in $G$. Then $(p_{i-4},z,p_{i+1},p_{i+2},p_{i+3},x,p_{i-2},p_{i-3},p_{i-4})$ is a $C_8$ in $G$.

        \noindent\textbf{Case 3d:} Next assume that $i-4$ has distance $6$ type $2$. In this case there is $y\not=x$ (as $p_{i+2}\notin N_G(x)$) such that $(p_{i-5},y,p_{i+2})$ is a path in $G$. In this case $(p_{i-5},y,p_{i+2},p_{i+3},x,p_{i-2},\dots,p_{i-5})$ is a $C_8$ in $G$.

        \begin{figure}
            \centering
            \includegraphics[scale = 0.63]{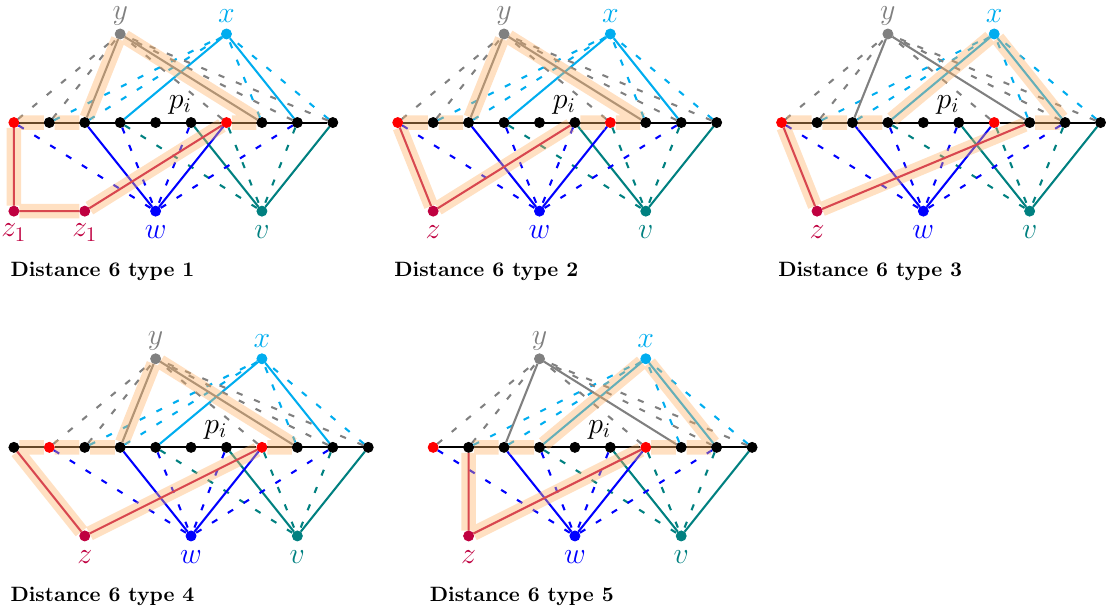}
            \caption{The five different distance $6$ types of $i-5$ in the case that $i-3$ has distance $5$ type $3$ (which implies that $i-2$ has distance $5$ type $1$) and $i-4$ has distance $6$ type $5$ in the proof of \Cref{claim:case04}.}
            \label{fig:claim04case4}
        \end{figure}
        \noindent\textbf{Case 3e:} Finally, consider the case that $i-4$ has distance $6$ type $5$. Hence, there is  $y\not=x$ (as $p_{i+2}\notin N_G(x)$) such that $(p_{i-3},y,p_{i+2})$ is a path in $G$.  To avoid $C_8$ we get that $p_{i-4},p_{i+3}\notin N_G(y)$. Additionally, $p_{i-5},p_{i+4}\notin N_G(y)$ by \Cref{claim:case057} and $p_{i+1}\notin N_G(y)$ as otherwise $(p_{i+1},y,p_{i+2},p_{i+3},x,p_{i-2},\dots,p_{i+1})$ is a (non-induced) $C_8$ in $G$. We consider the distance $6$ type of $i-5$ (see \Cref{fig:claim04case4} for an illustration of the different cases). 
        First assume that $i-5$ has distance $6$ type $1$. In this case there are $z_1,z_2$ not on $P$ and different from $y$ (as $p_{i-5},p_{i+1}\notin N_G(y)$) such that $(p_{i-5},z_1,z_2,p_{i+1})$ is a path in $G$. Hence, $(p_{i-5},z_1,z_2,p_{i+1},p_{i+2},y,p_{i-3},p_{i-4},p_{i-5})$ is a $C_8$ in $G$. 
        Similarly, in case $i-5$ has distance $6$ type $2$ there is
        $z\not=y$ (as $p_{i-5}\notin N_G(y)$) such that $(p_{i-5},z,p_{i})$ is a path in $G$. In this case $(p_{i-5},z,p_{i},\dots,p_{i+2},y,p_{i-3},p_{i-4},p_{i-5})$ is a $C_8$ in $G$. 
        Now assume that $i-5$ has distance $6$ type $3$. Then there is $z\not=x$ (as $p_{i+2}\notin N_G(x)$) such that $(p_{i-5},z,p_{i+2})$ is a path in $G$. Hence, $(p_{i-5},z,p_{i+2},p_{i+3},x,p_{i-2},\dots,p_{i-5})$ is a $C_8$ in $G$.
        In case $i-5$ has distance $6$ type $4$, there is $z\not=y$ (as $p_{i+1}\notin N_G(y)$) such that such that $(p_{i-6},z,p_{i+1})$ is a path in $G$. Then $(p_{i-6},z,p_{i+1},p_{i+2},y,p_{i-3},p_{i-4},p_{i-5},p_{i-6})$ is a $C_8$ in $G$. 
        Finally, assume that $i-5$ has distance $6$ type $5$. Hence, there is $z\not=x$ (as $p_{i-4}\notin N_G(x)$) such that $(p_{i-4},z,p_{i+1})$ is a path in $G$ and $(p_{i-4},z,p_{i+1},p_{i+2},p_{i+3},x,p_{i-2},p_{i-3},p_{i-4})$ is a $C_8$ in $G$. 
        
        As we found a $C_8$ in each possible case contradicting our assumption that $G$ is $C_8$-subgraph free, the claim must be true.
    \end{claimproof}
    
    Consider the distance $5$ type of $20$. By  \Cref{claim:case04} we know that $20$ must have distance $5$ type $1$ and hence there is $x$ such that $(p_{20},x,p_{25})$ is a path in $G$. Now consider the distance $5$ type of 22. Similarly, by \Cref{claim:case04} we know that 22 has distance $5$ type $1$ and hence there is $y$ such that $(p_{22},y,p_{27})$ is a path in $G$. 
    Furthermore, we know that $x\not=y$ by \Cref{claim:case057}.
    Therefore, $(p_20,x,p_{25},p_{26},p_{27},y,p_{22},p_{21},p_{20})$ is a $C_8$ in $G$. We conclude that $G$ has treedepth at most $c(4,4,42)$. \qed
\end{proof}

\subsection{Forbidding a Graph with Multiple Cycles}\label{sec:subgraphmult}
In this section we forbid subgraphs containing multiple cycles. By \Cref{obs:canonicalUnboundedExamples} we only need to consider graphs that are subgraphs of $K_{n,n}$ and $P_n \bowtie K_1$ for some sufficiently large $n$. We identify V-type and E-type graph as natural classes of forbidden subgraphs to study.
Note that for even numbers $\ell_1,\dots,\ell_k$,
both the graph $\sharedVertex{\ell_1,\dots,\ell_k}$ and $\sharedEdge{\ell_1,\dots,\ell_k}$ are subgraphs of both $K_{n,n}$ and $P_n\join K_1$ for some large enough $n$.
We show that for diameter $2$, forbidding a V-type or E-type graph with two arbitrary even length (larger than $4$) cycles  as a subgraph bounds the treedepth (\Cref{thm:diam2-CV24} and \Cref{thm:CEtwoLength}, respectively). Furthermore, forbidding a V-type graph with an arbitrary number of cycles of the same even length (larger than $4$) bounds the treedepth (\Cref{thm:diam2-CVktimesl}) while for equivalent E-type graphs treedepth remains unbounded (\Cref{thm:edgekCl}). Finally, we show that forbidding a V-type graph containing multiple cycles of two fixed even length the treedepth becomes unbounded (\Cref{thm:CVunbounded}).
For diameter $3$ we did not find any examples of forbidden V-type or E-type graphs that bound the treedepth. However, we are able to show that for certain combinations of length forbidding V-type or E-type graphs does not bound treedepth (\Cref{lem:samecyc}). Notably, forbidding a V-type graph containing just two cycles of the same length does not bound treedepth, a clear distinction from diameter $2$.

\begin{theorem}
\label{thm:diam2-CV24}
    For any $\ell_1,\ell_2>2$ the class of all $\sharedVertex{2\ell_1,2\ell_2}$-subgraph-free graphs of diameter at most~$2$ has bounded treedepth.
\end{theorem}
\begin{proof}
    Let $G$ be some $\sharedVertex{2\ell_1,2\ell_2}$-subgraph-free graphs with diameter at most~$2$, we claim $td(G) < c(2\ell_1+2\ell_2, 2\ell_1+2\ell_2, 4(\ell_1+\ell_2-1)+1)$. Suppose for contradiction $td(G) \geq c(2\ell_1+2\ell_2, 2\ell_1+2\ell_2, 4(\ell_1+\ell_2-1)+1)$. As $G$ cannot contain  $K_{2\ell_1+2\ell_2, 2\ell_1+2\ell_2}$ as a subgraph, Corollary~\ref{cor:longInducedPath} implies $G$ contains an induced path~$P$ of length $4(\ell_1+\ell_2-1)$. Let $P= (p_0, \ldots, p_{4(\ell_1+\ell_2-1)})$.

    As $G$ has diameter at most~$2$, $p_0$ and $p_{2\ell_1-1}$ must have a common neighbour, call this vertex $x$. The common neighbour of $p_{2\ell_1-1}$ and  $p_{2(\ell_1+\ell_2-1)}$ must also be $x$ otherwise $\sharedVertex{2\ell_1,2\ell_2}$ is contained as a subgraph with the pair of cycles of length $2\ell_1$ and $2\ell_2$ sharing the common vertex $x$. This also implies the common neighbours of $p_{2(\ell_1+\ell_2-1)}$ and $p_{4\ell_1+2\ell_2-3}$ as well as  $p_{4\ell_1+2\ell_2-3}$ and $p_{4\ell_1+4\ell_2-4}$ are $x$. However this leads to cycles of length $2\ell_1$ and $2\ell_2$ with a single common vertex $x$. \qed
  \end{proof}

\begin{theorem}  
\label{thm:diam2-CVktimesl}
    For any integers $\ell\geq 3$ and $k\geq 1$ the class of all $\sharedVertex{k*[2\ell]}$-subgraph-free graphs of diameter at most~$2$ has bounded treedepth.
\end{theorem}
\begin{proof}
    Let $G$ be a $\sharedVertex{k*[2\ell]}$-subgraph-free graph of diameter at most~$2$, we claim $td(G) \leq c(k (2\ell-1)+1, k (2\ell-1),4^{k+1}k^{k+4}\ell)$. Suppose for contradiction $td(G) > c(k (2\ell-1)+1, k (2\ell-1),4^{k+1}k^{k+4}\ell)$. $G$ cannot contain a large complete bipartite subgraph as $K_{k (2\ell-1)+1, k (2\ell-1)}$ contains $\sharedVertex{k*[2\ell]}$ as a subgraph. From Corollary~\ref{cor:longInducedPath},  $G$ contains an induced path, $P=(p_0,\dots,p_m)$ 
    of length $m=4^{k+1}k^{k+4}\ell$. 

    We define the distance of two pairs  $(p_i,p_j)$ and $(p_{i'},p_{j'})$ of vertices on $P$ with $i<j$ and $i'<j'$ as follows. If either $i\leq i'\leq j$ or $i\leq j'\leq j$ we set the distance of $(p_i,p_j)$ and $(p_{i'},p_{j'})$ to be $0$. Otherwise, the distance of $(p_i,p_j)$ and $(p_{i'},p_{j'})$ is the positive integer $d$ for which $j'+d=i$ if $j'<i$ or $j+d=i'$ if $j<i'$.
    
    Suppose some vertex $v \in V(G)$ has $k$  pairs of neighbours in $P$, $(p_i,p_j)$ such that $j = i+2\ell-2$ and these $k$ pairs have pairwise distance at least $1$. Such a vertex results in \forbiddenGraph\ as there are $k$ cycles of length $2\ell$ each containing the single common vertex $v$. Hence, no vertex is adjacent to $k$ pairs of neighbours $(p_i,p_{i+2\ell-2})$ of pairwise distance at least $1$.

    \begin{mclaim}\label{clm:dist2l-4}
            Let $x$ be some vertex in $G$, then $N(x) \cap \{p_0,\dots,p_{m-\ell}\}$ contains at most $(k-1)^2$ disjoint pairs $\{p_i,p_j\}$ where $j = i+2\ell-4$, $j \leq m- \ell$ and each pair has pairwise distance at least $\ell$.
    \end{mclaim}
    \begin{claimproof}
            Say $x$ has neighbours $p_i, p_{i+2\ell-4} \in \{p_0,\dots,p_{m-\ell}\}$. The vertices $p_{i+\ell-3}$, $p_{i+3\ell-5}$ must have a common neighbour, $x'$. We call the vertex $x'$, that is the common neighbour of $p_{i+\ell-3}$ and $p_{i+3\ell-5}$ the \emph{connector} of the pair $(p_i,p_j)$ of neighbours of $x$. If $x'=x$ then there is some $C_{2\ell}$ containing only $x$ and the path vertices $p_{i+\ell-3}, \dots, p_{i+3\ell-5}$, otherwise $x'\neq x$ and there is a $C_{2\ell}$ given by $(x, p_i, \ldots, p_{i+\ell-3}, x', p_{i+3\ell-5}, \ldots, p_{i+2\ell-4})$. Note there is a cycle of length $2\ell$ containing $x'$ and the path vertices $p_{i+\ell-3}$, \ldots, $p_{i+3\ell-5}$.

            Suppose $N(x) \cap \{p_0,\dots,p_{m-\ell}\}$ contains $(k-1)^2+1$ disjoint pairs with pairwise distance at least $\ell$. Let $x'_r$ denote the \emph{connector} for the $r$th pair with $X' = \{x'_1,\dots, x'_{(k-1)^2+1}\}$. If $X'$ contains $k$ pairwise distinct vertices, then $G$ contains \forbiddenGraph\ with $k$ cycles of length $2\ell$ each cycle, apart from possibly one, containing $x$ and a distinct vertex $x' \in X'$. Note that one of these distinct connectors could be the vertex $x$ itself, in which case the cycle does not contain an additional vertex $x'\in X'$. 
            On the other hand, consider $|X'| \leq k-1$. As there are $(k-1)^k+1$ connector there must be some $x'\in X'$ which is the connector of at least $k$ distinct pairs. Assume without loss of generality, $x'_1=x'_2=\ldots=x'_k$. However, this is a contradiction as there are $k$ cycles of length $2\ell$ containing vertices from $P$ and the single common vertex $x'_1$.
    \end{claimproof}

    \begin{mclaim}\label{clm:vbndnei}
            For every $x \in V(G)$, there are less than $4^kk^{k+1}$ vertices in $N(x) \cap \{p_0,\dots,p_{m-3\ell-4}\}$ 
            with pairwise distance at least $3\ell-4$.
    \end{mclaim}
    \begin{claimproof}
        Consider some vertex $x_0$ and let $Z_0 \subseteq N(x_0) \cap \{p_0,\dots,p_{m-3\ell-4}\}$ be some set of vertices with pairwise distance at least $3\ell-4$ along the path. Suppose $|Z_0| \geq 4^kk^{k+1}$. 
        Let $Z^+_0 = \{ p_{i+2\ell-4}: p_i \in Z_0 \}$, note $|Z^+_0| = |Z_0|$.
        
         In the following we recursively construct vertices $x_0, x_1,\ldots,x_{k-1}$, sets $Z_{0} \supseteq Z_1 \supseteq \ldots \supseteq Z_{k-1}$ and sets $\hat{Z}_{0} \supseteq \hat{Z}_{1} \supseteq \ldots \supseteq \hat{Z}_{\delta-1} \supseteq \hat{Z}_{k-1}$ such that for every $1\leq i \leq k-1$
         \begin{enumerate}
             \item $|Z_i|\geq \left\lfloor {\frac{|Z_{i-1}|- (k-1)^2|}{2(k-1)}} \right\rfloor \geq \left\lfloor {\frac{|Z_{i-1}|}{4(k-1)}} \right\rfloor$, where $|Z_{i-1}| \geq 2(k-1)^2$;
             \item $|\hat{Z}_i|\geq |Z_i|-(k-1)^2$;
             \item $Z_i\subseteq N(x_0, \dots, x_i)\cap \{p_0,\dots, p_m\}$;
             \item $\hat{Z}_i\cap N(x_0,\dots, x_{i})=\emptyset$;
             \item $N(x_i)\cap \hat{Z}_{i-1}\neq \emptyset$; and 
             \item $\hat{Z}_i\subseteq Z_0^+$.
         \end{enumerate} 
        
        Suppose for $\delta \leq k-2$ we have constructed vertices $x_0,x_1,\ldots,x_{\delta}$ and sets $Z_{0} \supseteq Z_1 \supseteq  \ldots \supseteq Z_{\delta}$ and $\hat{Z}_0 \supseteq \hat{Z}_1 \supseteq \ldots \supseteq \hat{Z}_{\delta-1}$ with the properties above. Let $Z^+_{\delta} = \{ p_{i+2\ell-4}: p_i \in Z_{\delta-1}\}$. Again from Claim~\ref{clm:dist2l-4}, at most $(k-1)^2$ vertices in $Z^+_{\delta}$ are adjacent to $x_{\delta}$. Let $\hat{Z}_{\delta} =Z^+_{\delta} \setminus N(x,x_1, \ldots, x_{\delta})$ with $|\hat{Z}_{\delta}| \geq |Z_{\delta}| -(k-1)^2$. Thus properties $2$, $4$ and $6$ hold.
        Let $(p_i,p_j)$ be some pair such that $p_i \in Z_{\delta}$, $j\not=i+2\ell-4$ and $p_j \in \hat{Z}_{\delta}$. As $p_i,p_j$ must have some common neighbour, $x' \neq x$, there is some $C_{2\ell}$ given by $(x,p_i, x', p_j, \cdots, p_{j-(2\ell-4)})$. Note that we can choose  $\left\lfloor {\frac{|Z_{\delta-1}- (k-1)^2|}{2}} \right\rfloor$ such pairs with distance at least $1$.
        Let $(p_{r_1},p_{r_2})$ denote the $r$th pair and let $x'_r$ denote that common neighbour of $p_{r_1},p_{r_2}$. Let $X' = \{x'_1,\dots, x'_{\left\lfloor {\frac{|Z_{\delta}|- (k-1)^2}{2}} \right\rfloor}\}$. If $X'$ contains $k$ distinct vertices, then without loss of generality $x'_1 \neq \ldots \neq x'_k$. For every $1 \leq r \leq k$, there is some $C_{2\ell}$ containing $x_0, x'_r$ and vertices from $\{p_{r_1}, \ldots, p_{r_2}\}$, given the pairs $(p_{r_1},p_{r_2})$ have pairwise distance at least $1$ this results in \forbiddenGraph. Therefore, by the pigeon hole principle there is some vertex, we call this $x_{\delta+1}$, which is the common vertex for at least $\left\lfloor {\frac{|Z_{|\delta}|- (k-1)^2}{2(k-1)}} \right\rfloor$ different pairs. Note that $x_{\delta+1}$ is adjacent to at least $\left\lfloor {\frac{|Z_{\delta}|- (k-1)^2}{2(k-1)}} \right\rfloor$ vertices in $\hat{Z}_{\delta}$, and $\left\lfloor {\frac{|Z_{\delta}|- (k-1)^2}{2(k-1)}} \right\rfloor$ vertices in $Z_{\delta}$ that is properties $1$, $3$ and $5$ hold.

        Given $|Z_0| \geq 4^kk^{k+1} \geq 4^{k-1}(k-1)^{k-1}\cdot ((k-1)^2+1)$ there exist vertices $x_0, \ldots, x_{k-1}$, and sets $Z_{0} \supseteq \ldots \supseteq Z_{k-1}$ and $\hat{Z}_{0} \supseteq \ldots \supseteq \hat{Z}_{k-1}$ with each of the properties described above. From property $5$, for every $1 \leq i \leq k-1$, $N(x_{i}) \cap \hat{Z}_{i-1} \neq \emptyset$, this implies there is some vertex in $\hat{Z}_{i-1}$ adjacent to $x_i$, let $y_i$ be some arbitrary vertex in $\hat{Z}_{i-1}$ adjacent to $x_i$. Let $y_k$ be some arbitrary vertex in $\hat{Z}_{k-1}$. From property $4$, $\hat{Z}_i\cap N(x_0,\dots, x_{i})=\emptyset$ meaning, vertices $y_i$ are distinct for each $1 \leq i \leq k$. Further from property $6$, $\{y_1, \ldots, y_k\} \subseteq Z^+_0$. That is for every $1 \leq i \leq k$ there is some distinct $p_j \in Z_0$, we call this vertex $y'_i$, such that $p_{j+2\ell-4} = y_i$. Let $Y_i = \{p_{j}, \ldots, p_{j+2\ell-4}\}$ where $j$ is chosen to satisfy $p_{j+2\ell-4}=y_i$. Note the sets $Y_1, \ldots, Y_k$ are pairwise disjoint.

        Further note $|Z_{k-1}| \geq 2k$, meaning that there exist a set of vertices $C = \{c_1,\ldots,c_k\} \subseteq Z_{k-1}$ such that $y'_i \notin C$ for any $1 \leq i \leq k$.
        Let $x_k$ be the common neighbour of $y_k$ and $c_k$, as $y_k \in \hat{Z}_{k-1}$, $x_k \neq x_i$ for any $0 \leq i \leq k-1$. However, this is a contradiction as for each $1 \leq i \leq k$, there is some cycle of length $2\ell$ containing $x_0$, $c_i$, $x_i$ and $Y_i$, given these cycles have a single common vertex, $x_0$, this describes \forbiddenGraph.
    \end{claimproof}

    We divide $P$ into pairwise disjoint subpaths each containing $m'$ vertices, where $m' = (2\ell-2)((k-1)^2+1)+\ell \geq 10\ell k^2$. We call each of these subpath a segment of $P$. As $m = 4^{k+1}k^{k+4}\ell$, there are at least $4^{k+1}k^{k+2}$ such segments. Let $x_0$ be the common neighbour of $p_0$ and $p_{2(\ell-1)}$. From Claim~\ref{clm:vbndnei}, $x_0$ cannot have neighbours in $4^kk^{k+1}$ different segments (excluding the final segment, i.e. that containing $p_m$). Given there are at least $4^kk^{k+1}+1$ segments there is some segment which does not contain a neighbour of $x_0$.

    \begin{mclaim}\label{clm-seg}
        Let $X$ be a set of external vertices and $Q = (q_0, \cdots, q_{m'})$ be some segment which does not contain a neighbour of any vertex in $X$. Then there is a $C_{2\ell}$ containing $p_0$ and exactly $2$ external vertices $ x\notin X$ and $x' \notin X$ and vertices from $Q$.
    \end{mclaim}
    \begin{claimproof}
        Let $x$ be the common neighbour of $p_0$ and $q_0$. Let $y$ the common neighbour of $p_0$ and $q_{2\ell-4}$. As both $x$ and $y$ have a neighbour in $Q$ we know that $x\notin X$ and $y\notin X$ by assumption. If $y \neq x$, we let $x'=y$, our claim holds as the cycle $(p_0,x,q_0, \dots, q_{2\ell-4},x')$ is of length $2\ell$ and contains $p_0$ and exactly $2$ external vertices $ x\neq x_0$ and $x' \neq x_0$ and vertices from $Q$. Otherwise $y = x$.

        Repeating this argument, as $m' > 2(2\ell-4)((k-1)^2+1)+\ell$, either there is some $1 \leq \delta \leq 2((k-1)^2+1)$ and vertex $x' \neq x$, $x' \notin  X$ such that $x$ is adjacent to $q_{\delta(2\ell-4)}$ and $x'$ is the common neighbour of $q_{(\delta+1)(2\ell-4)}$ and $p_0$ or $q_{\delta(2\ell-4)}$ is adjacent to $x$ for all $0 \leq \delta \leq 2((k-1)^2+1)$. In the former case, we have found a cycle $(p_0, x, q_{\delta(2\ell-4)}, \dots, q_{(\delta+1)(2\ell-4)}, x')$ of length $2\ell$ containing $p_0$, two external vertices $x\notin X$, $x'\notin X$ and vertices from $Q$ as claimed.  On the other hand, the latter case contradicts Claim~\ref{clm:dist2l-4} as $x$ has $(k-1)^2+1$ pairs of neighbours with pairwise distance at least $\ell$. 
    \end{claimproof}

    In the following we recursively define $k-1$ distinct segments $Q_1,\dots, Q_{k-1}$ and $2k-2$ distinct external vertices $x_1, \dots, x_{k-1}$ and $x_1',\dots, x_{k-1}'$  such that for every $i\in [k-1]$ we have that $x_i\neq x_0$, $x_i'\neq x_0$ and there is a cycle of length $2\ell$ containing $p_0$, $x_i, x_i'$ and vertices from $Q_i$.

    As highlighted above, there are at least $4^kk^{k+1}+1$ segments, from Claim~\ref{clm:vbndnei}, and hence there is some segment which does not contain a neighbour of $x_0$. Let $Q_1$ be some segment which does not contain any neighbour of $x_0$, from Claim~\ref{clm-seg} there is some pair of vertices, call these $x_1, x'_1 \neq x_0$ such that there is a cycle of length $2\ell$ containing $p_0$, $x_1, x_1'$ and vertices from $Q_1$.
    Suppose now there are segments $Q_1,\dots, Q_{\delta}$ and distinct external vertices $x_1, \dots, x_{\delta}$ and $x_1',\dots, x_{\delta}'$ for some $\delta < k-1$. Let $X_{\delta} = \{x_0,x_1, \dots, x_{\delta}, x_1',\dots, x_{\delta}'\}$. From Claim~\ref{clm:vbndnei}, no vertex in $X_{\delta}$ can have neighbours in $4^kk^{k+1}$ different segments, as $|X_{\delta}| = 2\delta+1$ and there are at least $(2\delta+1) \cdot 4^kk^{k+1}+1$ segments, there is some segment which does not contain a neighbour in $X$, let $Q_{\delta+1}$ be such a segment. Again from Claim~\ref{clm-seg} there is some pair of vertices, call these $x_{\delta+1}, x'_{\delta+1} \notin X_{\delta}$ such that there is a cycle of length $2\ell$ containing $p_0$, $x_{\delta+1}, x'_{\delta+1}$ and vertices from $Q_{\delta+1}$.

    Given segments $Q_1,\dots, Q_{k-1}$ and vertices $x_1, \dots, x_{k-1}$ and $x_1',\dots, x_{k-1}'$ as described above, there are $k-1$ cycles with a common vertex $p_0$. Further the cycle $(x_0, p_0, \ldots, p_{2(\ell-1)})$ has length $2\ell$ and contains no vertex from the segments $Q_1,\dots, Q_{k-1}$, as these segments do not contain a neighbour of $x_0$. That is, this describes \forbiddenGraph. \qed
\end{proof}

The two results above 
might suggest that we obtain bounded treedepth for any class of $\sharedVertex{\ell_1,\dots,\ell_k}$-subgraph-free graphs for $\ell_1,\dots,\ell_k>4$ even number of bounded diameter $2$. But in fact, this is not the case. If we allow only two different length ($6$ and $8$) and take sufficiently many $C_6$ and $C_8$ sharing a vertex, the treedepth is unbounded. That is the inverse of \Cref{obs:canonicalUnboundedExamples} is not true.

\begin{theorem}
\label{thm:CVunbounded}
    The class of $\sharedVertex{12\times [6],12\times [8]}$-subgraph-free graphs of diameter at most~$2$ has unbounded treedepth.
\end{theorem}
\begin{proof}
    For every $n\in \mathbb{N}$ we construct a 
    graph~$G_n$ which is $\sharedVertex{12\times [6],12\times [8]}$-subgraph-free, has diameter at most~$2$ and treedepth at least $\log(n)$.
    We construct $G_n$ from the disjoint union of $P_n=(p_0,\dots,p_n)$ and a complete graph~$K_{12}$ on vertex set $$Z=\{x_{0,1}, x_{1,2},x_{2,3},x_{3,0}, y_{0,2}, y_{1,3},y_{2,4}, y_{3,5},y_{4,6},y_{5,7},y_{6,0},y_{7,1}\}$$ 
    by adding edges as follows. For all $i,j\in [0,3]$, $j\equiv i+1\mod 4$ and $k\in [0,n]$ we add the edge $x_{i,j}p_k$ if and only if $k\equiv i \mod 4$ or $k\equiv j \mod 4$. Similarly, for all $i,j\in [0,7]$, $j\equiv i+2\mod 8$ and $j\in [0,n]$ we add the edge $y_{i,j}p_k$ if and only if $k\equiv i \mod 8$ or $k\equiv j \mod 8$. For an illustration of the construction see \Cref{fig:CVunbounded}.\\

    We first argue that $G_n$ has diameter $2$. Trivially, any two vertices in $Z$ are of distance $1$ of each other. Furthermore, the distance between any $z\in Z$ and $p_k$ is at most~$2$ as $p_k$ is adjacent to some vertex in $Z$ and $Z$ is a clique. Therefore, consider $p_k,p_\ell$ with $k<\ell$. 
    
    First assume that $k\equiv \ell \mod 4$.  Set $i,j\in [0,3]$, $j\equiv i+1\mod 4$ such that $k\equiv \ell \equiv i \mod 4$ and observe that both $p_k$ and $p_\ell$ are adjacent to $x_{i,j}$ and hence are of distance $2$.
    
    Next assume that $k+1\equiv \ell \mod 4$. We set  $i,j\in [0,3]$ such that $k\equiv i\mod 4$ and $\ell \equiv j \mod 4$ and observe that $x_{i,j}$ is a vertex in $G_n$. Hence, $p_k$ and $p_\ell$ have distance at most~$2$ as $p_k$ and $p_\ell$ are adjacent to $x_{i,j}$.

    Next assume that $k+3 \equiv \ell \mod 4$. We define $i,j\in [0,3]$ such that $\ell\equiv i\mod 4$ and $k \equiv j \mod 4$ and observe again that $x_{i,j}$ is a vertex in $G_n$. Hence, $p_k$ and $p_\ell$ have distance at most~$2$ as both $p_k$ and $p_\ell$ are adjacent to $x_{i,j}$.

     Finally, assume that $k+2 \equiv \ell \mod 4$. In this case either $k+2\equiv \ell \mod 8$ or $\ell+6 \equiv \ell \mod 8$. In the former case, let $i,j\in [0,7]$ such that $k\equiv i\mod 8$ and $\ell\equiv j\mod 8$ and observe that $y_{i,j}$ is a vertex in $G_n$. In the latter case we choose $i,j\in [0,7]$ such that $\ell\equiv i\mod 8$ and $k\equiv j\mod 8$ and remark that $y_{i,j}$ is a vertex in $G_n$. We conclude that $p_k$ and $p_\ell$ are of distance at most~$2$ as in both cases $p_k$ and $p_\ell$ are adjacent to $y_{i,j}$.\\

    We now argue that $G_n$ is $\sharedVertex{12\times [6],12\times [8]}$-subgraph-free. First observe, that for any $i,j\in [0,3]$, $j\equiv i\mod 4$ any $C_8$ in $G_n$ which contains $x_{i,j}$ must contain a second vertex from $Z$. Indeed, if this is not the case then there is a $C_8$ which comprise of $x_{i,j}$ and a subpath of $P$ of length $7$ with start and end vertex adjacent to $x_{i,j}$ which is impossible by construction. Additionally, for any $i,j\in [0,7]$, $j\equiv i \mod 8$ any $C_6$ in $G_n$ which contain $y_{i,j}$ must contain a second vertex from $Z$. Indeed, this follows from $y_{i,j}$ not being adjacent to any pair of vertices of distance $4$ on $P$. 

    Towards a contradiction, assume $G_n$ contains $\sharedVertex{12\times [6],12\times [8]}$ as a subgraph. We distinguish three cases. First assume that the shared vertex in the subgraph $\sharedVertex{12\times [6],12\times [8]}$ is $x_{i,j}$ for some $i,j\in [0,3]$. In this case, each of the $12$ $C_8$'s contained in $\sharedVertex{12\times [6],12\times [8]}$ must contain a second vertex in $Z$ and these additional vertices have to be pairwise different for different $C_8$'s. A contradiction as $Z$ contains only $12$ vertices in total. We obtain a contradiction in the same way considering the shared vertex to be $y_{i,j}$ for some $i,j\in [0,7]$ with the requirement of the $C_6$'s contained in $\sharedVertex{12\times [6],12\times [8]}$. Finally, in case the shared vertex in the subgraph $\sharedVertex{12\times [6],12\times [8]}$ is on $P$, each of the $24$ cycles of $\sharedVertex{12\times [6],12\times [8]}$ has to contain a vertex not on $P$ and all of them have to pairwise different, a contradiction.\\

    Finally, $G_n$ contains $P_n$ and hence has treedepth at least $\log(n)$ by \Cref{fact:pathLogtd}. We conclude that the class of $\sharedVertex{12\times [6],12\times [8]}$-subgraph-free graph of diameter at most $2$ has unbounded treedepth. \qed
\begin{figure}
        \centering
        \includegraphics{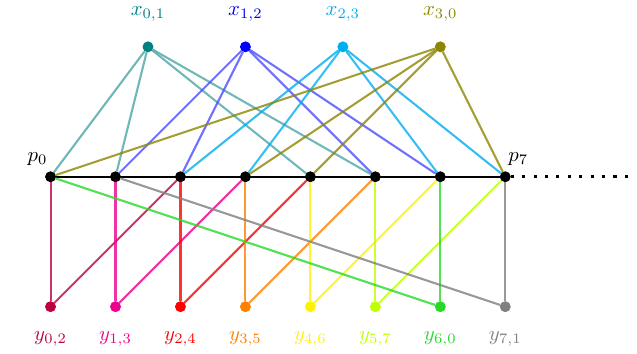}
        \caption{Construction of the graph~$G_{n}$ from the proof of \Cref{thm:CVunbounded}. The edges within the clique $\{x_{0,1}, x_{1,2},x_{2,3},x_{3,0}, y_{0,2}, y_{1,3},y_{2,4}, y_{3,5},y_{4,6},y_{5,7},y_{6,0},y_{7,1}\}$ are omitted.}
        \label{fig:CVunbounded}
\end{figure}
\end{proof}

\begin{theorem}
\label{thm:CEtwoLength}
    For any numbers $\ell_1,\ell_2>2$ the class of all $\sharedEdge{2\ell_1,2\ell_2}$-subgraph-free graphs of diameter at most~$2$ has bounded treedepth.
\end{theorem}
\begin{proof}
     Let $\ell_1,\ell_2>2$ be integers and assume $G$ is a $\sharedEdge{2\ell_1,2\ell_2}$-subgraph-free graph of diameter at most~$2$ and  treedepth at least $c(\ell_1+\ell_2,\ell_1+\ell_2,3\ell_1+2\ell_2-13)$. Note that $G$ cannot contain a large complete bipartite subgraph as $K_{\ell_1+\ell_2,\ell_1+\ell_2}$ contains $\sharedEdge{2\ell_1,2\ell_2}$ as a subgraph. Hence, by \Cref{cor:longInducedPath}, $G$ must contain $P_\ell=(p_0,\dots,p_\ell)$ where $\ell=6\ell_1+4\ell_2-13$ as an induced subgraph.

     First observe that $G$ having diameter at most~$2$ implies that for any two vertices $p,q$ of distance at least 3 on $P$ there must be a vertex $x$ not on $P$ adjacent to both $p$ and $q$. We also remark that a vertex $x$ not on $P$ being adjacent to two vertices $p,q$ of distance $k-2$ on $P$ forms a $C_k$ with the part of $P$ from $p$ to $q$.
     \begin{mclaim}\label{claim:impliedAdjCE}
        Let $i\in \{1,2\}$, $j\in [\ell-2\ell_1-2\ell_2+5]$ and $x$ a vertex not on $P$.
         If  $x$ is adjacent to $p_j$ and $p_{j+2\ell_i-2}$ then 
         \begin{enumerate}
             \item $x$ is adjacent to $p_{j+2\ell_i-3}$ and to $p_{j+2\ell_i+2\ell_{3-i}-5}$, and
             \item $x$ is not adjacent to $p_{j+2\ell_i+2\ell_{3-i}-4}$ and to $p_{j-1}$ if $j>1$.
         \end{enumerate}
     \end{mclaim}
     \begin{claimproof}
         Assume $x$ is adjacent to $p_j$ and $p_{j+2\ell_i-2}$. Let $y$ (possibly equal to $x$) be a vertex not on $P$ which is adjacent to $p_{j+2\ell_i-3}$ and $p_{j+2\ell_i+2\ell_{3-i}-5}$ ($y$ must exist as $p_{j+2\ell_i-3}$ and $p_{j+2\ell_i+2\ell_{3-i}-5}$ are of distance $(j+2\ell_i+2\ell_{3-i}-5) - (j+2\ell_i-3)=2\ell_{3-i}-2\geq 3$). If $x\not=y$ we get the copy of $H$ consisting of cycles $(x,p_j,\dots, p_{j+2\ell_i-2},x)$ and $(y,p_{j+2\ell_i-3},\dots, p_{j+2\ell_i+2\ell_{3-1}-5},y)$ with shared edge $\{p_{j+2\ell_i-3},p_{j+2\ell_i-2}\}$. Hence, $x=y$ and the first statement follows. 
         
         To show the second statement simply observe, that $x$ cannot be adjacent to $p_{n+m-4}$ as otherwise we get the copy of $H$ with cycles $(x,p_j,\dots, p_{j+2\ell_i-2},x)$ and $(x,p_{j+2\ell_i-2},\dots, p_{j+2\ell_i+2\ell_{3-i}-4})$ with shared edge $\{x,p_{j+2\ell_i-2}\}$. The assumption that $x$ is adjacent to $p_{j-1}$ (in case $j>1$) yields a contradiction with a symmetric argument.
     \end{claimproof}
     
     Let $x$ be a vertex not on $P$ which is adjacent to $p_0$ and $p_{2\ell_1-2}$ ($x$ must exist as the distance of $p_0$ and $p_{2\ell_1-2}$ is $2\ell_1-2\geq 3$). By \Cref{claim:impliedAdjCE}, this implies that $p_{2\ell_1-3}$ and $p_{2\ell_1+2\ell_2-5}$ are also adjacent to $x$ while $p_{2\ell_1+2\ell_2-4}$ is not adjacent to $x$. Applying \Cref{claim:impliedAdjCE} again (for $p_{2\ell_1-3}$ and $p_{2\ell_1+2\ell_2-5}$), implies that $p_{2\ell_1+2\ell_2-6}$ and $p_{4\ell_1+2\ell_2-8}$ are adjacent to $x$. Finally, applying \Cref{claim:impliedAdjCE} ($p_{2\ell_1+2\ell_2-6}$ and $p_{4\ell_1+2\ell_2-8}$) implies that $p_{2\ell_1+2\ell_2-7}$ is not adjacent to $x$.

     We now let $y$ (possibly equal to $x$) be a vertex not on $P$ which is adjacent to $p_{1}$ and $p_{2\ell_1-1}$ (which exists as $p_{1}$ and $p_{2\ell_1-1}$ are of distance $2\ell_1-2\geq 3$ on $P$). Applying \Cref{claim:impliedAdjCE} four times sequentially (similar as before) yields that $p_{2\ell_1-2}$, $p_{2\ell_1+2\ell_2-4}$, $p_{2\ell_1+2\ell_2-5}$, $p_{4\ell_1+2\ell_2-7}$, $p_{4\ell_1+2\ell_2-8}$ and $p_{4\ell_1+4\ell_2-10}$ are adjacent to $y$ while $p_{2\ell_1+2\ell_2-9}$ is not adjacent to $y$. As $y$ is adjacent to $p_{2\ell_1+2\ell_2-4}$ while $x$ is not adjacent to $p_{2\ell_1+2\ell_2-4}$ we know that $x\not=y$.

     Finally, we let $z$ be a vertex not on $P$ which is adjacent to $p_{2\ell_1+2\ell_2-7}$ and $p_{4\ell_1+2\ell_2-9}$ (such a vertex exists as $p_{2\ell_1+2\ell_2-7}$ and $p_{4\ell_1+2\ell_2-9}$ are of distance $2\ell_1-2\geq 3$ on $P$). Note that $z\not=x$ as $x$ is not adjacent to $p_{2\ell_1+2\ell_2-7}$ and $z\not=y$ as $y$ is not adjacent to $p_{4\ell_1+2\ell_2-9}$. Furthermore, $C_1=(z, p_{2\ell_1+2\ell_2-7},\dots, p_{4\ell_1+2\ell_2-9}, z)$ is a cycle of length $2\ell_1$ in $G$. Additionally, we obtain a second cycle  $C_2=(x, p_{2\ell_1-3}, p_{2\ell_1-2},y,p_{2\ell_1-1},\dots,p_{2\ell_1+2\ell_2-6},x)$ of length $2\ell_2$. As $C_1$ and $C_2$ precisely share the edge $p_{2\ell_1+2\ell_2-7}p_{2\ell_1+2\ell_2-6}$ we obtain a copy of $\sharedEdge{2\ell_1,2\ell_2}$. See \Cref{fig:CEtwoLength} for illustration. Since this contradicts the assumption that $G$ is $\sharedEdge{2\ell_1,2\ell_2}$-subgraph-free, $G$ has treedepth less than $c(\ell_1+\ell_2,\ell_1+\ell_2,6\ell_1+4\ell_2-13)$. \qed
    \begin{figure}
        \centering
        \includegraphics[scale=0.85]{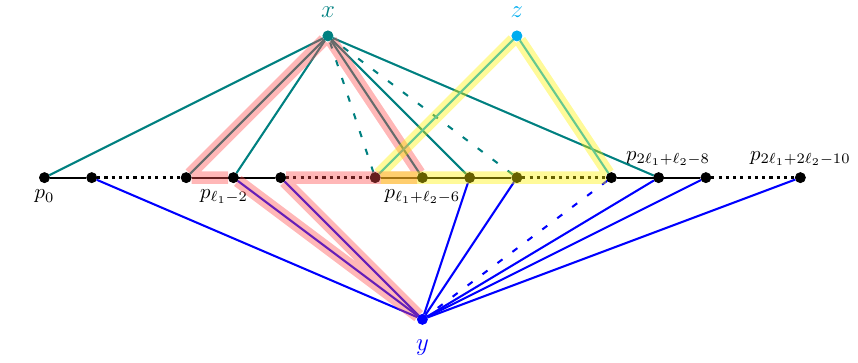}
        \caption{The construction from the proof of \Cref{thm:CEtwoLength} of a subgraph $\sharedEdge{\ell_1,\ell_2}$ of $G$ under the assumption that $G$ is $\sharedEdge{\ell_1,\ell_2}$-free, has diameter at most~$2$ and large treedepth.}
        \label{fig:CEtwoLength}
    \end{figure}
\end{proof}

\begin{theorem}
\label{thm:edgekCl}
    For any $\ell \geq 3$, there exists some $k \leq 2 (2\ell -3)$ such that the class of $\sharedEdge{k*[2\ell]}$-subgraph-free graphs of diameter at most~$2$ has unbounded treedepth.
\end{theorem}
\begin{proof}
    We set  $k = 2(2\ell-3)$. For all $n\in \mathbb{N}$ we construct a $\sharedEdge{k\times [2\ell]}$-subgraph-free graph~$G_n$ with diameter at most~$2$ and treedepth at least $\log(n)$ by taking $2\ell-3$ vertices $x_0,\dots, x_{2\ell-4}$  and a path $P=(p_0,\dots,p_n)$. Consider the binary string  $R=1(10)^{\ell-2}$ consisting of $1$ followed by $\ell-2$ repetitions of the string $10$.  
    We add the edge $p_ix_j$ if the   $i-j \mod (2\ell-3)$th bit of $R$ is $1$. See \Cref{fig:CEunbounded}. By \Cref{fact:pathLogtd}, $\td(G_n)\geq \log(n)$. 
    \begin{figure}[b]
        \vspace*{-4mm}
        \centering
        \includegraphics[scale=0.85]{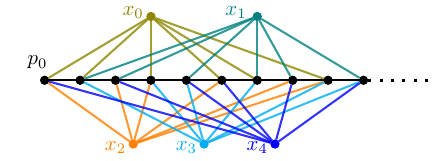}
        \caption{The construction from \Cref{thm:edgekCl} for $\ell=4$ for which the pattern is $R=11010$.}
        \label{fig:CEunbounded}
    \end{figure}

    Next we argue that $G_n$ has diameter $2$ by showing that any pair of vertices has distance at most $2$. Each $x_i$ is either adjacent to $p_{j-1}$ or $p_{j}$ for every $j\in [n]$. 
    Next consider $p_i$ and $p_j$ and observe that $p_i$ is both adjacent to $x_{i\mod (2\ell-3)}$ and $x_{i+1\mod (2\ell-3)}$. Furthermore, $p_j$ must be either adjacent to $x_{i\mod (2\ell-3)}$ or $x_{i+1\mod (2\ell-3)}$. 
    Finally, consider $x_i$ and $x_j$. Note that $x_i$ is adjacent to both $p_i$ and $p_{i+1}$ and $x_j$ must be adjacent to either $p_i$ or $p_{i+1}$. 
    
    Now, for a contradiction, assume $G_n$ contains $\sharedEdge{k\times [2\ell]}$ as a subgraph and $x\in V(G_n)$ is the vertex common to all cycles.
    As each vertex on $P$ has degree $\ell-1$, we have $x\in \{x_0,\dots, x_{2\ell-4}\}$.  
    Any $C_{2\ell}$ in $G[\{x\}\cup V(P)]$ is of the form $(x, p_i,\dots,p_{2\ell-2},x)$. As $R$ is of length~$2\ell-3$,  there is either one or two indices $j\in \{i,\dots, 2\ell-3\}$ such that $x$ is adjacent to both $p_j$ and $p_{j+1}$. If there is only one such index~$j$, $x$ cannot be adjacent to both $p_i$ and $p_{i+2\ell-2}$. Hence, there are two. But then there is only one index $j\in \{2\ell-2,\dots,4\ell-5\}$ such that $x$ is adjacent to both $p_j$ and $p_{j+1}$. Hence, $x$ cannot be adjacent to $p_{i+4\ell-4}$ and thus not both $(x,p_i,\dots, p_{i+2\ell-2},x)$ and $(x,p_{i+2\ell-2},\dots, p_{i+4\ell-4},x)$ are subgraphs of $G_n$. Therefore, at least every second of the $k$ cycles  of the subgraph isomorphic to $\sharedEdge{k\times [2\ell]}$ contains some $y\in \{x_0,\dots,x_{2\ell-3}\}\setminus\{x\}$. As $x$ is not adjacent to any vertex in $\{x_0,\dots,x_{2\ell-3}\}\setminus\{x\}$ we get that each of the vertices $y\in \{x_0,\dots,x_{2\ell-3}\}\setminus\{x\}$ can be in at most one cycle, a contradiction as we may only obtain $2(2\ell-4)+1<k$ cycles.
     \qed
  \end{proof}

\begin{theorem}
\label{lem:samecyc}
    For any integers $\ell_1,\ell_2\geq 2$, $\ell\geq 4$  the class of $\{\sharedVertex{4\ell_1,4\ell_2}, \sharedEdge{4\ell_1,4\ell_2}\}$-subgraph-free graphs of diameter at most~$3$ and the class of $\{\sharedVertex{2\ell,2\ell}, \sharedEdge{2\ell,2\ell}\}$-subgraph-free graphs of diameter at most~$3$ have unbounded treedepth.
\end{theorem}
\begin{proof}
    We use two different constructions for the two classes of graphs.
    First consider any pair of integers $\ell_1,\ell_2\geq 2$.
    For every $n \in \mathbb{N}$  we construct in the following a $\{\sharedVertex{4\ell_1,4\ell_2}, \sharedEdge{4\ell_1,4\ell_2}\}$-subgraph-free graph $G_n$. We construct $G_n$ from the disjoint union of a path $P = (p_0, \cdots, p_n)$ and two isolated vertices $x$ and $y$.  We add an edge $xp_i$ whenever the $i\mod 4$ is either $0$ or $1$ and we add an edge $yp_i$ otherwise. By \cref{fact:pathLogtd} $\td(G_n)\geq \log(n)$ as $G_n$ contains $P_n$.

    To see that $G_n$ has diameter at most $3$ observe that $x$ and $y$ have distance at most $3$ as $x$ is adjacent to $p_1$ while $y$ is adjacent to $p_2$. Furthermore, $x$ ($y$ resp.) has distance at most $3$ to every vertex $p_i$ of the path as $x$ ($y$ resp.) is adjacent to either $p_i$ or $p_{i+1}$ or $p_{i+2}$. Finally, consider any two vertices $p_i$ and $p_j$ and assume without loss of generality that $p_i$ is adjacent to $x$. Now either $p_j$ is also adjacent to $x$ or $p_j$ has a neighbour which is adjacent to $x$
    by construction. Hence, any pair of vertices on $P$ has distance at most $3$.  

    Finally, we argue that $G_n$ is $\{\sharedVertex{4\ell_1,4\ell_2}, \sharedEdge{4\ell_1,4\ell_2}\}$-subgraph-free. Observe that if $x$ (or $y$ resp.) is adjacent to $p_i$ then $x$ (or $y$ resp.) is not adjacent to $p_{i+4m-2}$ for any $m\geq 1$. Hence, $G_n[\{x\}\cup V(P)]$ and symmetrically $G_n[\{y\}\cup V(P)]$ are $\{C_{4\ell_1}, C_{4\ell_2}\}$-subgraph-free. We conclude that $G_n$ contains neither $\sharedVertex{4\ell_1,4\ell_2}$ nor $ \sharedEdge{4\ell_1,4\ell_2}$ as a subgraph.\\

    \begin{figure}[b]
        \vspace*{-7mm}
        \centering
        \includegraphics[scale=1]{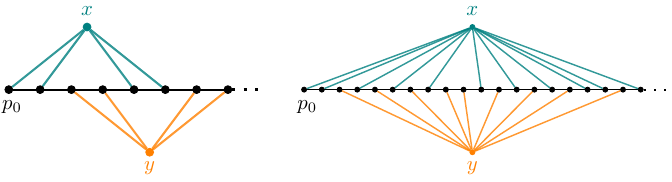}
        \caption{The constructions from \Cref{lem:samecyc}. To the left the first construction with pattern $1100$ and to the left the second construction for $\ell=5$ for which the pattern is $R=1101010100101010$.}
        \label{fig:diam3unbounded}
        \vspace*{-3mm}
    \end{figure}

    Next consider any $\ell\geq 2$.
    For every $n \in \mathbb{N}$ we construct a $\{\sharedVertex{2\ell,2\ell}, \sharedEdge{2\ell,2\ell}\}$-subgraph-free graph $\overline{G}_n$. We construct $\overline{G}_n$ from the disjoint union of a path $P = (p_0, \cdots, p_n)$ and two isolated vertices $x$ and $y$. Let $R\in \{0,1\}^\ast$ be the binary word $R=11 (01)^{\ell-2} 00 (10)^{\ell-2}$ consisting of the word $11$ followed by $\ell-2$ repetitions of the word $01$, the word $00$ and $\ell-2$ repetitions of the word $10$. We add an edge $xp_i$ whenever the $i\mod 4\ell-4$th bit of $R$ is $1$ and we add an edge $yp_i$ whenever the $i\mod 4\ell-4$th bit of $R$ is $0$. By \cref{fact:pathLogtd} $\td(\overline{G}_n)\geq \log(n)$ as $\overline{G}_n$ contains $P_n$.

    To see that $\overline{G}_n$ has diameter at most $3$ observe that $x$ and $y$ have distance at most $3$ as $x$ is adjacent to $p_1$ while $y$ is adjacent to $p_2$. Furthermore, $x$ ($y$ resp.) has distance at most $3$ to every vertex $p_i$ of the path as $x$ ($y$ resp.) is adjacent to either $p_i$ or $p_{i+1}$ or $p_{i+2}$. Finally, consider any two vertices $p_i$ and $p_j$ and assume without loss of generality that $p_i$ is adjacent to $x$. Now either $p_j$ is also adjacent to $x$ or $p_j$ has a neighbour which is adjacent to $x$
    by construction. Hence, any pair of vertices on $P$ has distance at most $3$.  

    Finally, we argue that $\overline{G}_n$ is $\{\sharedVertex{2\ell,2\ell}, \sharedEdge{2\ell,2\ell}\}$-subgraph-free. To see this, we observe that $RR$ has the property that the $i$th bit of $RR$ is $1$ if and only if the $i+2\ell-2$th bit of $RR$ is $0$ for every $i\in [4\ell-4]$. This directly implies that if $x$ (or $y$ resp.) is adjacent to $p_i$ then $x$ (or $y$ resp.) is not adjacent to $p_{i+2\ell-2}$. Hence, both $\overline{G}_n[\{x\}\cup V(P)]$ and $\overline{G}_n[\{y\}\cup V(P)]$ are $C_{2\ell}$-subgraph-free. We conclude that $\overline{G}_n$ can contain neither $\sharedVertex{2\ell,2\ell}$ nor $\sharedEdge{2\ell,2\ell}$ as a subgraph.
    \qed
\end{proof}

\section{Conclusions}\label{s-con}

We showed that bounding the diameter has a significant impact on the boundedness of width of a graph class, in particular for the subgraph relation and treedepth. We pose some open problems.
First, recall that Table~\ref{t-table} still contains three open cases: which classes of $F$-subgraph-free graphs have the diameter-width property for clique-width, treewidth or pathwidth? Towards solving these questions, we can show there are graphs $F$ such that $F$-subgraph-free graphs of diameter~$2$ have bounded pathwidth but unbounded treedepth; and also that 
there are graphs $F$ such that $F$-subgraph-free graphs of diameter~$2$ have bounded clique-width but unbounded treewidth 
(see \Cref{sec:subgraph}).

Second, we note that Demaine and Hajiaghayi~\cite{DH04} proved the diameter-treewidth property with even a linear (and optimal) diameter bound. We leave it as an open problem whether our diameter bound in Theorem~\ref{thm:classificationDiam5} can be optimized.

Third, recall that we classified boundedness of treedepth of $F$-subgraph-free graphs of diameter at most~$d$ for every constant $d\geq 4$. Completing the classification for $d=2$ is challenging due to the cases where $F$ is of $V$-type $C^V_x$ or $E$-type $C^E_x$. For $d=3$, we must consider $F=C_{2r}$ for $r\geq 2$. We showed the treedepth is unbounded for $r\in \{2,3\}$ and $d=2$, but bounded for $r=4$ and $d=3$. Our proof for the latter is an involved case analysis 
(see \Cref{sec:unicyclic}),
which seems not easy to extend to larger values of $r$. But by use of a computer~\cite{Ea25} we can also prove boundedness for 
$r\in \{5,\ldots,12\}$ and $d=3$. We therefore conjecture: 
\begin{conjecture}\label{con:diam3cyc}
    For every $r\geq 4$, the class of $C_{2r}$-subgraph-free graphs of diameter at most~$3$ has bounded treedepth.
\end{conjecture}
If a stronger version of Conjecture~\ref{con:diam3cyc} involving all graphs with one cycle
(analogous to Theorem~\ref{thm:diam2-unicyclic}) is true, just like the following conjecture (which is supported by all our results so far), then we will obtain a complete classification for $d=3$.

\begin{conjecture}\label{con:diam3TwoCyc}
    For a graph~$\forbiddenGraph$ with at least two cycles, the class $\mathcal{C}$ of $\forbiddenGraph$-subgraph-free graphs of diameter at most~$3$ has unbounded treedepth.
\end{conjecture}

\bibliography{bib}

\begin{thebibliography}{10}
\providecommand{\url}[1]{\texttt{#1}}
\providecommand{\urlprefix}{URL }
\providecommand{\doi}[1]{https://doi.org/#1}

\bibitem{BRST91}
Bienstock, D., Robertson, N., Seymour, P.D., Thomas, R.: Quickly excluding a
  forest. Journal of Combinatorial Theory, Series {B}  \textbf{52}(2),
  274--283 (1991)

\bibitem{BGHK95}
Bodlaender, H.L., Gilbert, J.R., Hafsteinsson, H., Kloks, T.: Approximating
  treewidth, pathwidth, frontsize, and shortest elimination tree. Journal of
  Algorithms  \textbf{18}(2),  238--255 (1995)

\bibitem{BL02}
Boliac, R., Lozin, V.V.: On the clique-width of graphs in hereditary classes.
  Proc. ISAAC 2002, Lecture Notes in Computer Science  \textbf{2518},  44--54
  (2002)

\bibitem{Carter89}
Carter, R.W.: Simple Groups of Lie Type. Wiley Classics Library, Wiley (1989)

\bibitem{Ch15}
Chuzhoy, J.: Improved bounds for the flat wall theorem. Proc. SODA 2015 pp.
  256--275 (2015)

\bibitem{CR05}
Corneil, D.G., Rotics, U.: On the relationship between clique-width and
  treewidth. SIAM Journal on Computing  \textbf{34}(4),  825--847 (2005)

\bibitem{Co90}
Courcelle, B.: The monadic second-order logic of graphs. {I.} {R}ecognizable
  sets of finite graphs. Information and Computation  \textbf{85}(1),  12--75
  (1990)

\bibitem{CMR00}
Courcelle, B., Makowsky, J.A., Rotics, U.: Linear time solvable optimization
  problems on graphs of bounded clique-width. Theory of Computing Systems
  \textbf{33}(2),  125--150 (2000)

\bibitem{CO00}
Courcelle, B., Olariu, S.: Upper bounds to the clique width of graphs. Discrete
  Applied Mathematics  \textbf{101}(1),  77--114 (2000)

\bibitem{DJP19}
Dabrowski, K.K., Johnson, M., Paulusma, D.: Clique-width for hereditary graph
  classes. London Mathematical Society Lecture Note Series  \textbf{456},
  1--56 (2019)

\bibitem{DP16}
Dabrowski, K.K., Paulusma, D.: Clique-width of graph classes defined by two
  forbidden induced subgraphs. The Computer Journal  \textbf{59}(5),  650--666
  (2016)

\bibitem{DH04}
Demaine, E.D., Hajiaghayi, M.: Equivalence of local treewidth and linear local
  treewidth and its algorithmic applications. Proc. SODA 2004 pp. 840--849
  (2004)

\bibitem{DPR22}
D\k{e}bski, M., Piecyk, M., Rz\k{a}\.{z}ewski, P.: Faster 3-coloring of
  small-diameter graphs. {SIAM} Journal on Discrete Mathematics
  \textbf{36}(3),  2205--2224 (2022)

\bibitem{DEJMW20}
Dujmovi\'{c}, V., Eppstein, D., Joret, G., Morin, P., Wood, D.R.: Minor-closed
  graph classes with bounded layered pathwidth. {SIAM} Journal on Discrete
  Mathematics  \textbf{34}(3),  1693--1709 (2020)

\bibitem{Ea25}
Eagling-Vose, T.:  (2025), \url{https://github.com/tevose/treedepth-diameter-3}

\bibitem{Ep00}
Eppstein, D.: Diameter and treewidth in minor-closed graph families.
  Algorithmica  \textbf{27}(3),  275--291 (2000)

\bibitem{ERTS66}
Erd{\H{o}}s, P., R{\'e}nyi, A., Sós, V.T.: On a problem of graph theory.
  Studia Scientiarum Mathematicarum Hungarica  \textbf{1},  215--235 (1966)

\bibitem{GajarskyHlineny12}
Gajarsk{\'{y}}, J., Hlinen{\'{y}}, P.: Faster deciding {MSO} properties of
  trees of fixed height, and some consequences. Proc. FSTTCS 2012, Leibniz
  International Proceedings in Informatics (LIPIcs)  \textbf{18},  112--123
  (2012)

\bibitem{GALVIN19827}
Galvin, F., Rival, I., Sands, B.: A {Ramsey-type} theorem for traceable graphs.
  Journal of Combinatorial Theory, Series B  \textbf{33}(1),  7--16 (1982)

\bibitem{GanianOrdyniak18}
Ganian, R., Ordyniak, S.: The complexity landscape of decompositional
  parameters for {ILP}. Artificial Intelligence  \textbf{257},  61--71 (2018)

\bibitem{GJW16}
Gutin, G., Jones, M., Wahlstr\"{o}m, M.: The {Mixed Chinese Postman Problem}
  parameterized by pathwidth and treedepth. {SIAM} Journal on Discrete
  Mathematics  \textbf{30}(4),  2177--2205 (2016)

\bibitem{Hi23}
Hickingbotham, R.: Induced subgraphs and path decompositions. Electronic
  Journal of Combinatorics  \textbf{30}(2),  P2.37 (2023)

\bibitem{IwataOgasawaraOhsaka18}
Iwata, Y., Ogasawara, T., Ohsaka, N.: On the power of tree-depth for fully
  polynomial {FPT} algorithms. Proc. STACS 2018, Leibniz International
  Proceedings in Informatics (LIPIcs)  \textbf{96},  41:1--41:14 (2018)

\bibitem{DBLP:journals/dam/KaminskiLM09}
Kamiński, M., Lozin, V.V., Milanič, M.: Recent developments on graphs of
  bounded clique-width. Discrete Applied Mathematics  \textbf{157}(12),
  2747--2761 (2009)

\bibitem{KS23}
Klimosov{\'{a}}, T., Sahlot, V.: $3$-coloring ${C}_4$ or ${C}_3$-free diameter
  two graphs. Proc. WADS 2023, Lecture Notes in Computer Science
  \textbf{14079},  547--560 (2023)

\bibitem{KLO18}
Kouteck\'{y}, M., Levin, A., Onn, S.: {A Parameterized Strongly Polynomial
  Algorithm for Block Structured Integer Programs}. Proc. ICALP 2018, Leibniz
  International Proceedings in Informatics (LIPIcs)  \textbf{107},  85:1--85:14
  (2018)

\bibitem{LR06}
Lozin, V.V., Rautenbach, D.: The tree- and clique-width of bipartite graphs in
  special classes. The Australasian Journal of Combinatorics  \textbf{34},
  57--67 (2006)

\bibitem{LR22}
Lozin, V.V., Razgon, I.: Tree-width dichotomy. European Journal of
  Combinatorics  \textbf{103},  103517 (2022)

\bibitem{MP15}
Martin, B., Paulusma, D.: The computational complexity of {D}isconnected {C}ut
  and $2{K}_2$-{P}artition. Journal of Combinatorial Theory, Series {B}
  \textbf{111},  17--37 (2015)

\bibitem{MPS22}
Martin, B., Paulusma, D., Smith, S.: Colouring graphs of bounded diameter in
  the absence of small cycles. Discrete Applied Mathematics  \textbf{314},
  150--161 (2022)

\bibitem{Marx10b}
Marx, D.: Can you beat treewidth? Theory of Computing  \textbf{6}(5),  85--112
  (2010)

\bibitem{MS16}
Mertzios, G.B., Spirakis, P.G.: Algorithms and almost tight results for
  $3$-{C}olorability of small diameter graphs. Algorithmica  \textbf{74}(1),
  385--414 (2016)

\bibitem{NM12}
Ne\v{s}et\v{r}il, J., de~Mendez, P.O.: Sparsity - {G}raphs, {S}tructures, and
  {A}lgorithms, Algorithms and Combinatorics, vol.~28. Springer (2012)

\bibitem{RS84}
Robertson, N., Seymour, P.D.: Graph minors. {III}. {P}lanar tree-width. Journal
  of Combinatorial Theory, Series {B}  \textbf{36}(1),  49--64 (1984)

\bibitem{RS86}
Robertson, N., Seymour, P.D.: Graph minors. {V}. {E}xcluding a planar graph.
  Journal of Combinatorial Theory, Series {B}  \textbf{41}(1),  92--114 (1986)

\bibitem{Tits60}
Tits, J.: Les groupes simples de {Suzuki} et de {Ree}. S{\'e}minaire Bourbaki
  \textbf{6},  65--82 (1961), talk no. 210

\end{thebibliography}
\end{document}